\ifx\compilefullpaper\undefined  
\documentclass[10pt, twocolumn, aps, nofootinbib, longbibliography, nobibnotes]{revtex4-1} 
\pdfoutput=1	

\usepackage{amsfonts, amssymb, amsmath, amsthm}
\usepackage{latexsym}
\usepackage[tracking=smallcaps]{microtype}	
\usepackage{url}
\usepackage{color}
\definecolor{DarkGray}{rgb}{0.1,0.1,0.5}
\usepackage[colorlinks=true,breaklinks, linkcolor=DarkGray,citecolor=DarkGray,urlcolor=DarkGray]{hyperref}	

\usepackage{graphicx}
\usepackage[tight, TABBOTCAP]{subfigure}
\usepackage{multirow}

\def\place #1#2#3{\mspace{#2}\makebox[0pt]{\raisebox{#3}{#1}}\mspace{-#2}}	

\newcommand{\ket}[1]{{|#1\rangle}}
\newcommand{\braket}[2]{{\langle#1|#2\rangle}}

\newcommand{\abs}[1]{{\lvert #1\rvert}}	

\def\cP {{\mathcal P}}

\newcommand{\identity}{\ensuremath{\boldsymbol{1}}} 

\newcounter{sprows}

\newlength{\spheight}
\newlength{\spraise}

\newlength{\commentslength}

\newcommand{\rem}[1]{}

\newtheorem{theorem}{Theorem}

\newtheorem{claim}[theorem]{Claim}

\newtheorem{proposition}[theorem]{Proposition}

\newtheorem{definition}[theorem]{Definition}

\newfont{\subsubsecfnt}{ptmri8t at 11pt}
\renewcommand{\subparagraph}[1]{\smallskip{\subsubsecfnt #1.}}

\newcommand{\eqnref}[1]{\hyperref[#1]{{(\ref*{#1})}}}
\newcommand{\thmref}[1]{\hyperref[#1]{{Theorem~\ref*{#1}}}}
\newcommand{\lemref}[1]{\hyperref[#1]{{Lemma~\ref*{#1}}}}
\newcommand{\corref}[1]{\hyperref[#1]{{Corollary~\ref*{#1}}}}
\newcommand{\defref}[1]{\hyperref[#1]{{Definition~\ref*{#1}}}}
\newcommand{\secref}[1]{\hyperref[#1]{{Sec.~\ref*{#1}}}}
\newcommand{\figref}[1]{\hyperref[#1]{{Fig.~\ref*{#1}}}}  
\newcommand{\tabref}[1]{\hyperref[#1]{{Table~\ref*{#1}}}}
\newcommand{\remref}[1]{\hyperref[#1]{{Remark~\ref*{#1}}}}
\newcommand{\appref}[1]{\hyperref[#1]{{Appendix~\ref*{#1}}}}
\newcommand{\claimref}[1]{\hyperref[#1]{{Claim~\ref*{#1}}}}
\newcommand{\factref}[1]{\hyperref[#1]{{Fact~\ref*{#1}}}}
\newcommand{\propref}[1]{\hyperref[#1]{{Proposition~\ref*{#1}}}}
\newcommand{\exampleref}[1]{\hyperref[#1]{{Example~\ref*{#1}}}}
\newcommand{\conjref}[1]{\hyperref[#1]{{Conjecture~\ref*{#1}}}}

\allowdisplaybreaks[1]

\def\COLOR{}
\ifdefined\COLOR

\else

\fi

\definecolor{Cayenne}{rgb}{0.5,0,0}
\definecolor{Midnight}{rgb}{0,0,0.5}
\definecolor{Plum}{rgb}{0.5,0,0.5}
\definecolor{Teal}{rgb}{0,0.5,0.5}
\definecolor{Clover}{rgb}{0,0.5,0}
\definecolor{Maroon}{rgb}{0.5,0,0.25}
\definecolor{Ocean}{rgb}{0,0.25,0.5}
\definecolor{Tangerine}{rgb}{1,0.5,0}
\definecolor{Strawberry}{rgb}{1,0,0.5}
\definecolor{Fern}{rgb}{0.25,0.5,0}
\definecolor{Aqua}{rgb}{0,0.5,1}
\definecolor{Moss}{rgb}{0,0.5,0.25}
\definecolor{Mocha}{rgb}{0.5,0.25,0}
\definecolor{Lemon}{rgb}{1,1,0}
\definecolor{Asparagus}{rgb}{0.5,0.5,0}
\definecolor{Grape}{rgb}{0.5,0,1}
\definecolor{Iron}{rgb}{.3,.3,.3}
\definecolor{Steel}{rgb}{.4,.4,.4}

\def\llbracket{{[\![}}
\def\rrbracket{{]\!]}}
\usepackage{soul}  
\usepackage{tikzsymbols}  

\usepackage{enumitem} 

\makeatletter
\let\save@mathaccent\mathaccent
\newcommand*\if@single[3]{%
  \setbox0\hbox{${\mathaccent"0362{#1}}^H$}%
  \setbox2\hbox{${\mathaccent"0362{\kern0pt#1}}^H$}%
  \ifdim\ht0=\ht2 #3\else #2\fi
  }
\newcommand*\rel@kern[1]{\kern#1\dimexpr\macc@kerna}
\newcommand*\widebar[1]{\@ifnextchar^{{\wide@bar{#1}{0}}}{\wide@bar{#1}{1}}}
\newcommand*\wide@bar[2]{\if@single{#1}{\wide@bar@{#1}{#2}{1}}{\wide@bar@{#1}{#2}{2}}}
\newcommand*\wide@bar@[3]{%
  \begingroup
  \def\mathaccent##1##2{%
    \let\mathaccent\save@mathaccent
    \if#32 \let\macc@nucleus\first@char \fi
    \setbox\z@\hbox{$\macc@style{\macc@nucleus}_{}$}%
    \setbox\tw@\hbox{$\macc@style{\macc@nucleus}{}_{}$}%
    \dimen@\wd\tw@
    \advance\dimen@-\wd\z@
    \divide\dimen@ 3
    \@tempdima\wd\tw@
    \advance\@tempdima-\scriptspace
    \divide\@tempdima 10
    \advance\dimen@-\@tempdima
    \ifdim\dimen@>\z@ \dimen@0pt\fi
    \rel@kern{0.6}\kern-\dimen@
    \if#31
      \overline{\rel@kern{-0.6}\kern\dimen@\macc@nucleus\rel@kern{0.4}\kern\dimen@}%
      \advance\dimen@0.4\dimexpr\macc@kerna
      \let\final@kern#2%
      \ifdim\dimen@<\z@ \let\final@kern1\fi
      \if\final@kern1 \kern-\dimen@\fi
    \else
      \overline{\rel@kern{-0.6}\kern\dimen@#1}%
    \fi
  }%
  \macc@depth\@ne
  \let\math@bgroup\@empty \let\math@egroup\macc@set@skewchar
  \mathsurround\z@ \frozen@everymath{\mathgroup\macc@group\relax}%
  \macc@set@skewchar\relax
  \let\mathaccentV\macc@nested@a
  \if#31
    \macc@nested@a\relax111{#1}%
  \else
    \def\gobble@till@marker##1\endmarker{}%
    \futurelet\first@char\gobble@till@marker#1\endmarker
    \ifcat\noexpand\first@char A\else
      \def\first@char{}%
    \fi
    \macc@nested@a\relax111{\first@char}%
  \fi
  \endgroup
}
\makeatother

\DeclareMathOperator{\Aut}{Aut}
\DeclareMathOperator{\even}{even}
\DeclareMathOperator{\odd}{odd}
\DeclareMathOperator{\AGL}{AGL}
\DeclareMathOperator{\GL}{GL}
\newcommand{\Z}{\mathbb{Z}}

\def\numgenerators {{\color{black}R}}
\def\numXgenerators {{\color{black}r_X}}
\def\numZgenerators {{\color{black}r_Z}}
\def\hammingparameter {{\color{black}r}}	
\def\codeparameter {{\color{black}p}}  

\newtheorem{open question}{Open Problem}

\begin{document}
\title{Short Shor-style syndrome sequences}

\author{Nicolas Delfosse}
\author{Ben W. Reichardt}
\affiliation{Microsoft Quantum and Microsoft Research, Redmond, WA 98052, USA}


\begin{abstract}
We optimize fault-tolerant quantum error correction to reduce the number of syndrome bit measurements.  
Speeding up error correction will also speed up an encoded quantum computation, and should reduce its effective error rate.  
We give both code-specific and general methods, using a variety of techniques and in a variety of settings.  
We design new quantum error-correcting codes specifically for efficient error correction, e.g., allowing single-shot error correction.  
For codes with multiple logical qubits, we give methods for combining error correction with partial logical measurements.  
There are tradeoffs in choosing a code and error-correction technique.  
While to date most work has concentrated on optimizing the syndrome-extraction procedure, we show that there are also substantial benefits to optimizing how the measured syndromes are chosen and used.

As an example, we design single-shot measurement sequences for fault-tolerant quantum error correction with the $16$-qubit extended Hamming code.  Our scheme uses $10$ syndrome bit measurements, compared to $40$ measurements with the Shor scheme. We design single-shot logical measurements as well: any logical $Z$ measurement can be made together with fault-tolerant error correction using only $11$ measurements. For comparison, using the Shor scheme a basic implementation of such a non-destructive logical measurement uses $63$ measurements.

We also offer ten open problems, the solutions of which could lead to substantial improvements of fault-tolerant error correction. 
\end{abstract}

\maketitle
\fi

\vfuzz2pt 

\section{Introduction}

In quantum error correction, substantial work has gone into devising ways for measuring code stabilizers efficiently and fault tolerantly.  Less work has gone into how to use those syndrome bits efficiently.  That is, how can error correction be performed using as few as possible stabilizer measurements, in the case that faults may occur during stabilizer measurement and when the syndrome bits themselves may be faulty?  Which stabilizers should be measured, how many times, and in what order? 
Since fault-tolerant quantum error correction is so challenging to implement, it is important to optimize it.  

Figure~\ref{f:3bitrepetitionnonft} gives a toy example to illustrate the problem.  
The issue of choosing which stabilizers to measure does not typically arise for topological codes~\cite{DennisKitaevLandahlPreskill01topological, BombinMartindelgado06colorcode, FowlerMariantoniMartinisCleland12surfacecodes}, because then the measured stabilizers are chosen based on geometry, and all are measured, either in parallel or close to it.  But for block codes, there are many options.  
Shor's foundational work~\cite{Shor96}, for example, suggests repeating full syndrome extraction $\Omega(d^2)$ times in a row, for distance~$d$.  Bomb{\' i}n~\cite{Bombin15singleshot} has shown that for some specific, highly structured codes, ``single-shot" error correction is possible, meaning each stabilizer generator is measured only once.  
Delfosse et~al.~\cite{DelfosseReichardtSvore20singleshot} have studied fault tolerant error correction for high-distance codes, and show that $O(d \log d)$ stabilizer measurements suffice for any code with distance $d \geq n^\alpha$ for a constant $\alpha > 0$.  In fact, in some cases, the number of stabilizer measurements can be substantially ``sub-single-shot": exponentially fewer measurements are needed than the number of parity checks. 

\begin{figure}[b]
\centering
\vspace{-.1cm}  
\subfigure[Not fault tolerant$\!\!$]{\includegraphics[scale=.769]{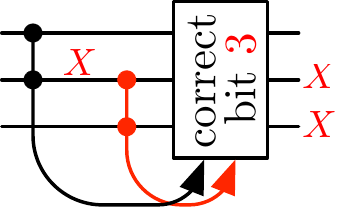}}
\hspace{.2cm}
\subfigure[Fault tolerant]{\includegraphics[scale=.769]{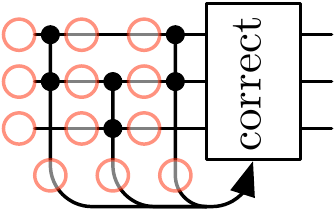} \hspace{.1cm} \raisebox{.3cm}{\includegraphics[scale=.615]{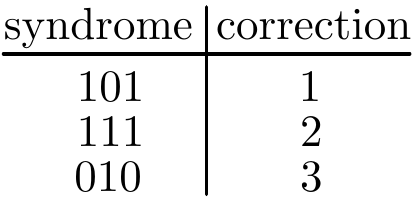}}}
\caption{For the three-bit repetition code $\{000, 111\}$, (a) it is not fault tolerant to correct errors based on the parities $1 \oplus 2$ and $2 \oplus 3$, because an internal fault on bit $2$ can be confused with an input error on bit~$3$.  (b) With the $1 \oplus 2$ parity measurement repeated, there are fault-tolerant correction rules: for at most one fault in the highlighted locations, any input error is corrected, and any internal fault results in a final error of weight zero or one.  
The third measurement can also be made adaptively, conditioned on at least one of the first two syndrome bits being nontrivial.  
} \label{f:3bitrepetitionnonft}
\end{figure}

Other research has focused on low-distance codes, as we will here.  
In an under-appreciated paper, Zalka~\cite{Zalka97} studied adaptive Shor-style error correction for the $\llbracket 7,1,3 \rrbracket$ Steane code. 
For $X$ error correction, Zalka extracts between four and eight $Z$ syndrome bits.  
(Zalka also considers applying multiple logical gates between error-correction steps, or even partial error-correction steps.)  
Using a technique very different from Shor-style error correction, 
Steane measures all $Z$ stabilizers simultaneously, and if the result is nontrivial measures an additional $\rho - 1$ full syndromes, where $\rho$ is optimized for each code, for example ranging from $\rho = 3$ for the $\llbracket 7,1,3 \rrbracket$ Steane code to $\rho = 4$ for the $\llbracket 23,1,7 \rrbracket$ Golay code~\cite[Table~I]{Steane03}.  

We, too, focus on small codes with distance $d \leq 7$.  Such codes could be practical for near-term quantum devices.  
We show, for example, that for CSS codes, mixing $X$ and $Z$ error correction can be more efficient than running them separately.  
For the $\llbracket 7,1,3 \rrbracket$ code, e.g., seven stabilizer measurements suffice for $X$ and $Z$ error correction together, versus ten measurements running them separately (\figref{f:steane713comparison} and \propref{t:hammingcodesnonadaptiveerrorcorrection}).  
Shor's method requires up to $24$ measurements, in comparison.  
Figure~\ref{f:codecomparison} shows more examples.  

\begin{figure}
\centering
\includegraphics[scale=.8]{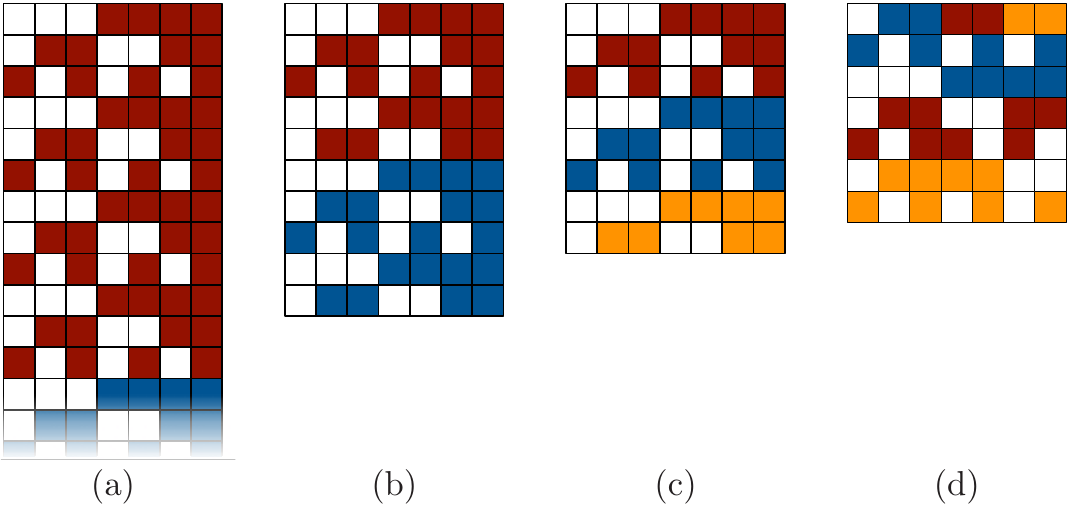}
\caption{Stabilizer measurement sequences for fault-tolerant error correction for the seven-qubit Steane code.  Here red represents Pauli~$Z$, blue is $X$ and orange is~$Y$; for example the first measurement in each sequence is of $IIIZZZ$.  (a)~As a baseline, Shor's method requires up to $24$ stabilizer measurements.  (b) For a distance-three CSS code with $\numgenerators$ independent stabilizers, $2\numgenerators - 2$ measurements suffice (\thmref{t:distancethreenonadaptive}).  (c) If the code is additionally self-dual and $Y$ measurements are allowed, then $3\numgenerators/2 - 1$ measurements suffice (\propref{t:hammingcodesnonadaptiveerrorcorrection}).  (d) Finally, for the seven-qubit code, if we allow measurements mixing $X$, $Y$ and $Z$ operators, then seven stabilizer measurements are enough.} \label{f:steane713comparison}
\end{figure}

\begin{figure}[t]
\centering
\includegraphics[scale=.9]{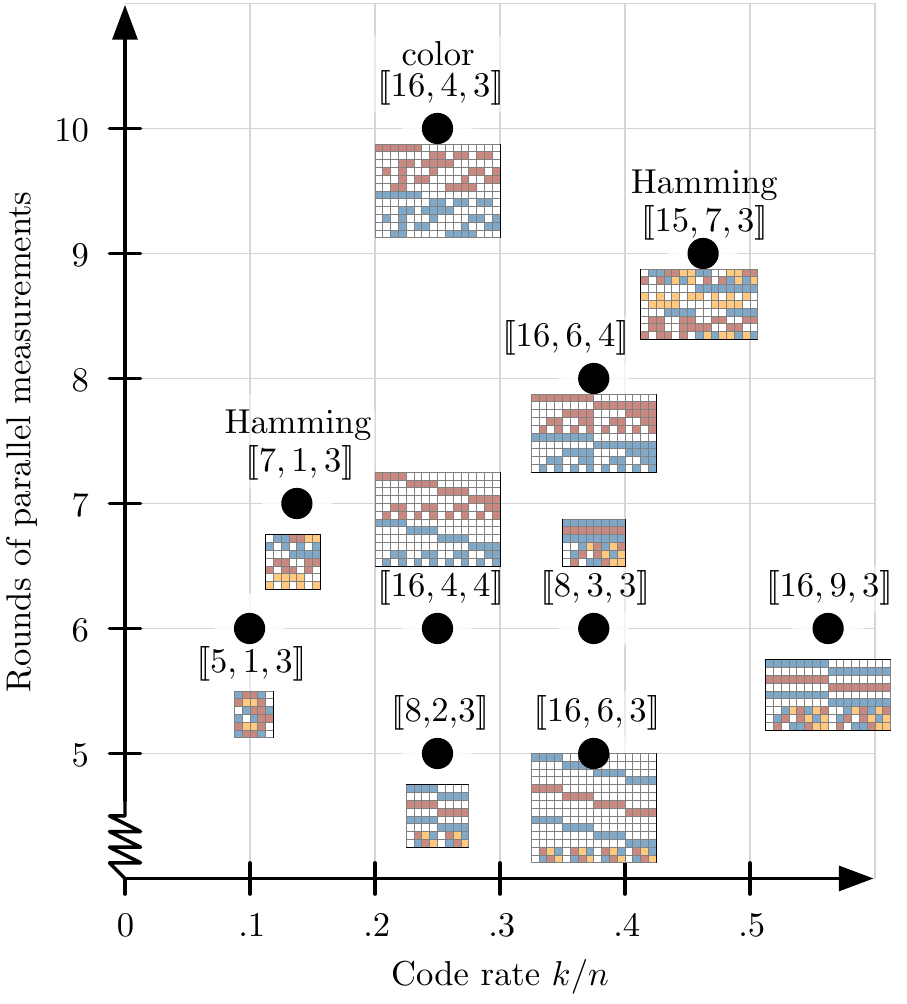}
\caption{A comparison of some of the small codes for which we find distance-\emph{three} fault-tolerant stabilizer measurement sequences.  For each code we have given the shortest known stabilizer measurement sequence, following the $X$, $Y$, $Z$ color convention of \figref{f:steane713comparison}.
(Distance-four fault-tolerant error correction for the $\llbracket 16,4,4 \rrbracket$ and $\llbracket 16,6,4 \rrbracket$ codes is considered in \secref{s:generalizedcodes}.)  
} \label{f:codecomparison}
\end{figure}

We study both nonadaptive error correction and adaptive error correction, in which the stabilizers you choose to measure can depend on previous measurement results.  
Adaptive measurements allow for significant improvements; see \figref{f:nonadaptiveadaptive}.  
We devise new quantum error-correcting codes specifically for efficient, single-shot error correction.  

\begin{figure*}
\begin{tabular}{c@{$\qquad$}c@{$\qquad$}c@{\hskip10pt}c}
\hline \hline
 & Available Pauli & \multicolumn{2}{c}{Length of our measurement sequence} \\
Code type ($d = 3$) & measurements & Nonadaptive case & Adaptive case \\ 
\hline
CSS code & All-$X$, all-$Z$ & $2 \numgenerators - 2$ & $\numgenerators$ to $2 \numgenerators - 2$ \hspace{0cm} (\secref{s:distance3csscodesadaptive}) \\
Self-dual CSS code & All-$X$, all-$Z$ & & $\numgenerators$ to $\tfrac32 \numgenerators$ \hspace{.5cm} (\claimref{t:selfdualcssadaptive}) \\
Stabilizer code & Arbitrary & $2 \numgenerators$ & $\numgenerators$ to $2 \numgenerators$ \hspace{.3cm} (\thmref{t:distancethreenoncss}) \\
\hline \hline
\end{tabular}
\caption{Summary of some of our general results for distance-three stabilizer codes, $\llbracket n, n - \numgenerators, 3 \rrbracket$.   
$\numgenerators$ is the number of stabilizer generators.  
By all-$X$, all-$Z$ measurements, we mean that the algorithm measures only operators that are tensor products of $X$ and $I$, or of $Z$ and $I$.  
In comparison, Shor's adaptive fault-tolerant measurement sequence has length $2 \numgenerators$ to $4 \numgenerators$ in all cases.} \label{f:nonadaptiveadaptive}
\end{figure*}

For example, for the standard $\llbracket 15,7,3 \rrbracket$ Hamming code, distance-three fault-tolerant error correction can be done with nine measurements of mixed $X$, $Y$ and $Z$ stabilizers (of weights eight or $12$), or $X$ and $Z$ error correction can be run separately with $14$ weight-eight measurements (\propref{t:hammingcodesnonadaptiveerrorcorrection}).  
(Shor's method uses $32$ measurements.)

We also consider fault-tolerant logical measurements.  
A standard way to measure logical $Z$s with a CSS code is to measure all the qubits in the computational basis.  
The logical measurement outcomes can then be obtained by correcting the measured noisy bit string.
This method destructively measures all the logical qubits in a code block.  
Alternatively, one can measure a subset of logical qubits by moving them to an ancilla block which is then measured destructively~\cite{Gottesman2013constantoverhead, NautrupFriisBriegel2017ftinterface, Breuckmann2017hyperbolic}, or using a Steane-type ancillary block~\cite{ZhengLaiBrunKwek20constant}.  
Here, we design measurement sequences that perform logical measurements without any extra ancillary block, 
eliminating the time and space required for ancilla preparation.  
Strikingly, combining error correction with logical qubit measurements can be substantially more efficient than running these operations separately.  
For example, for the $\llbracket 15,7,3 \rrbracket$ Hamming code, any weight-five logical $Z$ operator can be measured fault tolerantly with five measurements, while combining the logical $Z$ measurement with $X$ error correction needs only six measurements---one \emph{fewer} than $X$ error correction alone.  
(See \figref{f:1573measurementsforlogicalmeasurement}.)  

To further illustrate the savings obtained in this work, consider the $\llbracket 16, 6, 4 \rrbracket$ extended Hamming code. We prove that distance-three fault-tolerant error correction is possible with this code using 10 stabilizer measurements: 5 $X$ and 5 $Z$ measurements. A popular alternative is the Shor scheme which requires up to 40 stabilizer measurements (four rounds of five measurements for each error type) \cite{Shor96}.
Additionally, we design measurement sequences that allow for fault-tolerant logical measurement of any $X$ or $Z$ logical operator combined with fault-tolerant error correction with only 11 measurements.
For comparison, one could perform a logical $Z$ measurement by measuring a representative of the logical operator three times. To obtain the right outcome after a majority vote, $X$ error correction must be performed between 
logical measurements; then finally $Z$ error correction. 
Using Shor's scheme, this approach uses a total of $63 = 3(1+20)$ measurements. Our method is over five times faster.
The fact that 11 measurements suffice to perform simultaneously a logical measurement and fault-tolerant error correction for a code defined by 10 independent stabilizer generators is surprising. We introduce the concept of single-shot logical measurement in~\secref{s:furtherlogicalmeasurementproblems}.

Our aim here is not to propose one, most efficient error-correction procedure.  
Instead, we show a variety of new techniques for different codes.  
No doubt there is room for further improvement.  
In the end, choosing an error-correction method requires balancing tradeoffs, such as a space-time tradeoff between code size and error correction time.  
There is a rich scope for exploration.  

\medskip

In \secref{s:faulttolerance} we define fault-tolerant error correction.  
\secref{s:model} uses the $\llbracket 7,1,3 \rrbracket$ Steane code to introduce the problem of Shor-style fault-tolerant error correction.  
\secref{s:distance3csscodes} generalizes that example to arbitrary distance-three CSS codes, with both nonadaptive and adaptive stabilizer measurement orders.  
\secref{s:513code} introduces the non-CSS setting, using the $\llbracket 5,1,3 \rrbracket$ code as an example.  
\secref{s:hammingcodes} shows that even for CSS codes---for Hamming codes, including the $\llbracket 7,1,3 \rrbracket$ code---it can be more efficient to mix $X$ and $Z$ error correction than to run them separately.  
\secref{s:1643colorcode} shows that single-shot fault-tolerant error correction is possible, for a certain $\llbracket 16, 4, 3 \rrbracket$ color code.  
\secref{s:codesdesignedforfastec} presents five new families of codes, all of distance three or four, that give different tradeoffs for encoding rate versus the number of stabilizer measurement rounds needed for fault-tolerant error correction.  
\secref{s:logicalmeasurement} studies the problem of combining logical measurement with error correction.  
The appendices include several extensions.  For example, in \appref{s:othermeasurementmodels} we consider alternative models for stabilizer measurement, including flagged measurements (along the lines of the ``flag paradigm"~\cite{ChaoReichardt18flags}) and stabilizer measurements in parallel.

\section{Fault tolerance} \label{s:faulttolerance}

\begin{definition}[Fault tolerant error correction] \label{t:faulttolerance}
An error-correction procedure is fault tolerant to distance~$d = 2 t + 1$, if provided the number of input and internal faults is at most~$t$, the output error's weight, up to stabilizers, is at most the number of internal faults.  For CSS fault tolerance, the weight of a Pauli error is taken to be the maximum weight of its $X$ and $Z$ parts, up to stabilizers.  
\end{definition}

\noindent
For example, $X \otimes Y \otimes Z$ has weight three, but its $X$ and $Z$ parts, $X \otimes X \otimes I$ and $I \otimes Z \otimes Z$, have weight~two.  

Another fault-tolerance condition can also be required~\cite{AliferisGottesmanPreskill05, Gottesman09faulttolerance}: on an arbitrary input, provided that the number of internal faults is at most~$t$, the output should lie at most distance~$t$ from the codespace.  
This condition is important for concatenated fault-tolerance schemes, in which a corrupted codeword must be returned to the codespace so that the next higher level of error correction can correct an encoded error.  
With apologies to field theorists, we call this stricter definition ``concatenation fault tolerant" (CFT) error correction.  
For the fault-tolerant computation at the highest level of code concatenation, only \defref{t:faulttolerance} is needed.  

For perfect distance-three codes, CSS or not, fault tolerant and CFT error correction are equivalent, but this equivalence does not hold in general.  
Figure~\ref{f:6bitrepetition} shows an error-correction procedure for the six-bit repetition code that is fault tolerant to distance three but not CFT to distance three. 
Concatenation is a useful tool for proving the threshold theorem~\cite{AliferisGottesmanPreskill05, Gottesman09faulttolerance}.  
However, it is difficult to imagine multiple concatenation levels being used in practice because of the high qubit and time overhead.  
Therefore in the sequel we mostly consider only the weaker \defref{t:faulttolerance} (except in \secref{s:1643colorcode}). 

\begin{figure}
\centering
\subfigure[]{\includegraphics[scale=.769]{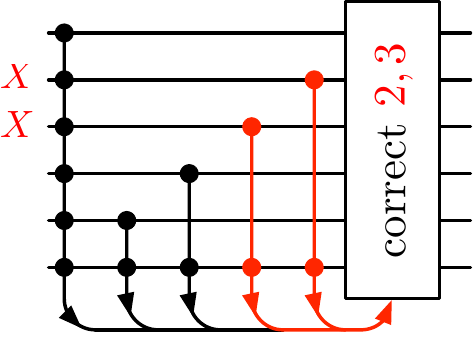}}
\hspace{.8cm}
\subfigure[]{\includegraphics[scale=.769]{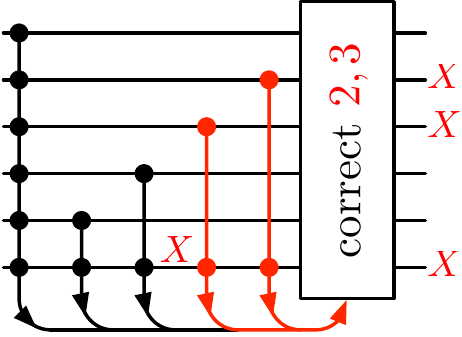}}
\caption{
Fault tolerance does not imply concatenation fault tolerance.  
For the six-bit repetition code $\{000000, 111111\}$, five syndrome bit measurements (single shot) suffice for error correction 
fault tolerant to distance three.  However, there are no correction rules to make this procedure concatenation fault tolerant to distance three.  In (a) are shown two input errors that must be corrected back to the codespace (by correcting either bits $2, 3$ or bits $1, 4, 5, 6$) for CFT to hold.  However, then a single internal fault (b) can result in the output being distance three from the codespace, contradicting CFT.  
} \label{f:6bitrepetition}
\end{figure}

(Concatenation fault-tolerant error correction can also be used for state preparation.  For example, for an $\llbracket n, k, d \rrbracket$ CSS code, $\alpha = \braket{0^n}{\overline 0{}^k} > 0$, so $\ket{0^n} = \tfrac{1}{\alpha} \ket{0^n} \! \braket{0^n}{\overline 0{}^k} = \tfrac{1}{\alpha} \big( \tfrac{I + Z}{2} \big)^{\otimes n} \ket{\overline 0{}^k} = \tfrac{1}{2^n \alpha} \sum_{S \subseteq [n]} Z_S \ket{\overline 0{}^k}$; and so starting from $\ket{0^n}$ one can fault-tolerantly prepare the encoded state $\ket{\overline 0{}^k}$ using a concatenation fault-tolerant $Z$ error-correction procedure.  However, there are usually more efficient methods for fault-tolerant state preparation~\cite{PaetznickReichardt11Golay, ZhengLaiBrun18stateprep}.)  

Note that \defref{t:faulttolerance} is for error correction.  Different definitions apply for error detection and for combinations of detection and correction.  (For example, a distance-three code can detect up to two errors, and a distance-four code can detect three errors, or detect two and correct one error.)  We focus on error correction.  

\section{Model: Error correction based on single syndrome bit measurements} \label{s:model}

In this section we introduce the model of Shor-style fault-tolerant error correction~\cite{Shor96}, based on fault-tolerantly measuring one syndrome bit at a time.  
We illustrate the model using the $\llbracket 7,1,3 \rrbracket$ Steane code as an example.  
In \secref{s:distance3csscodes}, we will generalize the arguments to distance-three CSS codes.

\subsection{Shor-style syndrome measurement}

\begin{figure}
\centering
\begin{tabular}{c}
\subfigure[Non-fault-tolerant measurement]{\raisebox{.2cm}{\includegraphics[scale=.769]{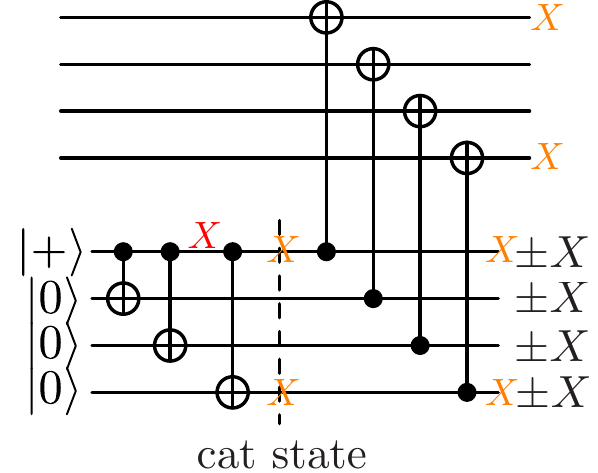}}} \\
\subfigure[Fault-tolerant measurement\label{f:shorstylesyndromemeasurementft}]{\raisebox{.2cm}{\includegraphics[scale=.769]{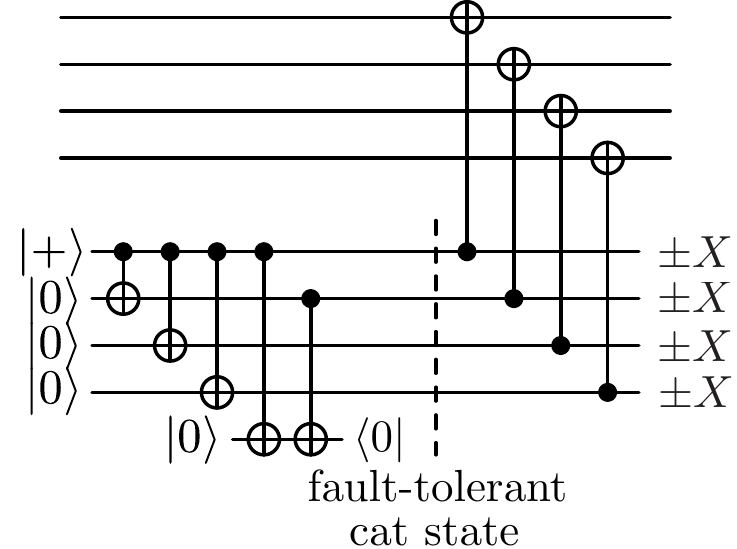}}}
\end{tabular}
\caption{
(a) A $w$-qubit cat state, $\tfrac{1}{\sqrt 2}(\ket{0^w} + \ket{1^w})$, can be used to measure $X^{\otimes w}$, as shown here with $w = 4$.  
However, the cat state should be prepared fault tolerantly, which here it is not: a single $X$ fault, at the location in red, spreads to a weight-two $X$ error on the data.  
(b) The $w = 4$ qubit cat state can be prepared fault tolerantly using one extra ancilla qubit; conditioning on measuring an even parity for the first two qubits avoids the problem in (a), so that this circuit is fault tolerant.
} \label{f:shorstylesyndromemeasurement}
\end{figure}

In Shor-style syndrome measurement schemes, stabilizers are measured one bit at a time using fault-tolerantly prepared cat states.  
For example, to measure $X^{\otimes w}$, one can first prepare a cat state $\tfrac{1}{\sqrt 2}(\ket{0^w} + \ket{1^w})$ using a fault-tolerant Clifford circuit.  
(For fault tolerance to distance $d = 2 t + 1$, the preparation circuit should satisfy that for any $k \leq t$ Pauli gate faults, the weight of the error on the output state, modulo stabilizers, is at most~$k$.)  
Then this cat state is coupled to the data with transversal CNOT gates and each of its qubits measured in the Hadamard, or $\ket + / \ket -$, basis.  
The parity of the $w$ measurements is the desired syndrome bit.  
See \figref{f:shorstylesyndromemeasurement}.  

One can also measure $X^{\otimes w}$ with a cat state on, potentially, fewer than $w$ qubits, or even without using a cat state at all~\cite{Stephens14colorcodeft, YoderKim16trianglecodes, ChaoReichardt18flags}.  For example, \figref{f:condensedshorstylesyndromemeasurement} shows two circuits for measuring $X^{\otimes 6}$ using only three ancilla qubits, both CSS fault tolerant to distance three.  

\begin{figure}
\centering
\begin{tabular}{c}
\includegraphics[scale=.769]{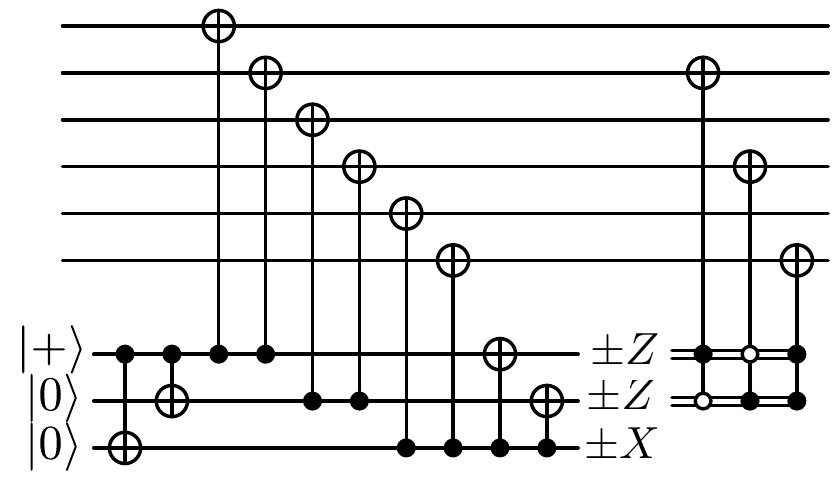} \\[.4cm]
\includegraphics[scale=.769]{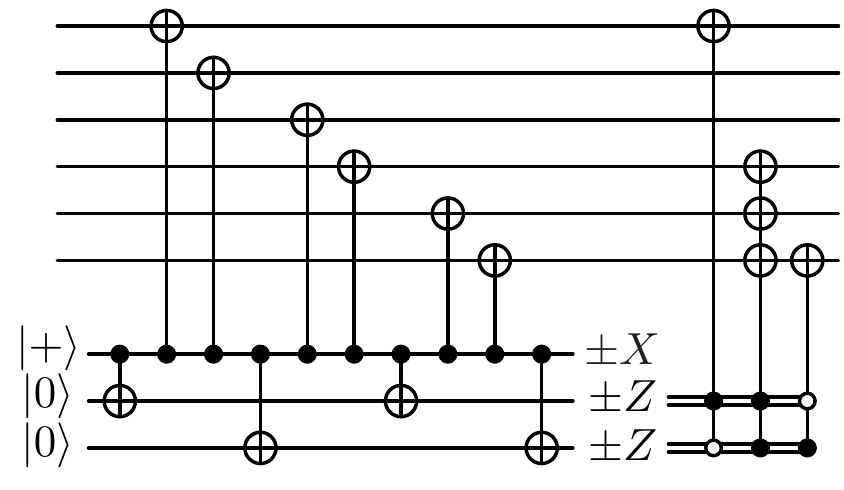}
\end{tabular}
\caption{Two circuits for measuring $X^{\otimes 6}$ using three ancilla qubits.  Both circuits are CSS fault tolerant to distance three, i.e., a single $X$ fault will result in a data error of weight at most one.  The first circuit is from~\cite{ChaoReichardt18flags}, generalizing a construction due to~\cite{Stephens14colorcodeft, YoderKim16trianglecodes}.  Note that the second circuit does not use a cat state.  It is a special case of a ``flag fault-tolerant" procedure from~\cite{PrabhuReichardt20catstates}.} \label{f:condensedshorstylesyndromemeasurement}
\end{figure}

Here we take (CSS) fault-tolerant syndrome bit measurement as a primitive, and use it as a building block for fault-tolerant error correction, fault-tolerant logical measurement, and other operations.  The details of how stabilizers are measured will not be important.  What is important is that they are measured one at a time, in sequence, and not all at once as in Steane-style~\cite{Steane97} or Knill-style~\cite{Knill03erasure} error correction.  
In \appref{s:othermeasurementmodels} we consider other syndrome measurement models, including flag fault-tolerant measurement~\cite{ChaoReichardt18flags} and models intermediate between Shor- and Steane-style syndrome measurement.

\subsection{
Error correction for the \texorpdfstring{$\llbracket 7,1,3 \rrbracket$}{[[7,1,3]]}~code}

Consider Steane's $\llbracket 7,1,3 \rrbracket$ code~\cite{Steane96css}, a self-dual CSS code with $Z$ stabilizers given by 
\begin{equation*}
\begin{tabular}{c@{$\otimes$}c@{$\otimes$}c@{$\otimes$}c@{$\otimes$}c@{$\otimes$}c@{$\otimes$}c}
$I$&$I$&$I$&$Z$&$Z$&$Z$&$Z$ \\
$I$&$Z$&$Z$&$I$&$I$&$Z$&$Z$ \\
$Z$&$I$&$Z$&$I$&$Z$&$I$&$Z$ 
\end{tabular}
\end{equation*}
The code can correct one input $X$ error---it has distance three---because every weight-one $X$ error has a distinct syndrome, e.g., $X_1$ gives syndrome $001$ because it commutes with the first two stabilizers and anticommutes with the third.  

However, it is not fault tolerant to simply measure these three stabilizers and apply the corresponding correction.  For example, it might be that the input is perfect but an $X_7$ fault occurs right after measuring the second stabilizer.  Then the observed syndrome will be $001$, and applying an $X_1$ correction will leave the data with a weight-two error, $X_1 X_7$.  Similarly, an $X_7$ fault after measuring the first stabilizer will give syndrome $011$ and therefore leave the corrected data with error $X_3 X_7$.  

To handle faults that occur during error correction, for this code we need more stabilizer measurements.  For example, say we measure the first stabilizer again, so the measurement sequence is 
\begin{equation*}
\begin{array}{c c c c c c c}
0&0&0&1&1&1&1\\
0&1&1&0&0&1&1\\
1&0&1&0&1&0&1\\
0&0&0&1&1&1&1
\end{array}
\end{equation*}
where we have adopted a less cumbersome notation, with $0$ meaning $I$ and $1$ meaning~$Z$.  Now an internal $X_7$ fault can result in the syndromes $0111$, $0011$, $0001$ or $0000$ (coming from suffixes of the last column above).  As none of these syndromes can be confused with that from an input error on a different qubit, an error-correction procedure can safely apply no correction at all in these cases.  (Alternatively, one could correct $X_7$ for the syndrome $0111$ and give no correction for $0011$ or $0001$.)  

However, the above four-measurement sequence still does not suffice for fault-tolerant $X$ error correction, because an internal fault on qubit~$3$ can also cause the syndrome $0010$.  A fifth measurement is needed to distinguish an input $X_1$ error from an internal $X_3$ fault.  For example, this measurement sequence works: 
\begin{equation} \label{e:713codefivemeasurementerrorcorrectionsequence}
\begin{array}{c c c c c c c}
0&0&0&1&1&1&1\\
0&1&1&0&0&1&1\\
1&0&1&0&1&0&1\\
0&0&0&1&1&1&1\\
0&1&1&0&0&1&1
\end{array}
\end{equation}

Note that after the first four stabilizer measurements, the only bad case remaining is the suffix $0010$ of column~$3$, $0110$.  For the fifth measurement, we can therefore use any stabilizer that distinguishes qubits~$1$ and~$3$.  This need not be one of the stabilizer generators, e.g., $0111100$ also works.  

In this paper, we will develop fault-tolerant stabilizer measurement sequences for other codes, including codes with distance $> 3$, for error correction and other operations.  In addition to fixed measurement sequences like Eq.~\eqnref{e:713codefivemeasurementerrorcorrectionsequence}, we will also consider adaptive measurement sequences, in which the choice of the next stabilizer to measure depends on the syndrome bits already observed.

\section{Distance-three CSS codes} \label{s:distance3csscodes}

Having established the setting of sequential fault-tolerant stabilizer measurements, let us next consider stabilizer measurement sequences for fault-tolerant error correction for general distance-three CSS codes.

\subsection{Algorithm for distance-three CSS codes}

The argument leading to Eq.~\eqnref{e:713codefivemeasurementerrorcorrectionsequence} suggests a general procedure for constructing measurement sequences for distance-three CSS fault-tolerant error correction: 

\begin{itemize}[leftmargin=*]
\item 
Call a pair $(i, j)$ of qubits ``bad" if an internal fault on qubit~$j$ can result in the same syndrome as an $X$ input error on a different qubit~$i$.  If the columns of the length-$m$ measurement sequence are $c_1, \ldots, c_n \in \{0,1\}^m$, then qubit~$j$ is bad if for some $k \in \{0, 1, \ldots, m\}$, the suffix $0^k (c_j)_{k+1} \ldots (c_j)_m = c_i$ with $i \neq j$.  
\item 
Then repeat, while there exists a bad pair $(i, j)$: Append to the measurement sequence a stabilizer that is $0$ on qubit~$i$ and $1$ on qubit~$j$, or vice versa.  
\end{itemize}
The algorithm eventually terminates because for a distance-three CSS code, for any pair $(i, j)$ there must exist a $Z$ stabilizer that distinguishes $X_i$ from~$X_j$.  (That is, unless the code is degenerate, i.e., $X_i X_j$ is a stabilizer.  For a degenerate code with weight-two stabilizers, the definition of ``bad" should require that $X_i$ and $X_j$ be inequivalent.)  
When there are no bad pairs left, the procedure is CSS fault tolerant to distance three.  

A natural greedy version of this algorithm might, for example, choose to add the stabilizer that eliminates the most bad qubit pairs.

\subsection{Nonadaptive measurement sequence for any distance-three CSS code} \label{s:distance3csscodesnonadaptive}

We next construct a fault-tolerant error-correction procedure for any 
distance-three CSS code: 

\begin{theorem} \label{t:distancethreenonadaptive}
Consider an $\llbracket n, n - \numZgenerators - \numXgenerators, 3 \rrbracket$ CSS code with $\numZgenerators$ independent $Z$ stabilizer generators $g_1, \ldots, g_{\numZgenerators}$.  
Then fault-tolerant $X$ error correction can be realized with $2 \numZgenerators - 1$ syndrome bit measurements, by measuring in order all the generators $g_1, \ldots, g_{\numZgenerators}$, followed by just $g_1, \ldots, g_{\numZgenerators - 1}$.  
\end{theorem}

For example, for Steane's $\llbracket 7,1,3 \rrbracket$ code, $X$ and $Z$ error correction can each be done with five measurements, as in Eq.~\eqnref{e:713codefivemeasurementerrorcorrectionsequence}.   
This is optimal, in the sense of using the fewest possible $Z$ measurements for $X$ error correction.  
More generally, for the $\llbracket 2^\hammingparameter - 1, 2^\hammingparameter - 1 - 2\hammingparameter, 3 \rrbracket$ Hamming code (see \secref{s:hammingcodes} below), $X$ and $Z$ error correction can each be done with $2 \hammingparameter - 1$ syndrome bit measurements.  
For some other distance-three codes, fewer measurements are possible, as we will see in Secs.~\ref{s:1643colorcode} and~\ref{s:codesdesignedforfastec}.  

\begin{proof}[Proof of \thmref{t:distancethreenonadaptive}]
The concern is that an internal $X$ fault might be confused with an input $X$ error.  (We need not worry about an incorrectly flipped measurement, since it at worst it could cause a weight-one correction to be wrongly applied.)  

For an internal fault occurring after the first $\numZgenerators$ measurements, the first $\numZgenerators$ outcomes will be trivial and therefore different from the case of any input error.  

This leaves as possibly problematic only internal faults occurring among the first $\numZgenerators$ measurements (after $g_1$ and before $g_{\numZgenerators}$).  Consider an internal fault that results in the measured syndrome $(s, t, s')$, where $s, s' \in \{0,1\}^{\numZgenerators - 1}$ are the first and second syndrome vectors for $g_1, \ldots, g_{\numZgenerators - 1}$, and $t \in \{0,1\}$ is the syndrome bit of $g_{\numZgenerators}$.  Assume that the fault occurs on qubit~$q$, after the measurement of a stabilizer $g_i$ that involves that qubit~$q$; if it occurs before every stabilizer involving that qubit~$q$, then it is equivalent to an input error.  Since $s_i = 0$, which is incorrect for an input error on qubit~$q$, this means that the syndromes $(s, t)$ and $(s', t)$ will be inconsistent in the sense that they correspond to different input errors; $(s', t)$ corresponds to input error $X_q$, while $(s, t)$ corresponds to some other, inequivalent input error or no input error.\footnote{One-qubit errors $X_i$ and $X_j$ are inequivalent if they correspond to different syndromes.  They can be equivalent, even if $i \neq j$, if the code is degenerate and $X_i X_j$ is a stabilizer.}  Therefore the syndrome $(s, t, s')$ is not consistent with any input error.  
\end{proof}

Observe that the important property for this proof to work is that both the first $\numZgenerators$ stabilizers measured and the last $\numZgenerators$ stabilizers measured form independent sets of generators.  They need not both be $\{g_1, \ldots, g_{\numZgenerators}\}$.

\subsection{Adaptive measurement sequence for any distance-three CSS code} \label{s:distance3csscodesadaptive}

\thmref{t:distancethreenonadaptive} constructs a nonadaptive error-correction procedure, in which the same stabilizers are measured no matter the outcomes.  An adaptive syndrome-measurement procedure can certainly be more efficient.  For example, if the first measured syndrome bits of $g_1, \ldots, g_{\numZgenerators}$ are all trivial, then one can end $X$ error correction without making the remaining $\numZgenerators - 1$ measurements.  The following adaptive $X$ error correction procedure works for any distance-three CSS code: 

{ \noindent \hrulefill \\
\centering \textbf{Adaptive $X$ error-correction procedure \\ for a distance-three CSS code} \\ } \smallskip

\begin{enumerate}[leftmargin=*]
\item Measure the $\numZgenerators$ stabilizer generators, stopping after the first nontrivial measurement outcome.  If all syndrome bits are trivial, then end error correction, having made $\numZgenerators$ measurements total.  
\item If the measurement of~$g_j$ is nontrivial, then measure $g_1, g_2, \ldots, g_{j-1}, g_{j+1}, \ldots, g_{\numZgenerators}$.  Apply the appropriate correction based on these and the nontrivial outcome for syndrome bit~$g_j$, having made $j + \numZgenerators - 1$ measurements total.  
\end{enumerate}
\vspace{-1\baselineskip}
\hrulefill
\medskip

The procedure uses between $\numZgenerators$ and $2 \numZgenerators - 1$ measurements, the worst case being if the first nontrivial measurement is for $g_{(j = \numZgenerators)}$.  It is advantageous to detect errors early.  

An alternative way to prove \thmref{t:distancethreenonadaptive} is to notice that the theorem's nonadaptive measurement sequence includes as subsequences this adaptive procedure's possible measurement sequences.  

The above procedures treat $X$ and $Z$ faults completely independently, and can tolerate, e.g., one internal $X$ fault and one internal $Z$ fault even on different qubits.  If instead we allow one internal fault total, $X$, $Z$ or $Y$, within the entire error-correction procedure, then the adaptive error-correction procedure can be shortened; see \appref{s:selfdualcssadaptive}.

\section{Distance-three stabilizer codes}
\label{s:513code}

The arguments from \secref{s:distance3csscodes} generalize to all distance-three stabilizer codes, with one important difference.  
We begin with an example to show how fault-tolerant error correction compares for CSS versus non-CSS codes.  
Then we will generalize the measurement sequences of Secs.~\ref{s:distance3csscodesnonadaptive} and~\ref{s:distance3csscodesadaptive} to arbitrary distance-three stabilizer codes.

\subsection{\texorpdfstring{$\llbracket 5,1,3 \rrbracket$}{[[5,1,3]]} code}

Consider the perfect $\llbracket 5,1,3 \rrbracket$ code~\cite{LaflammeMiquelPazZurek96}, encoding one logical qubit into five physical qubits to distance three.  The stabilizer group is generated by $XZZXI$ and its cyclic permutations $IXZZX, XIXZZ, ZXIXZ$.  For a deterministic (nonadaptive) distance-three fault-tolerant error-correction procedure, it suffices to measure fault-tolerantly the six stabilizers in \figref{f:513sequence}.  

\begin{figure}
\centering
\subfigure[\label{f:513sequence}]{
\raisebox{.2cm}{\includegraphics[scale=.769]{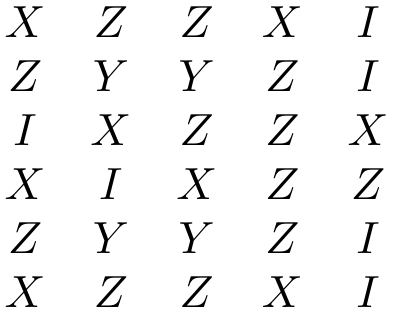}}
}
\subfigure[\label{f:513inputerrorsyndromes}]{\raisebox{.2cm}{\includegraphics[scale=.769]{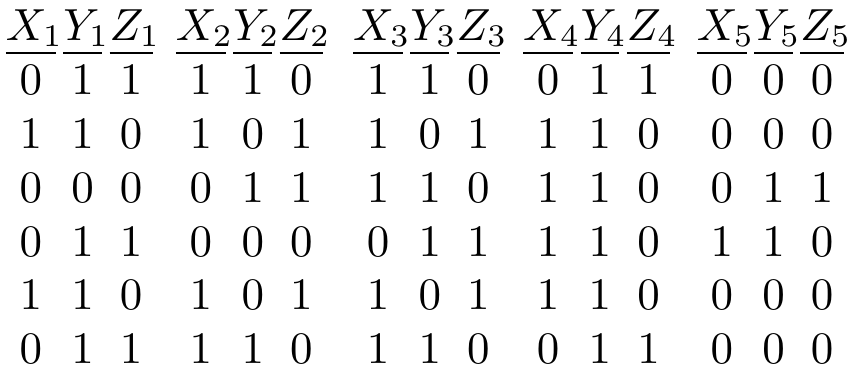}}}
\caption{
(a) Measurement sequence for fault-tolerant error correction for the $\llbracket 5,1,3 \rrbracket$ code.  
(b) Error syndromes for the $15$ weight-one input errors.  
} \label{f:513code}
\end{figure}

It matters which stabilizers are measured and in what order.  
For example, consider if one only measured the first five of the above six stabilizers.  
The syndromes for $X_1$ and $X_2$ input errors would be $01001$ and $11001$, respectively.  However, if the input were perfect and an $X_2$ fault occurred just after measuring the first syndrome bit, this would \emph{also} generate the syndrome $01001$.  Applying an $X_1$ correction would leave the weight-two error $X_1 X_2$ on the data ($X_1 X_2 Z_4$ is a logical operator).  
The problem here is that the suffix $01001$ of the $X_2$ input syndrome $11001$ matches the syndrome for an input error on a different qubit.  
If no syndrome suffixes collide in this way, then the error-correction procedure tolerates faults happening \emph{between} syndrome measurements.  

For CSS codes for which $X$ and $Z$ error correction are conducted separately, then it is sufficient to ensure that the procedures tolerate faults between syndrome measurements, i.e., that no suffix of an input error syndrome collides with the syndrome for an input error on a different qubit.  
For non-CSS codes, however, it is not sufficient to consider faults between syndrome measurements.  
For example, assume again that we only measure the first five of the stabilizers in \figref{f:513sequence}.  
While measuring the first stabilizer, it is possible that a single fault introduces an $X_1$ error while simultaneously flipping the syndrome bit.  This is not equivalent to an $X_1$ fault just before or just after the measurement, because $X_1$ commutes with the stabilizer.  
(For example, in \figref{f:shorstylesyndromemeasurementft} to measure $X^{\otimes 4}$, an $XZ$ fault on first CNOT coupling the data and cat state would cause an $X_1$ data error and flip the syndrome bit.)  
This fault leads to the syndrome $11001$, which is \emph{not} a suffix of $01001$.  It matches the syndrome for an input $X_2$ error, and applying an $X_2$ correction would leave the weight-two error $X_1 X_2$.  

The measurement sequence in \figref{f:513sequence} 
tolerates single $X, Y$ or $Z$ faults, anywhere within the support of a measured stabilizer, that also flip the syndrome bit.

\subsection{Nonadaptive and adaptive measurement sequences for any distance-three stabilizer code}

The nonadaptive and adaptive measurement sequences of Secs.~\ref{s:distance3csscodesnonadaptive} and~\ref{s:distance3csscodesadaptive} generalize to arbitrary distance-three stabilizer codes, using 
two more measurements: 

\begin{theorem} \label{t:distancethreenoncss}
Consider an $\llbracket n, n - \numgenerators, 3 \rrbracket$ stabilizer code with $\numgenerators$ independent stabilizer generators $g_1, \ldots, g_\numgenerators$.  
Then fault-tolerant error correction can be realized with $2 \numgenerators$ stabilizer measurements: $g_1, \ldots, g_\numgenerators, g_1, \ldots, g_\numgenerators$.  

With adaptive measurements, between $\numgenerators$ and $2 \numgenerators$ measurements suffice: measure $g_1, g_2, \ldots$ and if a syndrome bit is nontrivial, stop and measure $g_1, \ldots, g_\numgenerators$ to determine the correction.  
\end{theorem}

Compared to the adaptive procedure for CSS codes in \secref{s:distance3csscodesadaptive}, if the measurement of $g_j$ is nontrivial, then here we re-measure all of $g_1, \ldots, g_\numgenerators$, including $g_j$ again.  This is necessary, in general, because an internal fault that triggers $g_j$ can give an error that either commutes or anticommutes with $g_j$.

\section{Mixing $X$ and $Z$ error correction for Hamming codes} \label{s:hammingcodes}

In this section, we show that, even for CSS codes, mixing $X$ and $Z$ error correction can be more efficient than running them separately.  

The Hamming codes are a family of $\llbracket 2^\hammingparameter - 1, 2^\hammingparameter - 1 - 2\hammingparameter, 3 \rrbracket$ quantum error-correcting codes, for $\hammingparameter \geq 3$.  They are self-dual, perfect CSS codes.  
For example, the $\llbracket 7,1,3 \rrbracket$ and $\llbracket 15,7,3 \rrbracket$ Hamming codes have stabilizer generators given respectively by, in both Pauli $Z$ and $X$ bases, 
\begin{equation} \label{e:hammingcodestabilizers}
\begin{array}{c c c c c c c}
0&0&0&1&1&1&1\\
0&1&1&0&0&1&1\\
1&0&1&0&1&0&1
\end{array}
\;\;\text{ and }\;\;
\begin{smallmatrix}
0&0&0&0&0&0&0&1&1&1&1&1&1&1&1\\
0&0&0&1&1&1&1&0&0&0&0&1&1&1&1\\
0&1&1&0&0&1&1&0&0&1&1&0&0&1&1\\
1&0&1&0&1&0&1&0&1&0&1&0&1&0&1 
\end{smallmatrix}
\end{equation}

We will show: 

\begin{proposition} \label{t:hammingcodesnonadaptiveerrorcorrection} \ 
\begin{itemize}[leftmargin=*]
\item 
For the $\llbracket 7,1,3 \rrbracket$ Steane code, measuring in order 
the seven stabilizers from Eq.~\eqnref{e:713sevenstabilizers} below suffices for distance-three fault-tolerant error correction.  
For the $\llbracket 15,7,3 \rrbracket$ Hamming code, measuring the nine stabilizers of Eq.~\eqnref{e:1573ninestabilizers} suffices.  
\item 
In general, for the $2^\hammingparameter - 1$ qubit 
Hamming code, there is a sequence of $3 \hammingparameter - 1$ stabilizer measurements (beginning with the $2 \hammingparameter$ standard $Z$ and $X$ stabilizer generators, and ending with certain $\hammingparameter - 1$ $Y$ basis stabilizers) that suffice for distance-three fault-tolerant error correction.  
(See, e.g., Eqs.~\eqnref{e:713eightstabilizers} and~\eqnref{e:1573ZXYstabilizers} for the $\hammingparameter = 3$ and $\hammingparameter = 4$ cases, respectively.)  
\item 
For the $2^\hammingparameter - 1$ qubit  
Hamming code, one can separately correct $X$ and $Z$ errors fault tolerantly by measuring $2 \hammingparameter - 1$ $Z$ and $2 \hammingparameter - 1$ $X$ stabilizers, $4 \hammingparameter - 2$ stabilizer measurements total.  (This is a special case of \thmref{t:distancethreenonadaptive}.) 
\end{itemize}
\end{proposition}

Note that it is not fault tolerant just to measure the $\hammingparameter$ $Z$ and $\hammingparameter$ $X$ stabilizer generators fault tolerantly.  For example, with either code, an $X$ fault on qubit $3$ just before the last stabilizer measurement creates the same syndrome as an input $X$ error on qubit $1$.  But applying an $X_1$ correction would result in the error $X_1 X_3$, which is one away from the logical error $X_1 X_2 X_3$.  In fact, because of the perfect CSS property, the $2^\hammingparameter - 1$ possible weight-one input errors use all $2^\hammingparameter - 1$ possible nontrivial $\hammingparameter$-stabilizer syndromes.  Necessarily, therefore, some faults during syndrome extraction will lead to syndromes that are the same as syndromes from input errors.  Thus sequential measurement of any fixed set of $\hammingparameter$ stabilizer generators can never be fault tolerant.  More measurements are needed.  

Consider first the $\llbracket 7,1,3 \rrbracket$ code. 
Measuring in order the following seven stabilizers suffices for fault-tolerant error correction: 
\begin{equation} \label{e:713sevenstabilizers}
\begin{array}{r c c c c c c c}
&I&X&X&Z&Z&Y&Y\\
&X&I&X&I&X&I&X\\
&I&I&I&X&X&X&X\\
&I&Z&Z&I&I&Z&Z\\
&Z&I&Z&Z&I&Z&I\\
&I&Y&Y&Y&Y&I&I\\
&Y&I&Y&I&Y&I&Y
\end{array}
\end{equation}
As with the $\llbracket 5,1,3 \rrbracket$ code, the particular set of stabilizers and the order in which they are measured matters considerably.  It is not immediately obvious that this order works, but it can be verified by computing the syndromes for all $7 \times 3 = 21$ nontrivial one-qubit input errors as well as the results of internal faults.  

The first stabilizer in Eq.~\eqnref{e:713sevenstabilizers} mixes $X$, $Y$ and $Z$ operators.  Should this be undesirable in an experiment, the following sequence of eight measurements also allows for fault-tolerant error correction: 
\begin{equation} \label{e:713eightstabilizers}
\begin{array}{r c c c c c c c}
&I&I&I&Z&Z&Z&Z\\
&I&Z&Z&I&I&Z&Z\\
&Z&I&Z&I&Z&I&Z\\[.1cm]
&I&I&I&X&X&X&X\\
&I&X&X&I&I&X&X\\
&X&I&X&I&X&I&X\\[.1cm]
&I&Y&Y&Y&Y&I&I\\
&Y&I&Y&Y&I&Y&I
\end{array}
\end{equation}
Observe that the first six measurements are simply the standard $X$ and $Z$ stabilizer generators from Eq.~\eqnref{e:hammingcodestabilizers}.  
The last two $Y$ stabilizers are measured to prevent bad syndrome suffixes from internal faults.  

The construction of Eq.~\eqnref{e:713eightstabilizers} generalizes to the entire family of Hamming codes.  For the $2^\hammingparameter - 1$ qubit 
Hamming code, first measure the $\hammingparameter$ $Z$ and $\hammingparameter$ $X$ standard stabilizer generators.  Then make $\hammingparameter - 1$ further $Y$ measurements: measure in the $Y$ basis the first standard stabilizer generator times each of the generators $2$ through $\hammingparameter - 1$.    
This makes for $\hammingparameter + \hammingparameter + (\hammingparameter-1) = 3 \hammingparameter - 1$ stabilizer measurements total.  For example, for the $\hammingparameter = 4$, $\llbracket 15, 7, 3 \rrbracket$ Hamming code, the $11$-stabilizer sequence generalizing Eq.~\eqnref{e:713eightstabilizers} is 
\begin{equation} \label{e:1573ZXYstabilizers}
\!\!\!\!\!\!\begin{array}{r c c c c c c c c c c c c c c c}
&I&I&I&I&I&I&I&Z&Z&Z&Z&Z&Z&Z&Z\\
&I&I&I&Z&Z&Z&Z&I&I&I&I&Z&Z&Z&Z\\
&I&Z&Z&I&I&Z&Z&I&I&Z&Z&I&I&Z&Z\\
&Z&I&Z&I&Z&I&Z&I&Z&I&Z&I&Z&I&Z\\[.1cm]
&I&I&I&I&I&I&I&X&X&X&X&X&X&X&X\\
&I&I&I&X&X&X&X&I&I&I&I&X&X&X&X\\
&I&X&X&I&I&X&X&I&I&X&X&I&I&X&X\\
&X&I&X&I&X&I&X&I&X&I&X&I&X&I&X\\[.1cm]
&I&I&I&Y&Y&Y&Y&Y&Y&Y&Y&I&I&I&I\\
&I&Y&Y&I&I&Y&Y&Y&Y&I&I&Y&Y&I&I\\
&Y&I&Y&I&Y&I&Y&Y&I&Y&I&Y&I&Y&I
\end{array}
\end{equation}

If $Y$ measurements are impossible in an experiment, then by \thmref{t:distancethreenonadaptive} $2 \hammingparameter - 1$ $Z$ and $2 \hammingparameter - 1$ $X$ stabilizer measurements suffice for fault-tolerant error correction.  

Finally, the following nine measurements suffice for fault-tolerant error correction for the $\llbracket 15,7,3 \rrbracket$ code: 
\begin{equation} \label{e:1573ninestabilizers}
\!\!\!\begin{array}{r c c c c c c c c c c c c c c c}
I&X&X&Z&Z&Y&Y&X&X&I&I&Y&Y&Z&Z\\
Z&I&Z&X&Y&X&Y&I&Z&I&Z&X&Y&X&Y\\
I&I&I&I&I&I&I&X&X&X&X&X&X&X&X\\
Y&I&Y&I&Y&I&Y&Y&I&Y&I&Y&I&Y&I\\
I&Y&Y&Y&Y&I&I&I&I&Y&Y&Y&Y&I&I\\
I&I&I&X&X&X&X&I&I&I&I&X&X&X&X\\
I&Z&Z&I&I&Z&Z&I&I&Z&Z&I&I&Z&Z\\
I&Z&Z&I&I&Z&Z&Z&Z&I&I&Z&Z&I&I\\
Z&I&Z&Z&I&Z&I&X&Y&X&Y&Y&X&Y&X
\end{array}
\end{equation}
This is analogous to Eq.~\eqnref{e:713sevenstabilizers} for the $\llbracket 7,1,3 \rrbracket$ code, in the sense that it uses some weight-$12$ operators that mix $X$, $Y$ and $Z$.  
We do not know 
the least number of measurements for the $2^\hammingparameter - 1$ qubit Hamming code, in general.  

\begin{open question}
Find a minimum-length sequence of Pauli measurements for distance-three fault-tolerant error correction with the $\llbracket 2^\hammingparameter-1, 2^\hammingparameter- 1 - 2 \hammingparameter, 3 \rrbracket$ Hamming code.
\end{open question}

\section{Single-shot error correction with a \texorpdfstring{$\llbracket 16, 4, 3 \rrbracket$}{[[16,4,3]]} color code} \label{s:1643colorcode}

In this section, we show a $16$-qubit CSS code for which single-shot fault-tolerant error correction is possible, i.e., in which the stabilizer generators are each measured exactly once.  However, the code has a lower rate than the $15$-qubit Hamming code, so the more efficient error correction trades off against space efficiency.  

Consider the $16$-qubit color code in \figref{f:1643colorcode}~\cite{Reichardt18steane}.  
There is a qubit for each vertex, indexed as in the diagram, and for each shaded plaquette there is both a $Z$ stabilizer and an $X$ stabilizer on the incident qubits.  For example $Z^{\otimes 6} \otimes \identity$ and $X^{\otimes 6} \otimes \identity$ are stabilizers on the first six qubits, corresponding to a green hexagon.  This gives a $\llbracket 16, 4, 3 \rrbracket$ self-dual CSS code.  

\definecolor{lightgreen}{rgb}{.46,1,.484}
\definecolor{lightred}{rgb}{1,.494,.475}
\definecolor{lightblue}{rgb}{.478,.506,1}

\begin{figure}
\centering
\subfigure[\label{f:1643colorcode}]{
\raisebox{0cm}{\includegraphics[scale=.3]{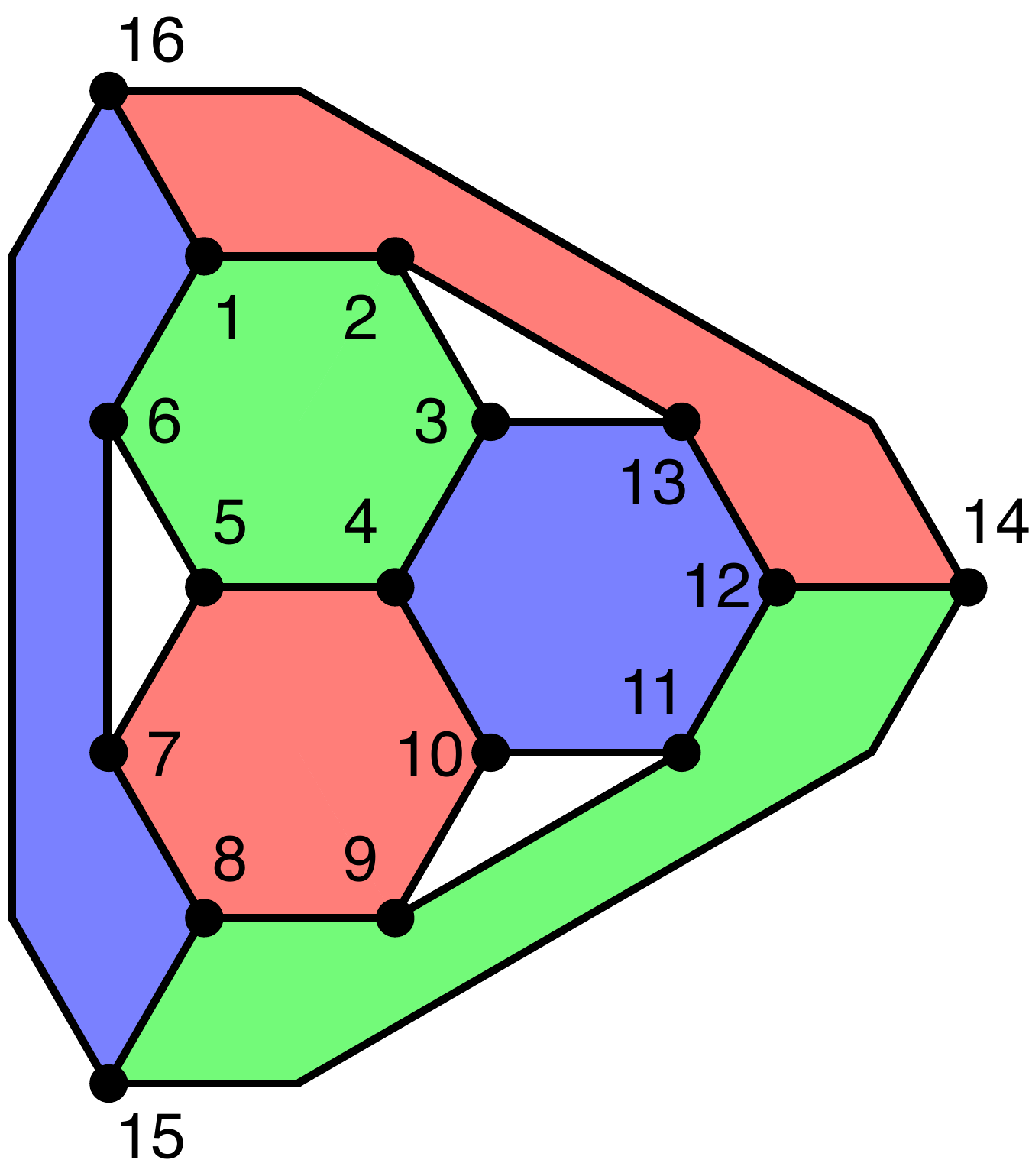}}
}
\subfigure[\label{f:1643syndromemeasurements}]{$
\begin{array}{cccccccccccccccc}
\multicolumn{6}{c}{\!\text{\sethlcolor{lightgreen}\hl{$1 \;\,\, 1 \;\, 1 \;\,\, 1 \;\, 1 \;\,\, 1$}}}&\cdot&\cdot&\cdot&\cdot&\cdot&\cdot&\cdot&\cdot&\cdot&\cdot \\
\cdot&\cdot&\cdot&\cdot&\cdot&\cdot&\cdot&\multicolumn{8}{c}{\!\text{\sethlcolor{lightgreen}\hl{$1 \;\,\, 1 \;\, \cdot \;\,\, 1 \;\, 1 \;\,\, \cdot \;\,\, 1 \;\,\, 1$}}}&\cdot \\
\cdot&\cdot&\cdot&\multicolumn{7}{c}{\text{\sethlcolor{lightred}\hl{$1 \;\,\, 1 \;\, \cdot \;\,\, 1 \;\, 1 \;\,\, 1 \;\,\, 1$}}}&\cdot&\cdot&\cdot&\cdot&\cdot&\cdot \\
\cdot&1&\cdot&1&\cdot&\cdot&\cdot&1&\cdot&\cdot&\cdot&\cdot&1&1&\cdot&1 \\
\cdot&\cdot&\cdot&1&\cdot&1&1&\cdot&\cdot&\cdot&\cdot&1&\cdot&\cdot&1&1 \\
\cdot&\cdot&\multicolumn{11}{c}{\!\text{\sethlcolor{lightblue}\hl{$1 \;\, 1 \;\, \cdot \;\,\, \cdot \;\, \cdot \;\,\, \cdot \;\, \cdot \;\,\, 1 \;\, 1 \;\,\, 1 \;\,\, 1$}}}&\cdot&\cdot&\cdot
\end{array}
$}
\subfigure[\label{f:1643colorcodeconcatenation}]{$
\begin{array}{cccccccccccccccc}
\multicolumn{6}{c}{\!\text{\sethlcolor{lightgreen}\hl{$1 \;\,\, 1 \;\, 1 \;\,\, 1 \;\, 1 \;\,\, 1$}}}&\cdot&\cdot&\cdot&\cdot&\cdot&\cdot&\cdot&\cdot&\cdot&\cdot \\
1&\cdot&\cdot&1&\cdot&\cdot&\cdot&\cdot&1&\cdot&1&\cdot&\cdot&1&1&\cdot\\
\cdot&1&\cdot&1&\cdot&\cdot&\cdot&1&\cdot&\cdot&\cdot&\cdot&1&1&\cdot&1\\
\cdot&\cdot&\cdot&\multicolumn{7}{c}{\text{\sethlcolor{lightred}\hl{$1 \;\,\, 1 \;\, \cdot \;\,\, 1 \;\, 1 \;\,\, 1 \;\,\, 1$}}}&\cdot&\cdot&\cdot&\cdot&\cdot&\cdot \\
\cdot&\cdot&\multicolumn{11}{c}{\!\text{\sethlcolor{lightblue}\hl{$1 \;\, 1 \;\, \cdot \;\,\, \cdot \;\, \cdot \;\,\, \cdot \;\, \cdot \;\,\, 1 \;\, 1 \;\,\, 1 \;\,\, 1$}}}&\cdot&\cdot&\cdot\\
\multicolumn{6}{c}{\!\text{\sethlcolor{lightgreen}\hl{$1 \;\,\, 1 \;\, 1 \;\,\, 1 \;\, 1 \;\,\, 1$}}}&\cdot&\cdot&\cdot&\cdot&\cdot&\cdot&\cdot&\cdot&\cdot&\cdot \\
\cdot&\cdot&\cdot&1&\cdot&1&1&\cdot&\cdot&\cdot&\cdot&1&\cdot&\cdot&1&1\\
\cdot&\cdot&\cdot&\multicolumn{7}{c}{\text{\sethlcolor{lightred}\hl{$1 \;\,\, 1 \;\, \cdot \;\,\, 1 \;\, 1 \;\,\, 1 \;\,\, 1$}}}&\cdot&\cdot&\cdot&\cdot&\cdot&\cdot
\end{array}
$}
\caption{(a) A $\llbracket 16,4,3 \rrbracket$ color code.  
(b) Sequence of stabilizer measurements, in both $Z$ and $X$ bases, allowing for single-shot fault-tolerant error correction (\propref{t:1643colorcode}).  Here we have written $\cdot$ in place of $0$ to draw attention to the structure.  The highlighted stabilizers correspond to plaquettes in~(a).  
(c) Sequence of stabilizer measurements, in both $Z$ and $X$ bases, allowing for \emph{concatenation} fault-tolerant error correction (\propref{t:1643colorcodeconcatenation}).} \label{}
\end{figure}

\begin{proposition} \label{t:1643colorcode}
For this $\llbracket 16, 4, 3 \rrbracket$ code, the sequence of stabilizer measurements in \figref{f:1643syndromemeasurements}, in both $Z$ and~$X$ bases, allows for fault-tolerant error correction.  
The $12$ weight-six stabilizers measured are independent and so the sequence gives single-shot error correction with no redundant stabilizer measurements.  
\end{proposition}

The proposition can be rapidly verified by noting that the columns in \figref{f:1643syndromemeasurements} are all distinct (so the code has distance three), and furthermore their suffixes are all distinct from these columns (so faults during syndrome extraction cannot be confused with input errors).  For example, the syndrome along the second column is $100100$, and its suffix $000100$ does not appear as any column.  

Although this $\llbracket 16, 4, 3 \rrbracket$ code allows for single-shot fault-tolerant error correction, it might still be preferable in practice to use a code like the $\llbracket 15, 7, 3 \rrbracket$ Hamming code.  In addition to having higher rate, the $15$-qubit Hamming code allows for error correction with only $11$ stabilizer measurements, shown in Eq.~\eqnref{e:1573ZXYstabilizers}, or $14$ stabilizers if $Y$ operators cannot be measured.  

\subsection*{Concatenation fault-tolerant error correction}

Unlike for the perfect CSS and perfect codes considered through \secref{s:hammingcodes}, for this color code, the syndrome measurement sequence in \propref{t:1643colorcode} is not enough for concatenation fault tolerance (see \secref{s:faulttolerance}).  In particular, the input error $X_1 X_2$ leads to the syndrome $000100$ if there are no internal faults.  For concatenation fault tolerance, some correction needs to be applied to restore the state to the codespace.  However, the same syndrome can arise from a perfect input if the fourth syndrome bit is incorrectly flipped, so fault tolerance requires that the correction have weight at most one.  No correction works.  

Concatenation fault tolerance is possible with two more stabilizer measurements: 

\begin{proposition} \label{t:1643colorcodeconcatenation}
For the $\llbracket 16,4,3 \rrbracket$ color code of \figref{f:1643colorcode}, the sequence of stabilizer measurements in \figref{f:1643colorcodeconcatenation}, in both $Z$ and $X$ bases, allows for concatenation fault-tolerant error correction.  
\end{proposition}

The proposition is verified by using a computer to find a consistent correction for every possible syndrome.

\section{Codes designed for fast error correction} \label{s:codesdesignedforfastec}

We have so far designed efficient and fault-tolerant syndrome-measurement schemes for existing codes.  One can also design codes to facilitate efficient fault-tolerant syndrome measurement.  To do so, let us begin by considering three simple codes; then we will generalize them.

\subsection{Base codes}

Here are the parity checks for a $[4, 1, 4]$ classical linear code (a repetition code), and the stabilizer generators for $\llbracket 8, 2, 3 \rrbracket$ and $\llbracket 8, 3, 3 \rrbracket$ quantum stabilizer codes: 
\begin{align} \label{e:basecodes}
\raisebox{.5cm}{$
\begin{gathered}
\underline{\text{$[4, 1, 4]$ code}} \\
\begin{array}{cccc}
1&1&1&1\\
0&0&1&1\\
0&1&0&1
\end{array}
\end{gathered}$}
\end{align}\begin{align*}
\begin{gathered}
\underline{\text{$\llbracket 8, 2, 3 \rrbracket$ code}} \\
\begin{array}{cccccccc}
Z&Z&Z&Z&I&I&I&I\\
X&X&X&X&I&I&I&I\\
I&I&I&I&Z&Z&Z&Z\\
I&I&I&I&X&X&X&X\\
I&X&Y&Z&I&X&Y&Z\\
I&Z&X&Y&I&Z&X&Y
\end{array}
\end{gathered}
&&
\begin{gathered}
\underline{\text{$\llbracket 8, 3, 3 \rrbracket$ code}} \\
\begin{array}{cccccccc}
Z&Z&Z&Z&Z&Z&Z&Z\\
X&X&X&X&X&X&X&X\\
I&I&Z&Y&X&Z&Y&X\\
I&Z&X&I&X&Y&Z&Y\\
I&X&I&Z&Z&X&Y&Y\\
\ 
\end{array}
\end{gathered}
\end{align*} 
This $\llbracket 8,3,3 \rrbracket$ code is equivalent to the first in a family of $\llbracket 2^\codeparameter, 2^\codeparameter - \codeparameter - 2, 3 \rrbracket$ codes 
by Gottesman~\cite{Gottesman96codes, Gottesman97thesis}.  

For the $[4,1,4]$ repetition code, one can measure 
the 
parity checks in order $1111$, $0011$, $0101$.  This already suffices for fault-tolerant error correction, because an internal fault cannot 
trigger the first parity check and therefore cannot 
be confused with an input error.  With adaptive control, the last two checks need only be measured should the first parity be odd.  This observation is not immediately interesting because the code is classical.  We will generalize it to a family of quantum codes below.  

For the $\llbracket 8,3,3 \rrbracket$ code, one can measure in order $Z^{\otimes 8}, X^{\otimes 8}, Z^{\otimes 8}$, and then the remaining three stabilizer generators---six stabilizer measurements total---and this will suffice for fault-tolerant error correction.  Indeed, with perfect stabilizer measurement the $Z^{\otimes 8}, X^{\otimes 8}$ measurements suffice to identify the type $X$, $Y$ or $Z$ of any one-qubit error, and the last three stabilizer generators localize the error.  For fault-tolerance, we measure $Z^{\otimes 8}$ a second time in order to handle the case of an $X$ or $Y$ fault occurring after the first $Z^{\otimes 8}$ measurement.  No other syndrome suffixes can be problematic; a fault occurring after the three transversal stabilizer measurements will not trigger any of them and therefore cannot be confused with an input error.  

Furthermore, with this $\llbracket 8,3,3 \rrbracket$ code, should the experimental hardware support adaptive measurements, one can first measure just $Z^{\otimes 8}$ and $X^{\otimes 8}$, and then only if one or both are nontrivial continue on to measure $Z^{\otimes 8}$ and the last three stabilizer generators.  The first two measurements suffice to detect any one-qubit input error.  

The $\llbracket 8,2,3 \rrbracket$ code above is similar to the $\llbracket 8,3,3 \rrbracket$ code.  For fault-tolerant error correction one can measure the stabilizers $Z^{\otimes 4} \otimes \identity, X^{\otimes 4} \otimes \identity, Z^{\otimes 4} \otimes \identity$, then $\identity \otimes Z^{\otimes 4}, \identity \otimes X^{\otimes 4}, \identity \otimes Z^{\otimes 4}$, 
and then the last three stabilizer generators.  Stabilizers supported on the first four qubits can potentially be measured in parallel to the stabilizers on the last four qubits.  With adaptive control, if the results of measuring 
$Z^{\otimes 4} \otimes \identity, X^{\otimes 4} \otimes \identity$ and $\identity \otimes Z^{\otimes 4}, \identity \otimes X^{\otimes 4}$ 
are trivial, then further stabilizers need not be measured.

\subsection{Generalized codes} \label{s:generalizedcodes}

Next we generalize the above base codes in order to develop families of distance-three quantum error-correcting codes with fault-tolerant error-correction procedures that are efficient, in the sense of requiring few stabilizer measurements.

\subsubsection{Generalizing the $[4, 1, 4]$ classical repetition code}

Let us start by extending the $[4, 1, 4]$ classical repetition code; the procedures for generalizing the other codes will be quite similar.  

Consider the following two parity-check matrices on $16$ and $24$ bits, respectively: 
\begin{gather*}
\begin{smallmatrix}
1&1&1&1&\cdot&\cdot&\cdot&\cdot&\cdot&\cdot&\cdot&\cdot&\cdot&\cdot&\cdot&\cdot\\
\cdot&\cdot&\cdot&\cdot&1&1&1&1&\cdot&\cdot&\cdot&\cdot&\cdot&\cdot&\cdot&\cdot\\
\cdot&\cdot&\cdot&\cdot&\cdot&\cdot&\cdot&\cdot&1&1&1&1&\cdot&\cdot&\cdot&\cdot\\
\cdot&\cdot&\cdot&\cdot&\cdot&\cdot&\cdot&\cdot&\cdot&\cdot&\cdot&\cdot&1&1&1&1\\
\cdot&\cdot&1&1&\cdot&\cdot&1&1&\cdot&\cdot&1&1&\cdot&\cdot&1&1&\\
\cdot&1&\cdot&1&\cdot&1&\cdot&1&\cdot&1&\cdot&1&\cdot&1&\cdot&1
\end{smallmatrix}
\\[.2cm]
\begin{smallmatrix}
1&1&1&1&\cdot&\cdot&\cdot&\cdot&\cdot&\cdot&\cdot&\cdot&\cdot&\cdot&\cdot&\cdot&\cdot&\cdot&\cdot&\cdot&\cdot&\cdot&\cdot&\cdot\\
\cdot&\cdot&\cdot&\cdot&1&1&1&1&\cdot&\cdot&\cdot&\cdot&\cdot&\cdot&\cdot&\cdot&\cdot&\cdot&\cdot&\cdot&\cdot&\cdot&\cdot&\cdot\\
\cdot&\cdot&\cdot&\cdot&\cdot&\cdot&\cdot&\cdot&1&1&1&1&\cdot&\cdot&\cdot&\cdot&\cdot&\cdot&\cdot&\cdot&\cdot&\cdot&\cdot&\cdot\\
\cdot&\cdot&\cdot&\cdot&\cdot&\cdot&\cdot&\cdot&\cdot&\cdot&\cdot&\cdot&1&1&1&1&\cdot&\cdot&\cdot&\cdot&\cdot&\cdot&\cdot&\cdot\\
\cdot&\cdot&\cdot&\cdot&\cdot&\cdot&\cdot&\cdot&\cdot&\cdot&\cdot&\cdot&\cdot&\cdot&\cdot&\cdot&1&1&1&1&\cdot&\cdot&\cdot&\cdot\\
\cdot&\cdot&\cdot&\cdot&\cdot&\cdot&\cdot&\cdot&\cdot&\cdot&\cdot&\cdot&\cdot&\cdot&\cdot&\cdot&\cdot&\cdot&\cdot&\cdot&1&1&1&1\\
\cdot&\cdot&1&1&\cdot&\cdot&1&1&\cdot&\cdot&1&1&\cdot&\cdot&1&1&\cdot&\cdot&1&1&\cdot&\cdot&1&1&\\
\cdot&1&\cdot&1&\cdot&1&\cdot&1&\cdot&1&\cdot&1&\cdot&1&\cdot&1&\cdot&1&\cdot&1&\cdot&1&\cdot&1
\end{smallmatrix}
\end{gather*}
We have written $\cdot$ in place of $0$ to draw attention to the structure.  The bits are divided into blocks of four, with $1111$ parity checks, and the last two parity checks have the same form, $0011$ or $0101$, on each block.  

The above parity checks define self-orthogonal $[16, 10, 4]$ and $[24, 16, 4]$ classical linear codes.  
By the CSS construction, using the same parity checks in both the $Z$ and $X$ bases, they induce $\llbracket 16, 4, 4 \rrbracket$ and $\llbracket 24, 8, 4 \rrbracket$ self-dual quantum stabilizer codes.  
These quantum codes can also be constructed by concatenating two copies of the $\llbracket 2\codeparameter, 2\codeparameter - 2, 2 \rrbracket$ code with the $\llbracket 4,2,2 \rrbracket$ code.  
The codes can be extended by adding eight qubits at a time, two blocks of four. 
For $\codeparameter \geq 2$, this defines $\llbracket 8 \codeparameter, 4 (\codeparameter - 1), 4 \rrbracket$ self-dual CSS codes.  

These codes are potentially of interest for a variety of reasons, e.g., the $\llbracket 16, 4, 4 \rrbracket$ code isn't too far off in terms of rate from the perfect CSS $\llbracket 15, 7, 3 \rrbracket$ Hamming code, yet it has higher distance. 

For us they are of interest because they allow for distance-three fault-tolerant error correction with \emph{single-shot} nonadaptive 
stabilizer measurements. 
As for the $[4, 1, 4]$ classical code described above, it suffices to measure the $Z^{\otimes 4}$ and $X^{\otimes 4}$ stabilizers on each block, and then (if some block's stabilizer measurements are nontrivial) measure the last two parity checks in the $Z$ and/or $X$ bases: $(IIZZ)^{\otimes \codeparameter}, (IZIZ)^{\otimes \codeparameter}, (IIXX)^{\otimes \codeparameter}, (IXIX)^{\otimes \codeparameter}$.  No redundant syndrome information needs to be measured for distance-three fault-tolerant error correction.  An $X$ fault, for example, occurring in a block after the $ZZZZ$ measurement necessarily leads to a syndrome different from that caused by any weight-one $X$ input error (which always triggers some $ZZZZ$ stabilizer).  

\begin{theorem} \label{t:generalizing414}
Each code in this family of $\llbracket 8 \codeparameter, 4 (\codeparameter - 1), 4 \rrbracket$ self-dual CSS codes, for $\codeparameter \geq 2$, allows for single-shot \emph{distance-three} fault-tolerant error correction, and sub-single-shot with adaptive measurements. 
\begin{itemize}[leftmargin=*]
\item Nonadaptive case: $4 (\codeparameter + 1)$ measurements, depth $\codeparameter + 4$.  
\item Adaptive case: $4 \codeparameter$ or $4 \codeparameter + 2$ measurements, depth $\codeparameter + 2$.  
\end{itemize}
\end{theorem}

Observe that measurements on different blocks of four qubits can be implemented in parallel.  Thus while $4 (\codeparameter + 1)$ stabilizer measurements are needed ($12$ measurements for the $\llbracket 16, 4, 4 \rrbracket$ code, most comparable to the $\llbracket 15, 7, 3 \rrbracket$ Hamming code), these measurements can be implemented in only six rounds.  

\medskip

Distance-three fault-tolerant error correction can handle up to one input error or internal fault.  Although the codes have distance four, the above single-shot error-correction procedure is \emph{not} fault tolerant to distance four.  
Distance-four fault tolerance requires that two faults causing an output error of weight at least two should be distinguishable from one or zero faults.  
For example, with the $\llbracket 16, 4, 4 \rrbracket$ code, an input $X_2$ error and an $X_6$ fault after the first round of $Z$ stabilizer measurements gives the same syndrome ($100000$) as an input $X_1$ error.  With this correction applied, $X_1 X_2 X_6$ is equivalent to a logical error times $X_5$.  For distance-four fault tolerance, it suffices to append the parity checks 
\begin{equation*}
\begin{array}{cccccccccccccccc}
\cdot&1&1&\cdot&\cdot&1&1&\cdot&\cdot&1&1&\cdot&\cdot&1&1&\cdot \\
1&1&1&1&1&1&1&1&1&1&1&1&1&1&1&1
\end{array}
\end{equation*}
This then takes ten rounds of stabilizer measurements.

\subsubsection{Generalizing the $[8, 4, 4]$ and $[16, 11, 4]$ classical codes} \label{s:generalizing844and16114codes}

The above procedure defined a family of quantum error-correcting codes based on the classical $[4, 1, 4]$ code.  We can similarly define families of quantum codes with blocks of size $8$, $16$, or larger powers of two.  

Start, for example, with $8$- and $16$-bit Reed-Muller codes defined by the following parity checks: 
\begin{align*}
\begin{gathered}
\underline{\text{$[8, 4, 4]$ code}} \\
\begin{array}{cccccccc}
1&1&1&1&1&1&1&1\\
0&0&0&0&1&1&1&1\\
0&0&1&1&0&0&1&1\\
0&1&0&1&0&1&0&1
\end{array}
\end{gathered}
&&
\raisebox{.3cm}{
$\begin{gathered}
\underline{\text{$[16, 11, 4]$ code}} \\
\begin{smallmatrix}
1&1&1&1&1&1&1&1&1&1&1&1&1&1&1&1\\
\cdot&\cdot&\cdot&\cdot&\cdot&\cdot&\cdot&\cdot&1&1&1&1&1&1&1&1\\
\cdot&\cdot&\cdot&\cdot&1&1&1&1&\cdot&\cdot&\cdot&\cdot&1&1&1&1\\
\cdot&\cdot&1&1&\cdot&\cdot&1&1&\cdot&\cdot&1&1&\cdot&\cdot&1&1\\
\cdot&1&\cdot&1&\cdot&1&\cdot&1&\cdot&1&\cdot&1&\cdot&1&\cdot&1
\end{smallmatrix}
\end{gathered}$
}
\end{align*}

The first, $[8, 4, 4]$, code can be used to define a family of $\llbracket 8\codeparameter, 6(\codeparameter - 1), 4 \rrbracket$ self-dual CSS quantum codes, for $\codeparameter \geq 2$, by putting a separate copy of the first parity check $11111111$ on each block of eight qubits (in both $Z$ and $X$ bases), while copying the other three parity checks across to be the same on each block.  For example, the $\llbracket 16, 6, 4 \rrbracket$ code's $Z$ and $X$ basis parity checks are each 
\begin{equation} \label{e:1664codeparitychecks}
\begin{array}{cccccccccccccccc}
1&1&1&1&1&1&1&1&0&0&0&0&0&0&0&0\\
0&0&0&0&0&0&0&0&1&1&1&1&1&1&1&1\\
0&0&0&0&1&1&1&1&0&0&0&0&1&1&1&1\\
0&0&1&1&0&0&1&1&0&0&1&1&0&0&1&1\\
0&1&0&1&0&1&0&1&0&1&0&1&0&1&0&1
\end{array}
\end{equation}
(Note that these five parity checks are equivalent to those of the above $[16,11,4]$ code.)  

The $[16, 11, 4]$ code above can similarly be used to define a family of $\llbracket 16\codeparameter, 14 \codeparameter - 8, 4 \rrbracket$ self-dual CSS codes, for $\codeparameter \geq 1$, by putting a $1^{16}$ parity check on each block of $16$ qubits, while copying the other four parity checks across each block.  

\begin{theorem} \label{t:generalizing844}
Each of these $\llbracket 8\codeparameter, 6(\codeparameter - 1), 4 \rrbracket$ and $\llbracket 16\codeparameter, 14 \codeparameter - 8, 4 \rrbracket$ codes allows for single-shot \emph{distance-three} fault-tolerant error correction.  By measuring disjoint qubit blocks in parallel, eight measurement rounds suffice for the $\llbracket 8\codeparameter, 6(\codeparameter - 1), 4 \rrbracket$ codes, and ten measurement rounds for the $\llbracket 16\codeparameter, 14 \codeparameter - 8, 4 \rrbracket$ codes.  
\end{theorem}

Compared to \thmref{t:generalizing414}, these codes achieve a higher encoding rate, trading off slower error correction.  

Some particularly interesting codes in these families are $\llbracket 16, 6, 4 \rrbracket$, $\llbracket 24, 12, 8 \rrbracket$ and $\llbracket 32, 20, 4 \rrbracket$.  They can be compared to the $\llbracket 15, 7, 3 \rrbracket$ and $\llbracket 31, 21, 3 \rrbracket$ Hamming codes.  The rates are similar, but the new codes have a higher distance and allow for faster fault-tolerant error correction (cf.~\propref{t:hammingcodesnonadaptiveerrorcorrection}).  The stabilizers measured also have the same weights as those of the closest Hamming codes: weight-$8$ stabilizers for the $\llbracket 16, 6, 4 \rrbracket$ code, weights 8 or $12$ stabilizers for the $\llbracket 24, 12, 4 \rrbracket$ code, and weight-$16$ stabilizers for the $\llbracket 32, 20, 4 \rrbracket$ code.  

The $\llbracket 15, 7, 3 \rrbracket$ code can be obtained by puncturing the $\llbracket 16, 6, 4 \rrbracket$ code.  
This raises the following question: 

\begin{open question}
Let $C$ be a linear code, and $C'$ be obtained by puncturing one coordinate of $C$.  
Given a sequence of $m$ parity-check measurements for fault-tolerant error correction with $C$, is there a sequence of $m + \ell$ parity-check measurements for fault-tolerant error correction with $C'$, where $\ell$ is a constant independent of~$C$?
\end{open question}

Like \thmref{t:generalizing414}, \thmref{t:generalizing844} is for distance-three fault tolerance; the error-correction procedures tolerate up to one fault either on the input or during error correction.  For these codes, single-shot stabilizer measurements do not suffice for distance-four fault-tolerant error correction.  But once again, only two more measurement rounds, in both $Z$ and $X$ bases, are needed to upgrade the protection; appending the parity checks 
\begin{equation*}
\begin{array}{cccccccccccccccc}
\cdot&1&1&\cdot&1&\cdot&\cdot&1&\cdot&1&1&\cdot&1&\cdot&\cdot&1 \\
1&1&1&1&1&1&1&1&1&1&1&1&1&1&1&1
\end{array}
\end{equation*}
to those of Eq.~\eqnref{e:1664codeparitychecks} gives distance-four fault-tolerant error correction.  

\medskip

The single-shot sequences for distance-three fault tolerance obtained in Theorems~\ref{t:generalizing414} and~\ref{t:generalizing844} rely on two subfamilies of Reed-Muller codes: the repetition codes and the extended Hamming codes.  
The structure of Reed-Muller codes might help to design short fault-tolerant sequences for higher distance.  

\begin{open question}
Find a minimum-length sequence of parity check measurements for distance-seven fault-tolerant error correction with the family of distance-eight Reed-Muller codes.  
\end{open question}

This question naturally generalizes to the whole family of Reed-Muller codes and to punctured Reed-Muller codes, with potential applications to magic state distillation~\cite{BravyiKitaev04magic, BravyiHaah2012magic}.

\subsubsection{Generalizing the \texorpdfstring{$\llbracket 8,2,3 \rrbracket$}{[[8,2,3]]} and \texorpdfstring{$\llbracket 8,3,3 \rrbracket$}{[[8,3,3]]} codes}

Similar to how we extended the $[4, 1, 4]$, $[8, 4, 4]$ and $[16, 11, 4]$ codes by adding on more blocks, the $\llbracket 8,2,3 \rrbracket$ and $\llbracket 8,3,3 \rrbracket$ codes of Eq.~\eqnref{e:basecodes} can be extended.  By adding either two blocks of four qubits to the former code, or one block of eight qubits to the latter code, we obtain $16$-qubit codes, with respective stabilizers: 
\begin{align*}
\begin{gathered}
\underline{\text{$\llbracket 16, 6, 3 \rrbracket$ code}} \\
\begin{array}{cccccccccccccccc}
Z&Z&Z&Z&I&I&I&I&I&I&I&I&I&I&I&I\\
X&X&X&X&I&I&I&I&I&I&I&I&I&I&I&I\\
I&I&I&I&Z&Z&Z&Z&I&I&I&I&I&I&I&I\\
I&I&I&I&X&X&X&X&I&I&I&I&I&I&I&I\\
I&I&I&I&I&I&I&I&Z&Z&Z&Z&I&I&I&I\\
I&I&I&I&I&I&I&I&X&X&X&X&I&I&I&I\\
I&I&I&I&I&I&I&I&I&I&I&I&Z&Z&Z&Z\\
I&I&I&I&I&I&I&I&I&I&I&I&X&X&X&X\\
I&X&Y&Z&I&X&Y&Z&I&X&Y&Z&I&X&Y&Z\\
I&Z&X&Y&I&Z&X&Y&I&Z&X&Y&I&Z&X&Y
\end{array}
\end{gathered} \\[.3cm]
\begin{gathered}
\underline{\text{$\llbracket 16, 9, 3 \rrbracket$ code}} \\
\begin{array}{cccccccccccccccc}
Z&Z&Z&Z&Z&Z&Z&Z&I&I&I&I&I&I&I&I\\
X&X&X&X&X&X&X&X&I&I&I&I&I&I&I&I\\
I&I&I&I&I&I&I&I&Z&Z&Z&Z&Z&Z&Z&Z\\
I&I&I&I&I&I&I&I&X&X&X&X&X&X&X&X\\
I&I&Z&Y&X&Z&Y&X&I&I&Z&Y&X&Z&Y&X\\
I&Z&X&I&X&Y&Z&Y&I&Z&X&I&X&Y&Z&Y\\
I&X&I&Z&Z&X&Y&Y&I&X&I&Z&Z&X&Y&Y\\
\ 
\end{array}
\end{gathered}
\end{align*}
By adding more blocks, these procedures yield families of $\llbracket 8 \codeparameter, 4 \codeparameter - 2, 3 \rrbracket$ and $\llbracket 8 \codeparameter, 6 \codeparameter - 3, 3 \rrbracket$ codes, respectively.  While the second code family has higher rate, its stabilizers also have higher weight.  This tradeoff can be continued by applying the extension procedure to the other $\llbracket 2^\codeparameter, 2^\codeparameter - \codeparameter - 2, 3 \rrbracket$ codes defined by Gottesman~\cite{Gottesman96codes}.   

These codes do not allow for single-shot fault-tolerant error correction.  However, they require only one round of measuring redundant syndrome bits, in parallel.  After measuring the $Z$ and $X$ parity checks on each block of four or eight qubits, one can measure the $Z$ parity checks a second time.  Then finally measure the block-crossing stabilizers.  This procedure is fault tolerant for essentially the same reason that it works for the $\llbracket 8,3,3 \rrbracket$ code, described above.  (Faults occurring after the block stabilizer measurements will not trigger any of them, and therefore cannot be confused with an input error.)  We conclude: 

\begin{theorem}
For these $\llbracket 8 \codeparameter, 4 \codeparameter - 2, 3 \rrbracket$ and $\llbracket 8 \codeparameter, 6 \codeparameter - 3, 3 \rrbracket$ codes, fault-tolerant error correction is possible with five or six rounds, respectively, of stabilizer measurements.  
\end{theorem}


\section{Error correction for higher-distance codes}

So far, our examples and general constructions have been for distance-three and -four codes.  In this section, we will consider some distance-five and -seven CSS codes, and present for them nonadaptive stabilizer measurement sequences for fault-tolerant error correction.

\subsection{\texorpdfstring{$\llbracket 17, 1, 5 \rrbracket$}{[[17,1,5]]} color code}

Begin by considering the following $\llbracket 17, 1, 5 \rrbracket$ color code: 
\begin{equation*} 
\raisebox{-2.5cm}{\includegraphics[scale=.3]{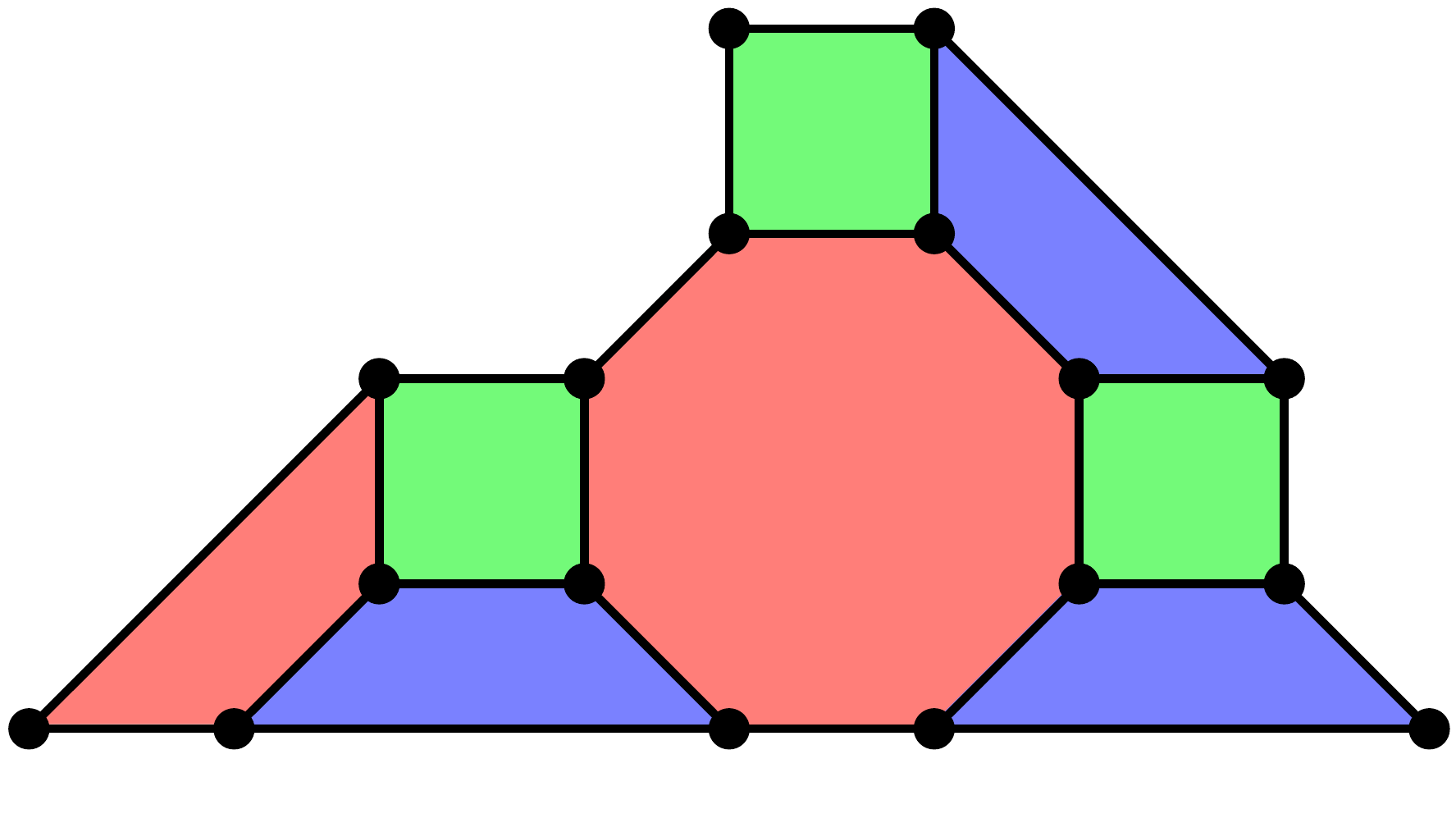}}
\end{equation*}
Fault-tolerant $Z$ error correction can be accomplished with nine rounds of fault-tolerantly measuring $X$ \mbox{plaquette} stabilizers, $20$ stabilizer measurements total, in the following order: 
\begin{equation*}
\begin{split}
\includegraphics[scale=.08]{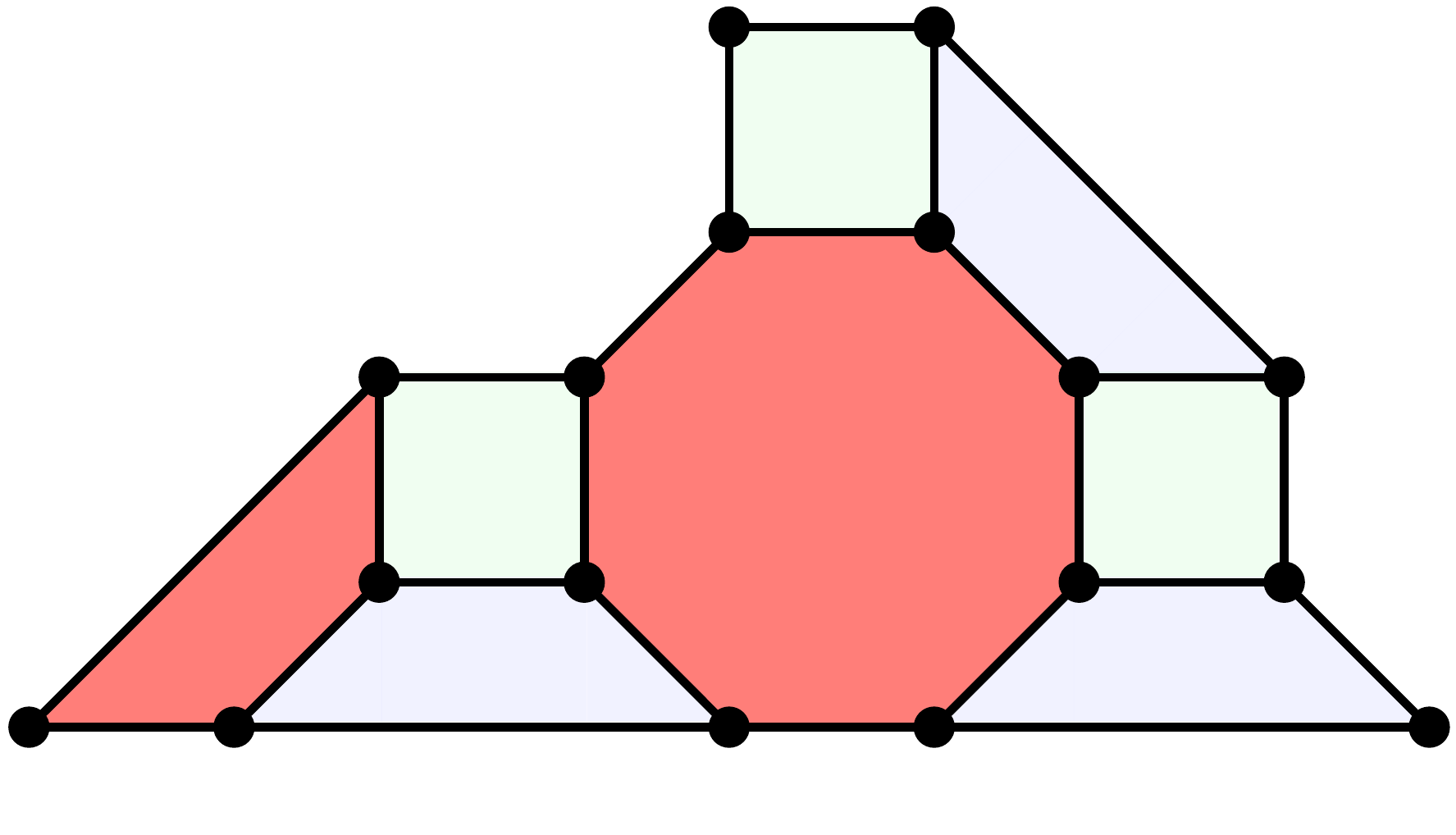}\place{\footnotesize 1}{-68mu}{18pt}\;\;
\includegraphics[scale=.08]{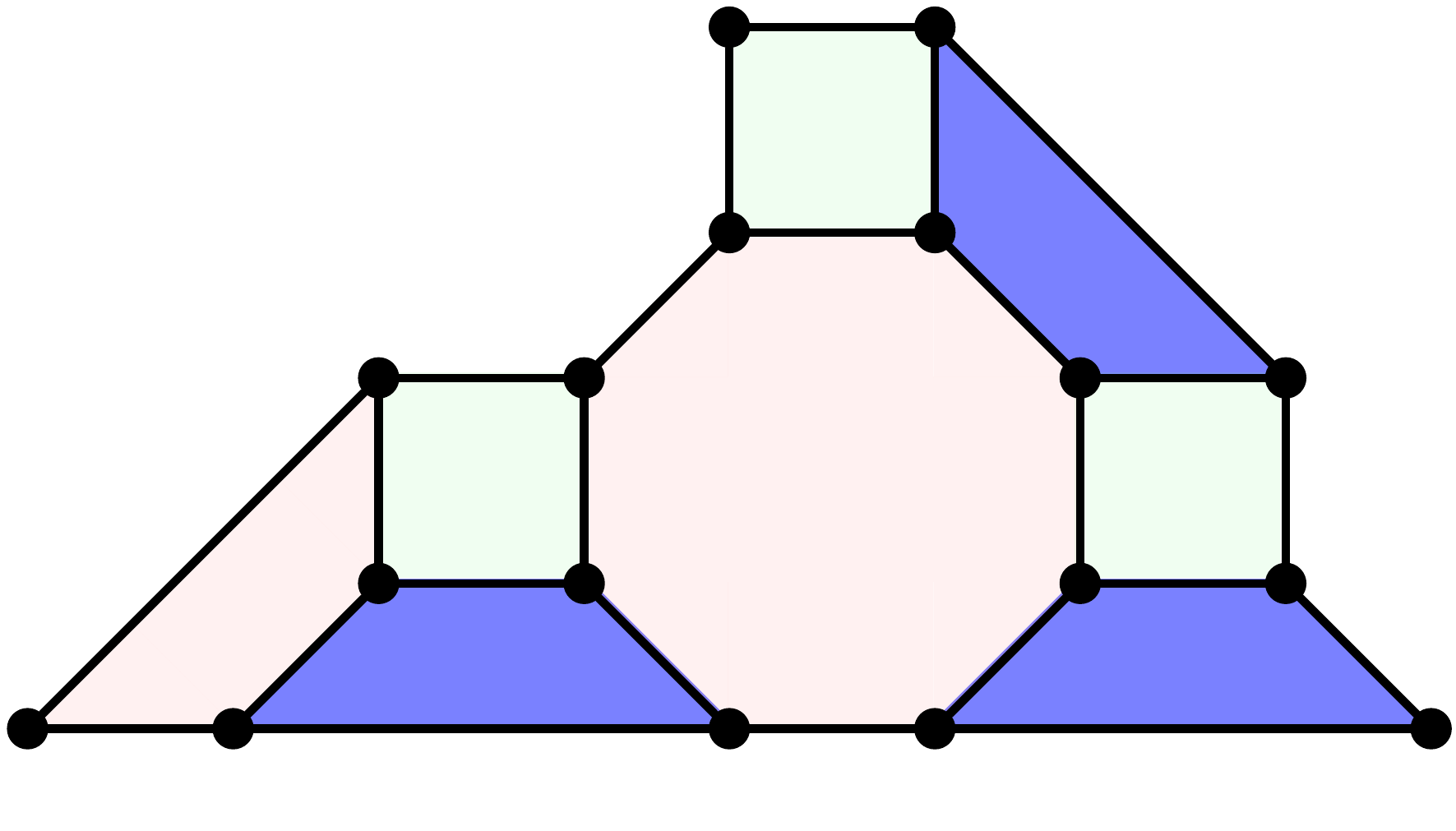}\place{\footnotesize 2}{-68mu}{18pt}\;\;
\includegraphics[scale=.08]{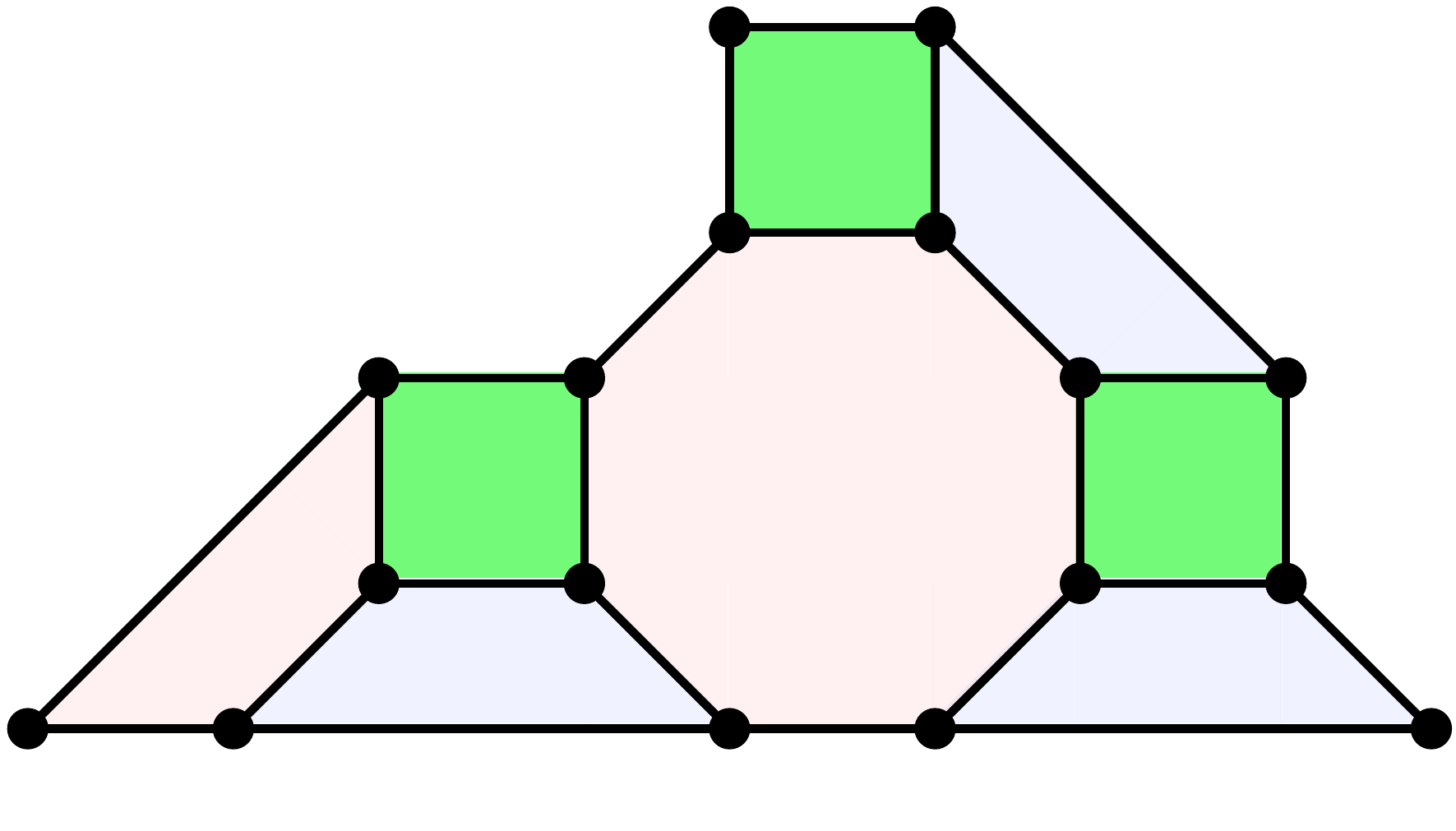}\place{\footnotesize 3}{-68mu}{18pt}\;\;
\includegraphics[scale=.08]{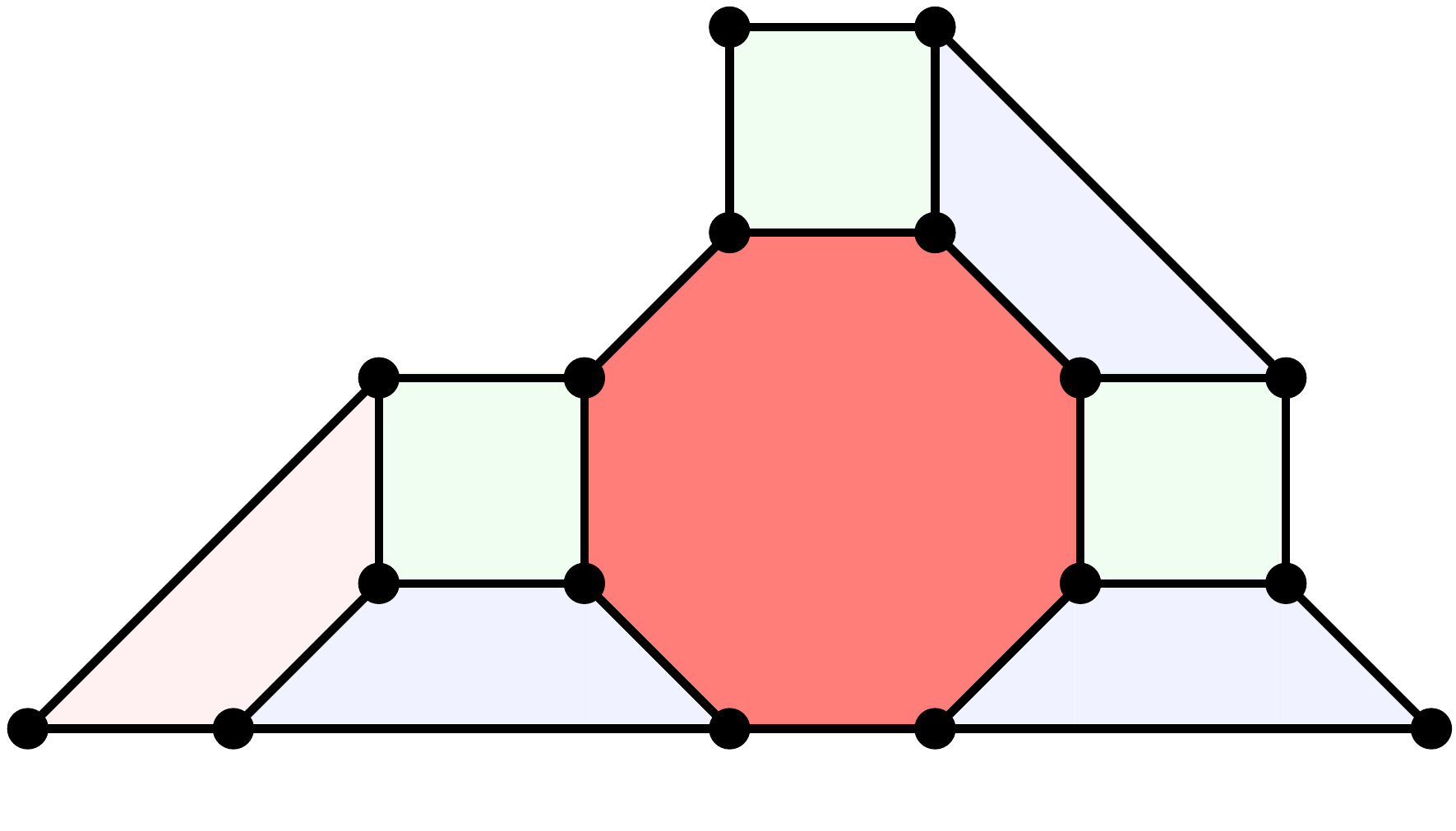}\place{\footnotesize 4}{-68mu}{18pt}\;\;
\includegraphics[scale=.08]{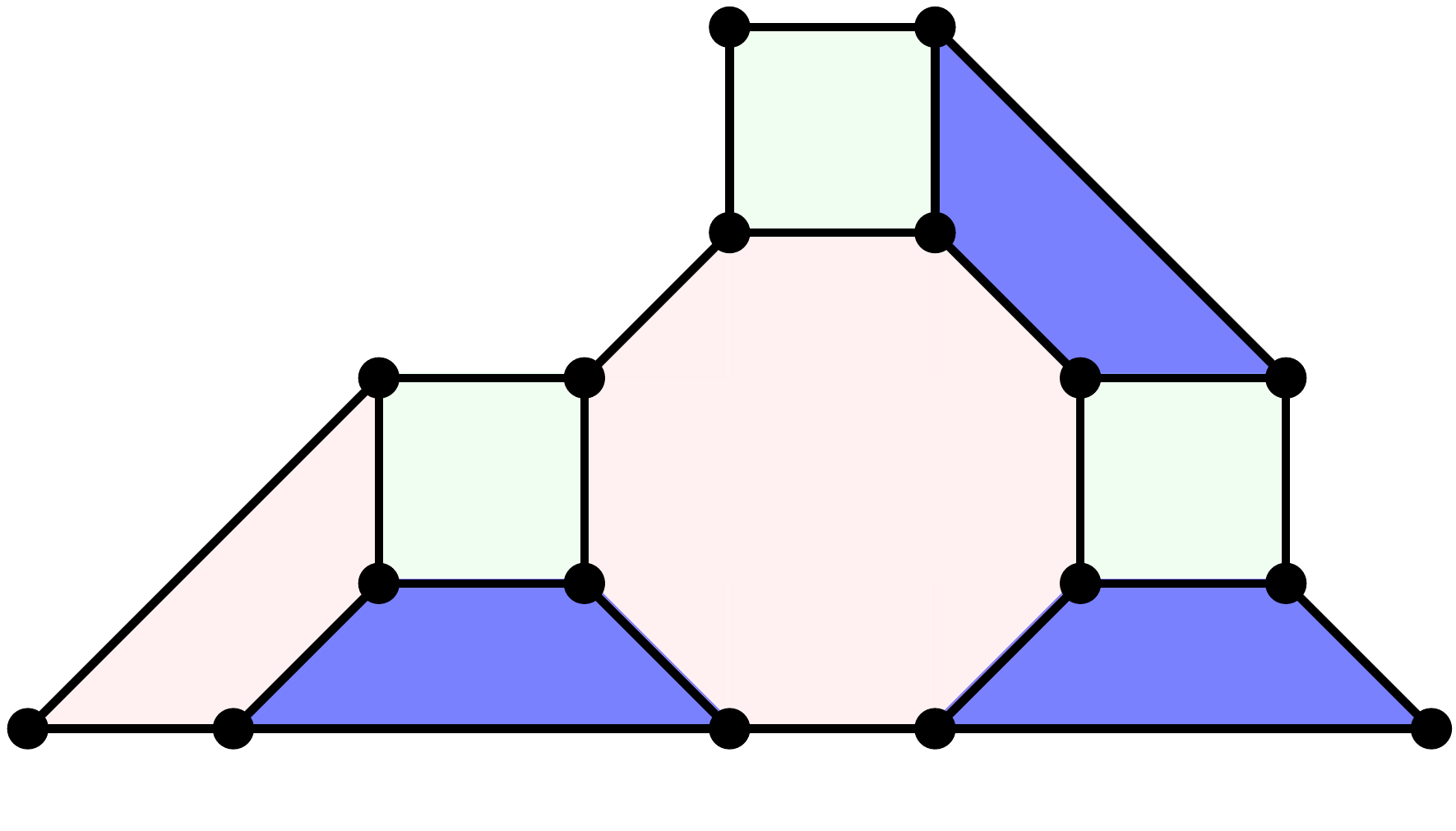}\place{\footnotesize 5}{-68mu}{18pt} \\
\includegraphics[scale=.08]{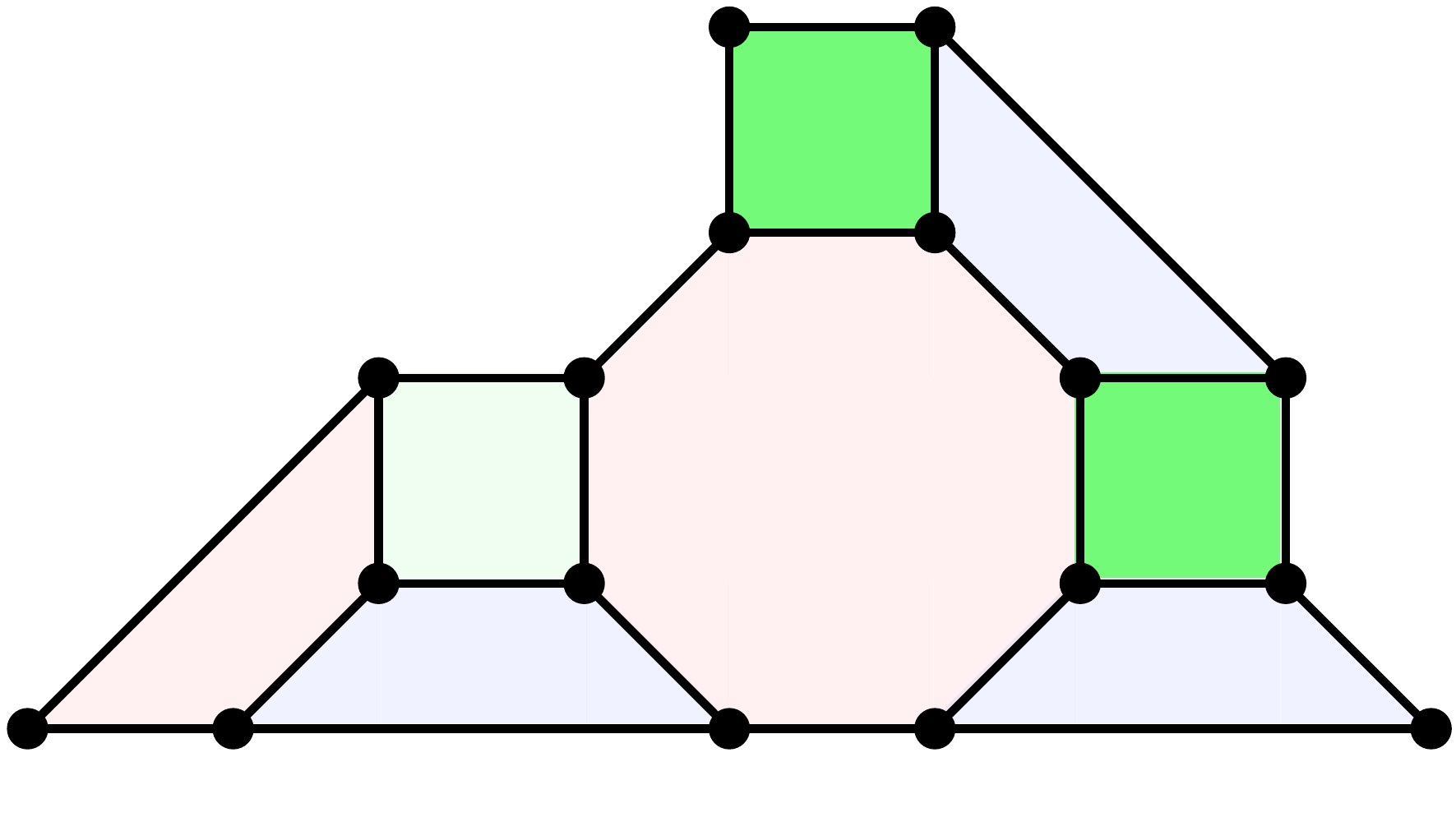}\place{\footnotesize 6}{-68mu}{18pt}\;\;
\includegraphics[scale=.08]{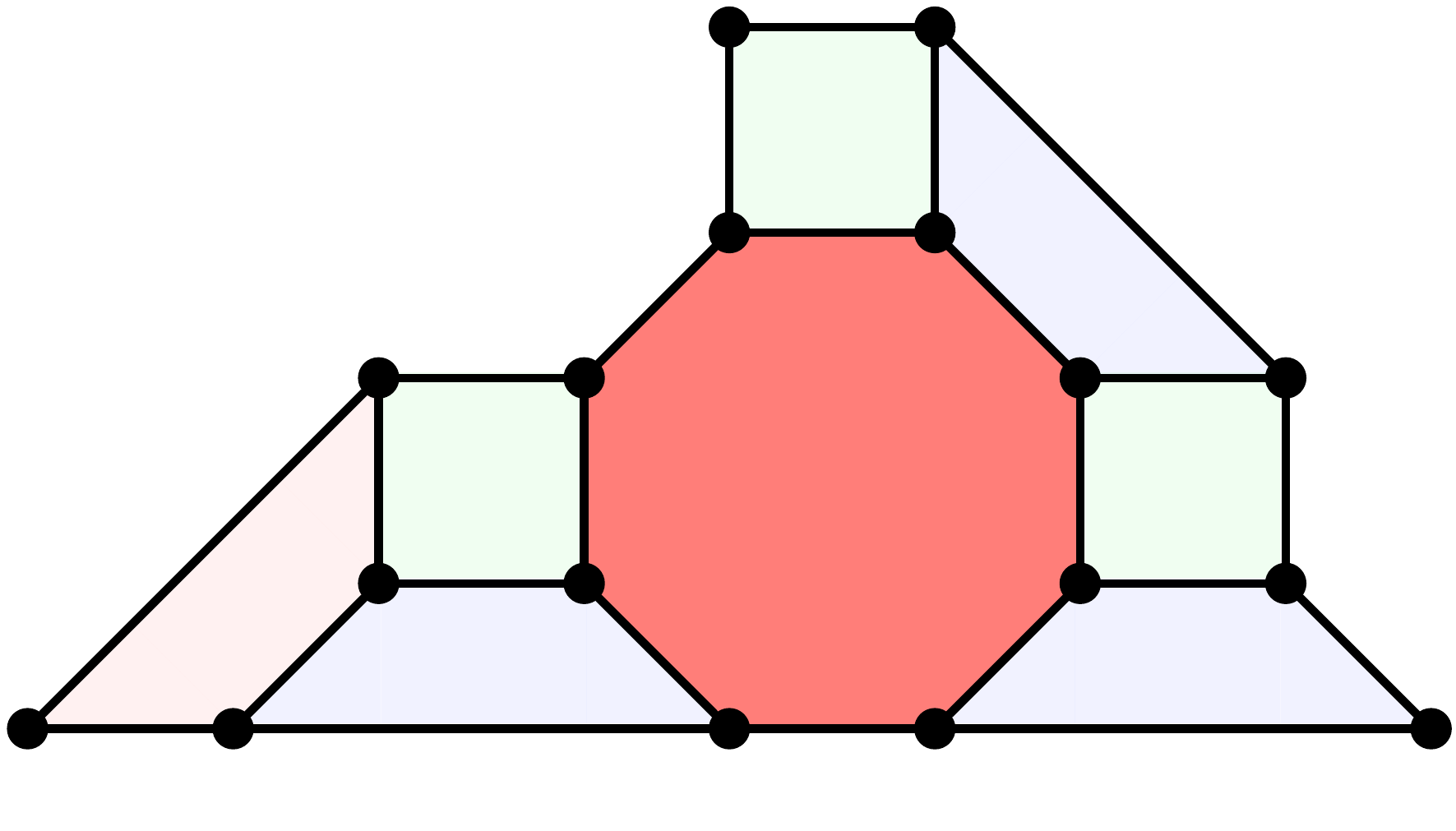}\place{\footnotesize 7}{-68mu}{18pt}\;\;
\includegraphics[scale=.08]{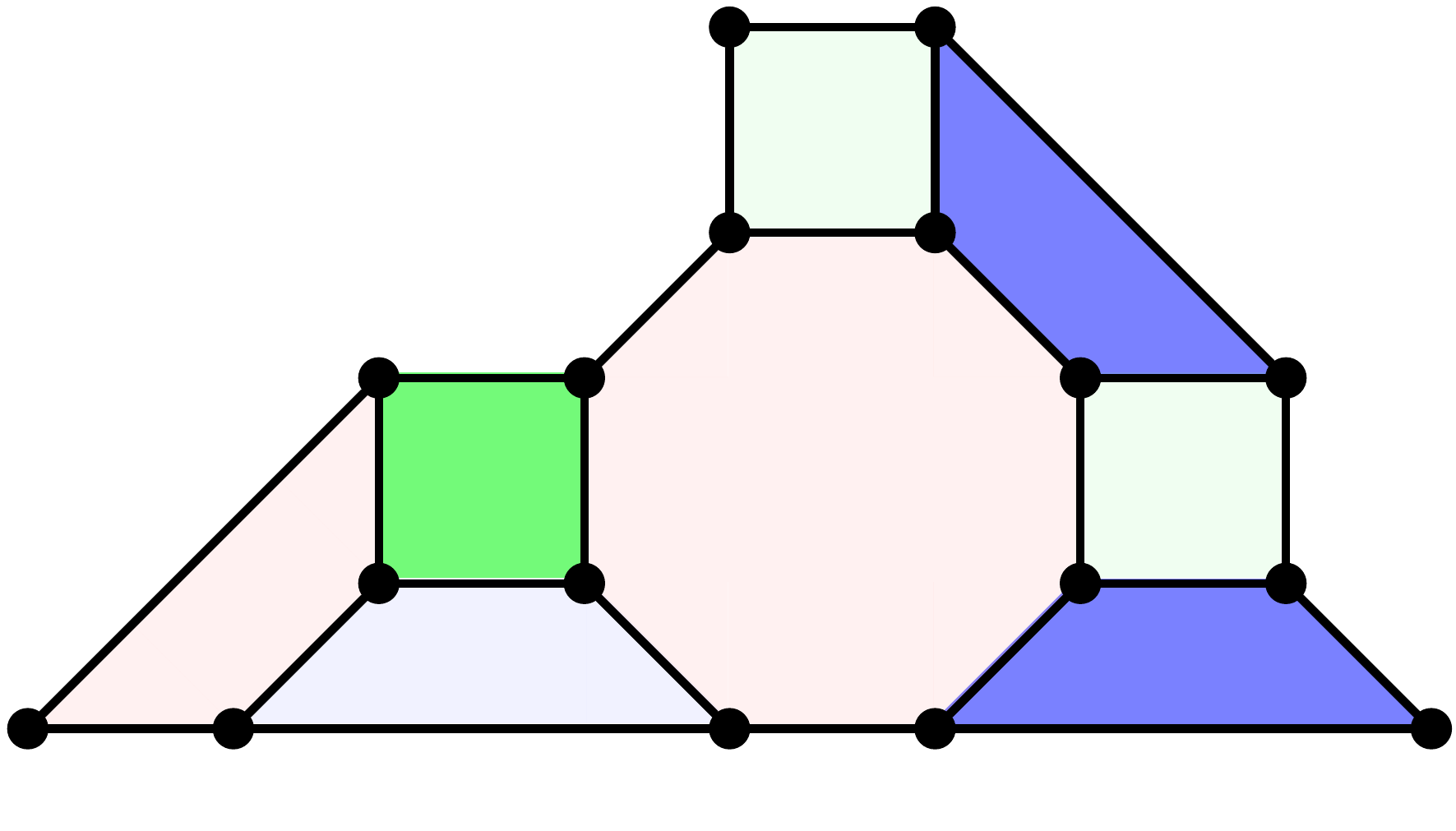}\place{\footnotesize 8}{-68mu}{18pt}\;\;
\includegraphics[scale=.08]{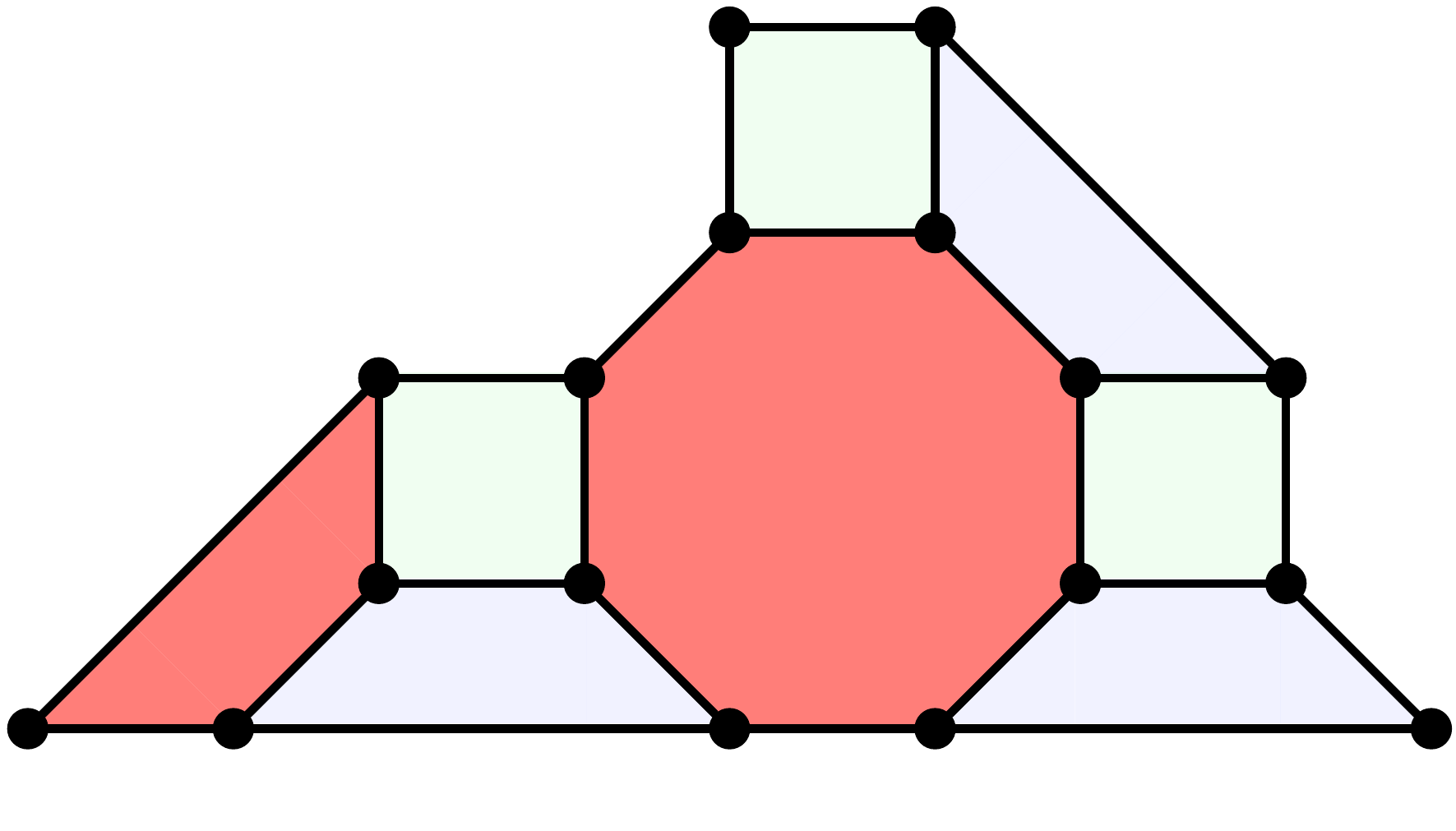}\place{\footnotesize 9}{-68mu}{18pt} \hspace{.8cm}
\end{split}
\end{equation*}
Here, we have highlighted the plaquette stabilizers that should be measured in each round.  

There could well be more-efficient syndrome measurement sequences.  We have verified the fault tolerance of this one by a computer enumeration over all possible combinations of up to two input errors or internal faults. 

Topological codes like the color and surface codes have the advantage that their natural stabilizer generators are geometrically local for qubits embedded in a two-dimensional surface~\cite{DennisKitaevLandahlPreskill01topological, BombinMartindelgado06colorcode, FowlerMariantoniMartinisCleland12surfacecodes}. For error correction, it may therefore be preferable to measure a sequence of only these stabilizer generators, and not measure any nontrivial linear combination of generators.  The above measurement sequence satisfies this property, while our measurement sequence for the $\llbracket 16,4,4 \rrbracket$ color code, in \secref{s:1643colorcode}, does not.  

Note that our notion of fault tolerance is different from the circuit-level noise model of~\cite{DennisKitaevLandahlPreskill01topological}.  However, the idea of measuring just a subset of the plaquettes at each round is relevant to both settings.  It could improve the performance of topological codes, since syndrome extraction is often the dominant noise source.


\subsection{\texorpdfstring{$\llbracket d^2, 1, d \rrbracket$}{[[17,1,5]]} surface codes}

For odd $d \geq 3$, there are $\llbracket d^2, 1, d \rrbracket$ surface codes~\cite{BravyiKitaev98surfacecode, TomitaSvore14surfacecodes}, illustrated below for $d = 3$ and $d = 5$: 
\begin{equation*}
\raisebox{-2.5cm}{\includegraphics[scale=.3]{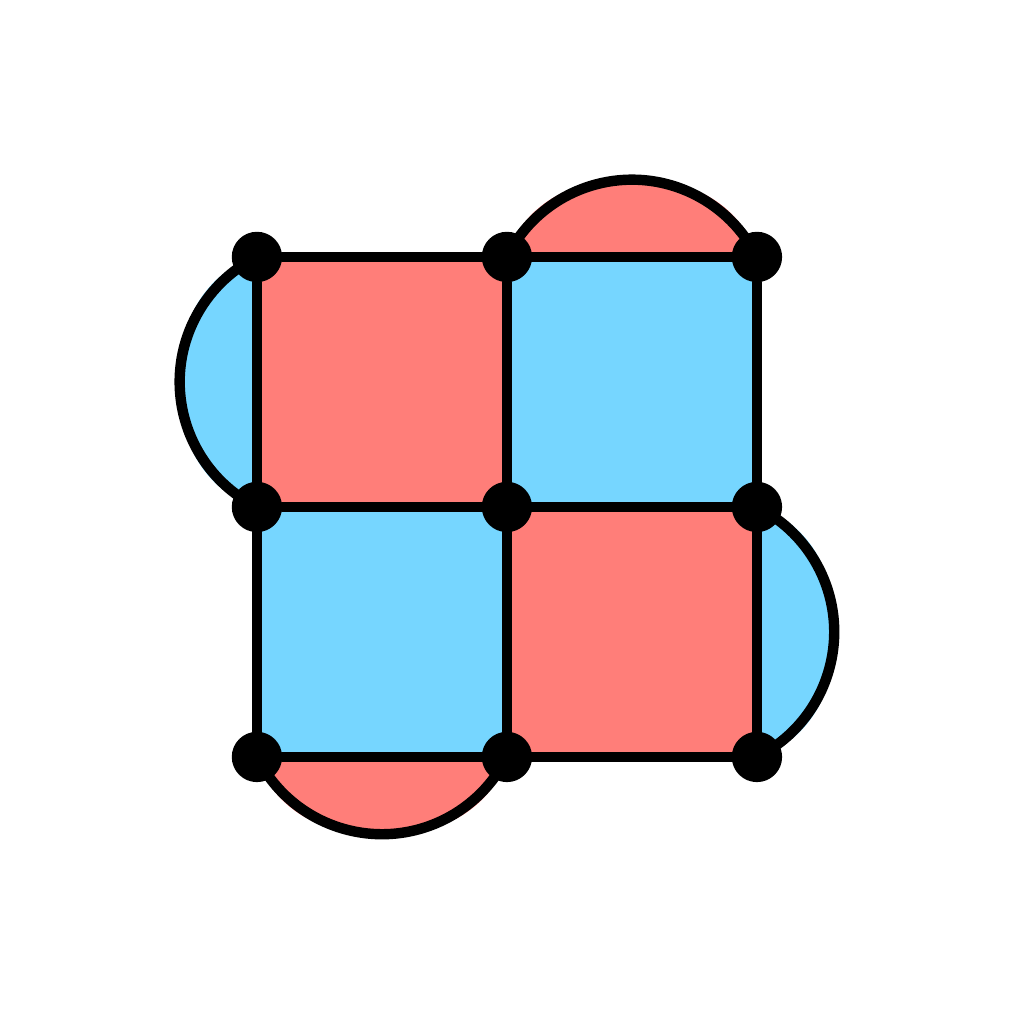}}
\raisebox{-2.5cm}{\includegraphics[scale=.3]{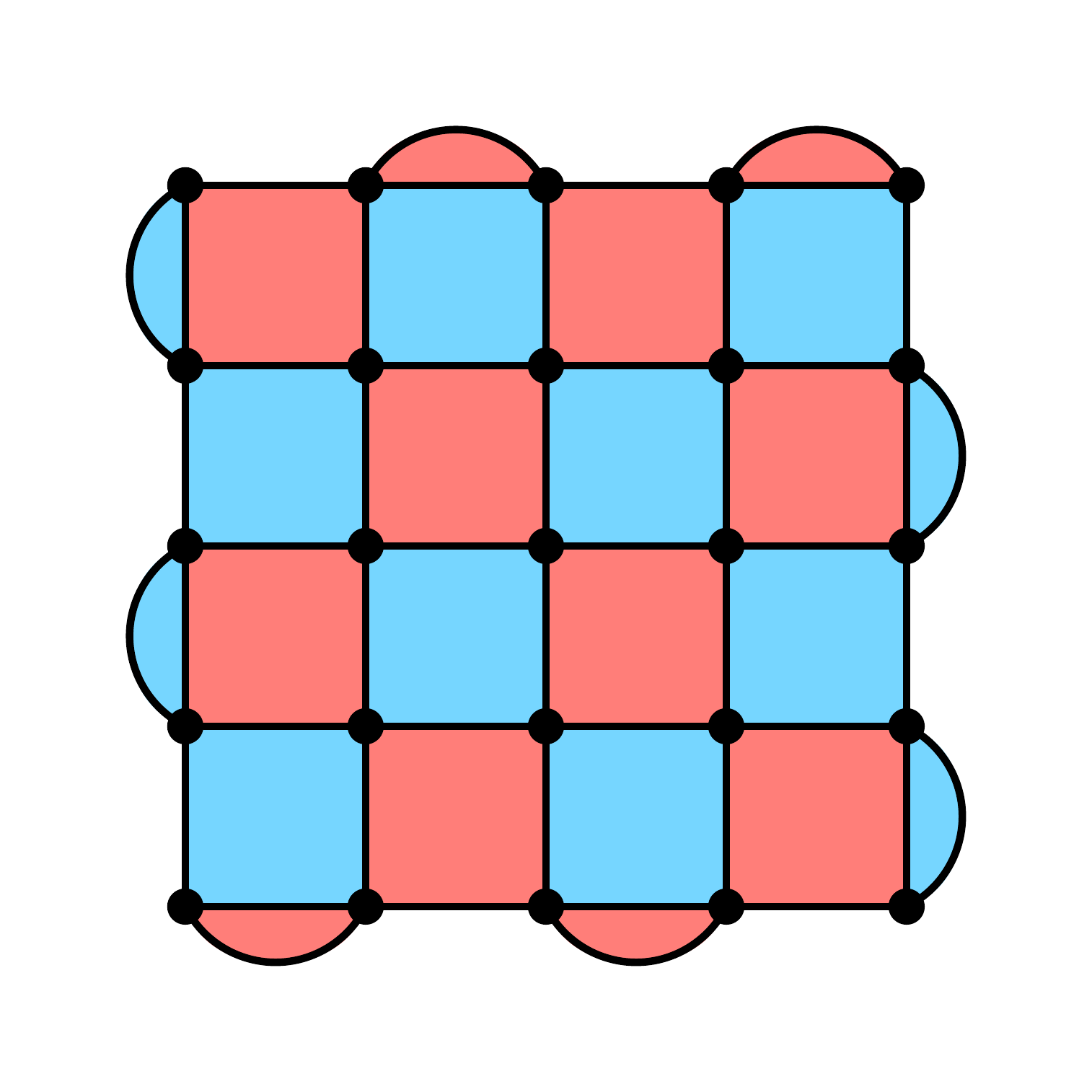}}
\end{equation*}
Here the qubits are placed at the vertices.  Red plaquettes correspond to $Z$ stabilizers on the involved qubits, and blue plaquettes to $X$ stabilizers.  (The codes are CSS, but not self dual.)  

For the $\llbracket 9,1,3 \rrbracket$ code, six stabilizer measurements, applied in three rounds, suffice for fault-tolerant $X$ error correction: 
\begin{equation*}
\begin{split}
\includegraphics[scale=.2]{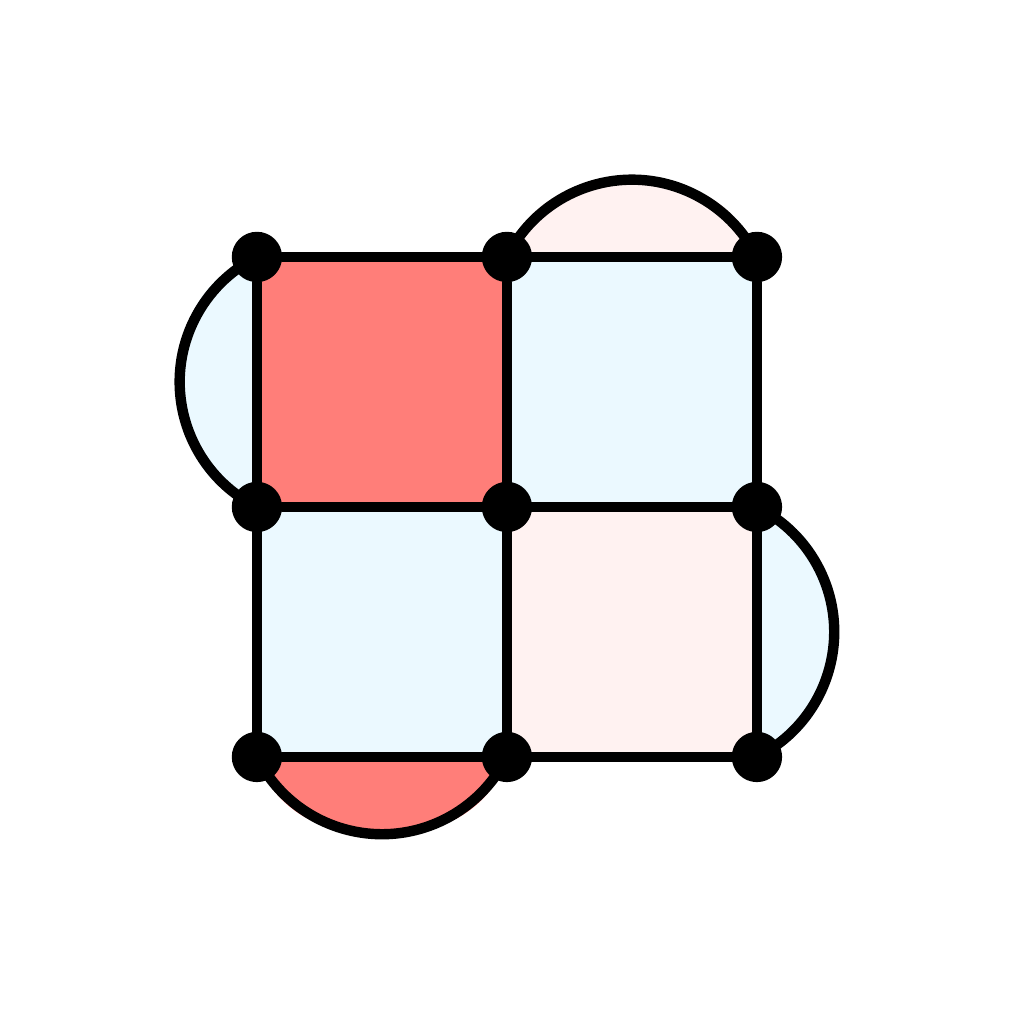}
\;\;
\includegraphics[scale=.2]{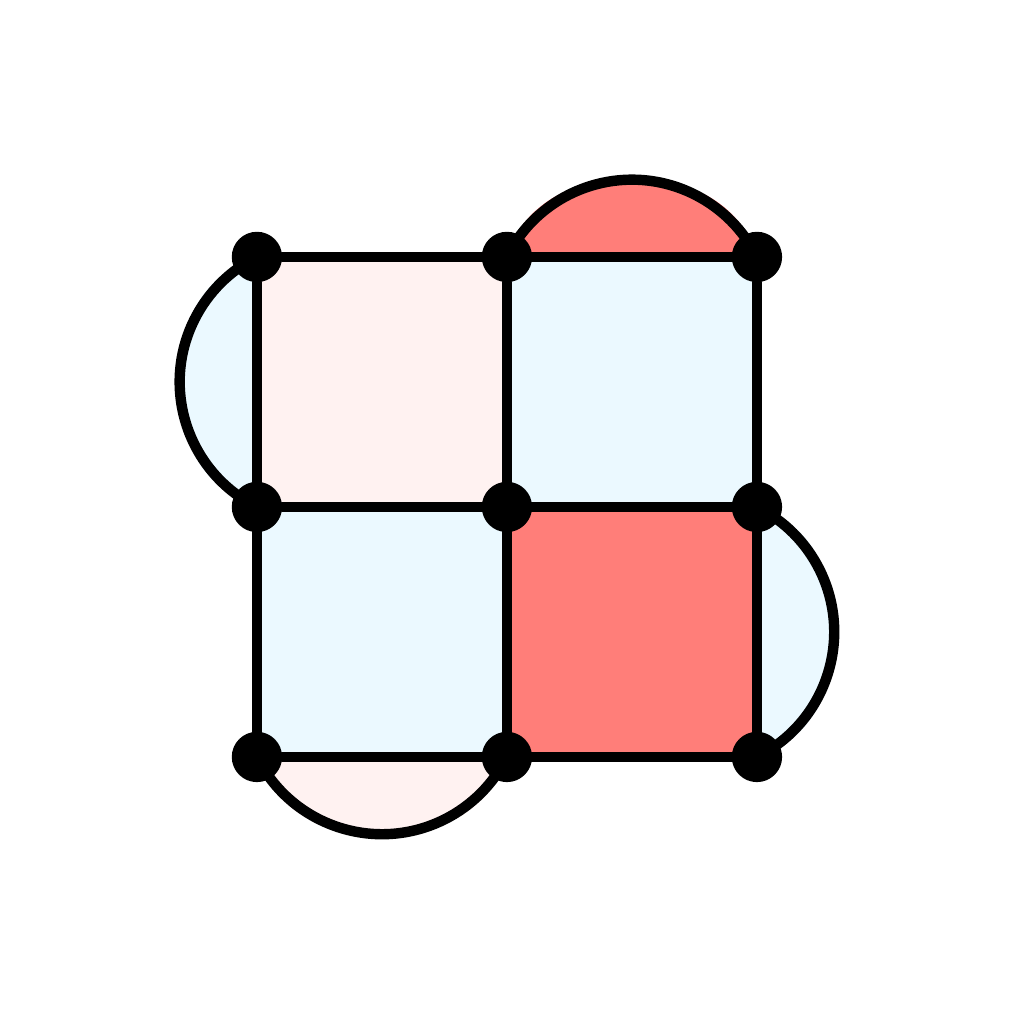}
\;\;
\includegraphics[scale=.2]{images/913surfacecode_round1}
\end{split}
\end{equation*}
A symmetrical sequence works for $Z$ error correction.  

For the $\llbracket 25,1,5 \rrbracket$ code, $30$ $Z$ measurements, applied in five rounds, suffice for distance-five fault-tolerant $X$ error correction: 
\begin{equation*}
\begin{split}
\includegraphics[scale=.1]{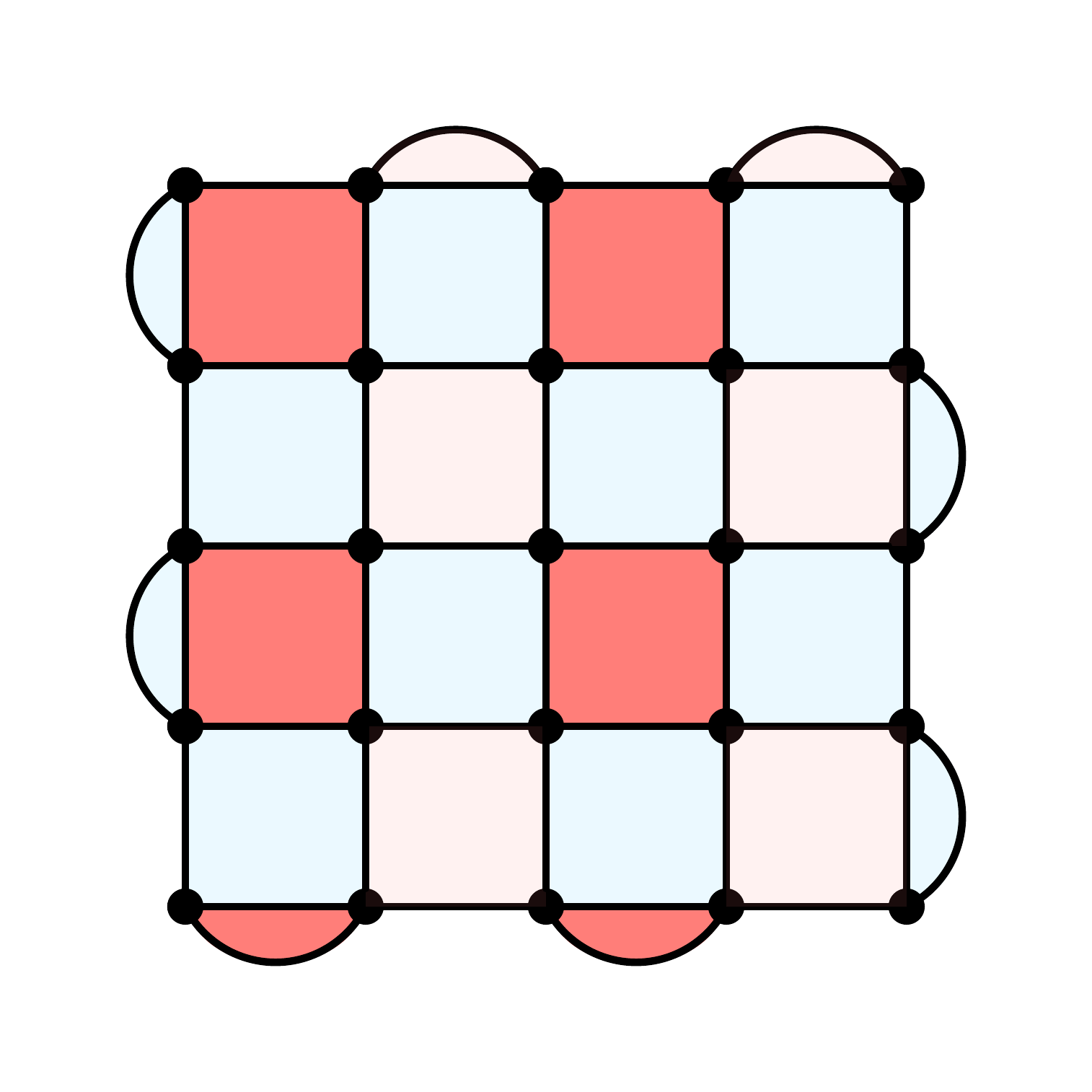}
\;\;
\includegraphics[scale=.1]{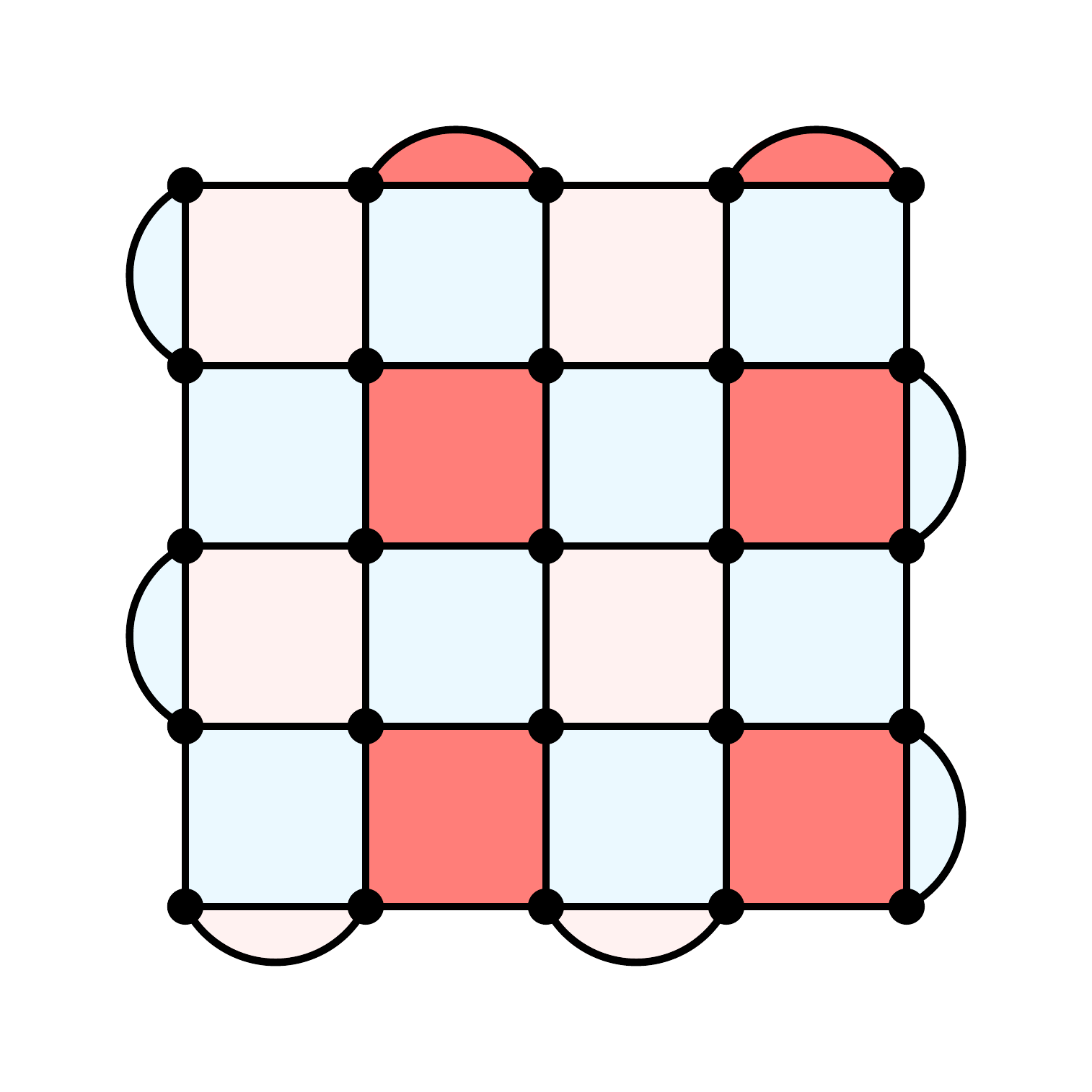}
\;\;
\includegraphics[scale=.1]{images/2515surfacecode_round1}
\;\;
\includegraphics[scale=.1]{images/2515surfacecode_round2}
\;\;
\includegraphics[scale=.1]{images/2515surfacecode_round1}
\end{split}
\end{equation*}

\medskip 

We leave the generalization of this scheme to the whole family of surface codes as an open question.  

\begin{open question}
For the $\llbracket d^2, 1, d \rrbracket$ surface code, prove that the natural $d$-round generalization of the above measurement sequences is fault tolerant to distance~$d$.  
\end{open question}

The standard method of surface code error correction uses $d$ rounds of $(d-1)^2$ measurements each.
By generalizing our sequence, one could potentially achieve fault-tolerant error correction with half as many measurements.  
However, each measurement in our model must be fully fault tolerant, which requires more ancilla qubits than the standard method.  
It might be interesting to compare the two approaches numerically.  

At least asymptotically, fewer measurements are required if products of plaquette stabilizers can be measured.  
Indeed, there exist sequences of $O(d \log d)$ stabilizer measurements for distance-$d$ fault-tolerant error correction for surface codes~\cite{DelfosseReichardtSvore20singleshot}.  

\begin{open question} [Sub-single-shot topological codes]
Construct \emph{explicit} sequences of $O(d \log d)$ stabilizer measurements for distance-$d$ fault-tolerant error correction for surface codes and color codes.  
\end{open question}


\subsection{Cyclic codes: \texorpdfstring{$\llbracket 31, 11, 5 \rrbracket$}{[[31,11,5]]} BCH and \\ \texorpdfstring{$\llbracket 23, 1, 7 \rrbracket$}{[[23,1,7]]}~Golay codes}

The $\llbracket 31, 11, 5 \rrbracket$ BCH code is a self-dual CSS code whose $Z$ and $X$ stabilizer groups are both generated by 
\begin{align*}
\begin{smallmatrix}
1&0&0&0&0&0&0&0&0&0&1&1&0&1&0&1&0&1&1&1&1&0&0&1&0&0&1&0&1&0&0
\end{smallmatrix}
\end{align*}
and its cyclic permutations, $10$ generators for each group.\footnote{This presentation can be recovered in 
Magma~\cite{Magma97} with 
commands ``C := \mbox{BCHCode}(GF(2),~31,~5); ParityCheckMatrix(C)".}    

Distance-five fault-tolerant $X$ error correction can be done by measuring this stabilizer generator and its next $26$ right cyclic permutations---$27$ measurements total.  (This is not optimal, however, as we have also found working sequences of $26$ measurements.)  

\medskip

The $\llbracket 23, 1, 7 \rrbracket$ Golay code is a self-dual CSS code whose $Z$ and $X$ stabilizer groups are both generated by 
\begin{align*}
\begin{smallmatrix}
1&0&0&0&0&0&0&0&0&0&0&1&1&1&1&1&0&0&1&0&0&1&0
\end{smallmatrix}
\end{align*}
and its cyclic permutations, $11$ generators for each group.  

Measuring the syndrome bit of this stabilizer generator and its next $29$ right cyclic permutations---$30$ measurements total---is sufficient for distance-seven fault-tolerant $X$ error correction.  
(For a $\llbracket 21, 3, 5 \rrbracket$ punctured Golay code, we have also verified that $22$ stabilizer measurements suffice for distance-five fault-tolerant $X$ error correction.)  

\medskip 

Conceivably, the additional structure of cyclic codes could allow for a general analysis that is more efficient than \thmref{t:distancethreenonadaptive}.  We leave this to further study.  

\begin{open question}
Specializing to cyclic codes, find minimum-length measurement sequences for distance-$d$ fault-tolerant error correction.  
\end{open question}

\section{Logical measurement} \label{s:logicalmeasurement}

We have considered adaptive and nonadaptive fault-tolerant syndrome measurement sequences that allow for fault-tolerant error correction.  In a fault-tolerant quantum computer, however, one also needs fault-tolerant implementations of logical operations, the simplest being logical measurement.  

Fault-tolerant logical measurement is not so simple as measuring a logical operator, or even doing so repeatedly.  
For example, with the $\llbracket 7,1,3 \rrbracket$ Steane code, $Z_1 Z_2 Z_3$ is a logical $Z$ operator, but if you use it to measure a codeword with an $X_1$ error, you will get the wrong answer every time.  Instead, different logical operators need to be measured to implement a fault-tolerant logical $Z$ measurement.  Two good measurement sequences are given in \figref{f:713logicalmeasurementsequences}.  

\begin{figure}
\centering
\begin{tabular}{c@{$\quad\qquad$}c}
\subfigure[\label{f:713logicalmeasurementsequences5logicalops}]{
$\begin{array}{r c c c c c c c}
&.&.&1&1&.&.&1\\
&1&1&1&.&.&.&.\\
&1&.&.&1&1&.&.\\
&1&.&.&.&.&1&1\\
&.&1&.&1&.&1&.
\end{array}$}
&
\subfigure[]{
$\begin{array}{r c c c c c c c}
&.&1&.&.&1&.&1\\
&1&1&1&.&.&.&.\\
&1&.&.&1&1&.&.\\
&1&.&.&.&.&1&1\\
&1&.&1&.&1&.&1
\end{array}$}
\end{tabular}
\caption{
A fault-tolerant logical measurement for the $\llbracket 7,1,3 \rrbracket$ code can be implemented with either of these measurement sequences.  The first sequence measures five equivalent logical operators, while the second measures four logical operators and a code stabilizer.  Both sequences also work for fault-tolerant error correction.  
} \label{f:713logicalmeasurementsequences}
\end{figure}

\begin{definition}
Consider a sequence $\cP = (P_1, \ldots, P_k)$, where each $P_j$ is either a $Z$ stabilizer or a $Z$ operator equivalent to logical~$Z$.  Let $b_j$ be $0$ or $1$ in the respective two cases.  Then measuring $\cP$ allows for distance-three fault-tolerant measurement of logical~$Z$ provided that no up to two faults can lead to measurement outcomes $(b_1, \ldots, b_k)$.  
\end{definition}

With both sequences in \figref{f:713logicalmeasurementsequences}, for example, measuring only the first four operators is not sufficient for a fault-tolerant logical~$Z$ measurement, because an input fault on the first qubit and a syndrome fault on the first measurement outcome would lead to outcomes $1111$.  Thus logical $0$ with one fault could not be distinguished from logical~$1$ with one fault.  

In this section, we will study measurement sequences that allow for fault-tolerant logical measurements.  We will focus on the $\llbracket 15,7,3 \rrbracket$ and $\llbracket 16,6,4 \rrbracket$ codes introduced earlier, because of their practical interest.  The codes also have a rich group of qubit permutation automorphisms that simplifies a case-by-case consideration of the many different logical operators.  For example, with six encoded qubits, the $\llbracket 16,6,4 \rrbracket$ code has $2^6 - 1$ nontrivial logical $Z$ operators that we might wish to measure---but we will see below that up to code-preserving qubit permutations (which can be implemented by relabeling qubits) there are only two equivalence classes of logical $Z$ operators.  

We will also study measurement sequences that allow for fault-tolerant logical measurements combined with fault-tolerant error correction.  It turns out that one can do both together faster than running logical measurement and error correction in sequence.  We will consider logical measurements across multiple code blocks, e.g., measuring $\overline{Z}_1 \otimes \overline{Z}_1$ on two $\llbracket 15,7,3 \rrbracket$ code blocks.  Finally, we will consider combining multiple logical measurements, e.g., measuring $\overline{Z}_1$ and $\overline{Z}_2$ together faster than measuring them separately in sequence.

\subsection{Logical operators and permutation automorphisms}

\begin{figure*}
\centering
\begin{tabular}{c@{$\quad\qquad$}c}
\subfigure[$\llbracket 15,7,3 \rrbracket$ code \label{f:hadamardcodelogicalbases1573}]{
$\begin{array}{r c c c c c c c c c c c c c c c}
\widebar Z_1, \widebar X_1:&1&1&.&1&.&.&.&1&.&.&.&.&.&.&1\\
\widebar Z_2, \widebar X_2:&1&1&.&.&1&.&.&.&.&1&.&1&.&.&.\\
\widebar Z_3, \widebar X_3:&1&1&.&.&.&1&.&.&.&.&1&.&.&1&.\\
\widebar Z_4, \widebar X_4:&1&1&.&.&.&.&1&.&1&.&.&.&1&.&.\\
\widebar Z_5, \widebar X_5:&1&.&.&1&.&1&.&.&1&1&.&.&.&.&.\\
\widebar Z_6, \widebar X_6:&1&.&.&1&.&.&1&.&.&.&.&1&.&1&.\\
\widebar Z_7, \widebar X_7:&1&.&.&.&.&.&.&1&.&1&.&.&1&1&.
\end{array}$}
&
\subfigure[$\llbracket 16,6,4 \rrbracket$ code \label{f:hadamardcodelogicalbases1664}]{\raisebox{-.2cm}{
$\begin{array}{r c c c c c c c c c c c c c c c c}
\widebar Z_1, \widebar X_2: &1&1&1&1&.&.&.&.&.&.&.&.&.&.&.&.\\
\widebar Z_2, \widebar X_1: &1&.&.&.&1&.&.&.&1&.&.&.&1&.&.&.\\
\widebar Z_3, \widebar X_4: &1&1&.&.&1&1&.&.&.&.&.&.&.&.&.&.\\
\widebar Z_4, \widebar X_3: &1&.&1&.&.&.&.&.&1&.&1&.&.&.&.&.\\
\widebar Z_5, \widebar X_6: &1&1&.&.&.&.&.&.&1&1&.&.&.&.&.&.\\
\widebar Z_6, \widebar X_5: &1&.&1&.&1&.&1&.&.&.&.&.&.&.&.&.
\end{array}$}}
\end{tabular}
\caption{
Possible bases for the logical qubits for the (a) $\llbracket 15,7,3 \rrbracket$ and (b) $\llbracket 16,6,4 \rrbracket$ codes.  (To explain the notation, for example in (b), $\protect\widebar Z_1 = Z^{\otimes 4} \otimes I^{\otimes 12}$ and $\protect\widebar X_2 = X^{\otimes 4} \otimes I^{\otimes 12}$.)  There are other bases that might be useful, for example for the $\llbracket 15,7,3 \rrbracket$ code one can choose a basis with six weight-four operators that also work for the unpunctured $16$-qubit code, and one weight-seven operator.  In the basis shown here, the operators have weight five and are self-dual, so transversal Hadamard implements logical transversal Hadamard.  
} \label{f:hadamardcodelogicalbases}
\end{figure*}

Bases for the encoded qubits for the $\llbracket 15,7,3 \rrbracket$ and $\llbracket 16,6,4 \rrbracket$ codes are given in \figref{f:hadamardcodelogicalbases}.  These bases are only for reference, as the details are not important here.  

The weight of a logical operator $\widebar P$ is the least Hamming weight of any stabilizer-equivalent operator.  The weight distributions of the two codes' $Z$ or $X$ logical operators are given in \figref{f:15731664weightdistributions}.   

\begin{figure}
\centering
\begin{tabular}{c@{$\qquad\qquad$}c}
\subfigure[$\llbracket 15,7,3 \rrbracket$ code]{
\begin{tabular}{c @{$\;\;$} c}
\hline \hline
Weight & \# Operators  \\
\hline
0 & 1 \\
3 & 35 \\
4 & 35 \\
5 & 28 \\
6 & 28 \\
7 & 1 \\
\hline \hline \\
\end{tabular}
}
&
\raisebox{-.66cm}{
\subfigure[$\llbracket 16,6,4 \rrbracket$ code]{
\begin{tabular}{c @{$\;\;$} c}
\hline \hline
Weight & \# Operators  \\
\hline
0 & 1 \\
4 & 35 \\
6 & 28 \\
\hline \hline \\
\end{tabular}
}}
\end{tabular}
\caption{Distributions of weights of the $Z$ or $X$ logical operators for the $\llbracket 15,7,3 \rrbracket$ and $\llbracket 16,6,4 \rrbracket$ codes.} \label{f:15731664weightdistributions}
\end{figure}

The permutation automorphism group of a code is the set of qubit permutations that preserve the codespace.
The permutation automorphism group of the $\llbracket 15,7,3 \rrbracket$ code has order $20,\!160$, and is isomorphic to $A_8$ and $\mathrm{GL}_4(\Z_2)$.  It is generated by the three permutations\footnote{Magma commands to find the automorphism group are {``C := LinearCode$\langle$GF(2),15$\vert$
[0,0,0,0,0,0,0,1,1,1,1,1,1,1,1],
[0,0,0,1,1,1,1,0,0,0,0,1,1,1,1],
[0,1,1,0,0,1,1,0,0,1,1,0,0,1,1],
[1,0,1,0,1,0,1,0,1,0,1,0,1,0,1]$\rangle$;
AutomorphismGroup(C);"}.} 
\begin{gather*}
(1,2,3)(4,14,10)(5,12,9)(6,13,11)(7,15,8) \\
(1,10,5,2,12)(3,6,4,8,9)(7,14,13,11,15) \\
(1,10,15,3,8,13)(4,6)(5,12,11)(7,14,9)
\end{gather*}

The permutation automorphism group of the $\llbracket 16,6,4 \rrbracket$ has order $322,\!560$, and is generated by the permutations 
\begin{gather*}
(1,2)(3,4)(5,6)(7,8)(9,10)(11,12)(13,14)(15,16) \\
(1,2,4,8,16,15,13,9)(3,6,12,7,14,11,5,10) \\
(9,10)(11,12)(13,14)(15,16)
\end{gather*}
A subgroup of order $16$ acts trivially, with no logical effect (the first permutation above, e.g., has no logical effect). \appref{s:automorphism_group} provides further details on these automorphism groups and their relationship.  

A large permutation automorphism group allows for a rich set of logical operations to be applied by simply permuting the physical qubits, or perhaps just by relabeling them~\cite{HarringtonReichardt11permutations, Grassl13automorphisms}.  That is not our concern here.  Instead, observe: 

\begin{claim}
For both codes, any two logical operators with the same weight are related by a qubit permutation in the automorphism group.  
\end{claim}

\noindent
(Logical operators with different weights of course cannot be related by a permutation automorphism.)  

Therefore, up to permutation automorphisms, there are five equivalence classes of nontrivial logical operators for the $\llbracket 15,7,3 \rrbracket$ code, and just two equivalence classes for the $\llbracket 16,6,4 \rrbracket$ code.  This greatly simplifies our problem of specifying sequences for measuring logical operators fault tolerantly.  It is sufficient to find a sequence that works for one logical operator in each weight equivalence class; then for any logical operator of the same weight, a working measurement sequence can be obtained by applying the appropriate qubit permutation.

\subsection{\texorpdfstring{$\llbracket 15,7,3 \rrbracket$}{[[15,7,3]]} code: Measurement and error~correction}

We have found that every logical operator can be fault-tolerantly measured using at most six fault-tolerant measurements.  For weight-four and weight-five logical operators, three and five measurements suffice, respectively; see \figref{f:1573measurementsforlogicalmeasurement}.  

Certainly, three measurements are needed for a fault-tolerant logical measurement.  With two or fewer measurements, a single measurement fault would not be correctable.  For measuring weight-four logical operators, three measurements suffice, because every such operator has three representatives with disjoint supports.  For example, these three logical operators are equivalent up to stabilizers: 
\begin{equation} \label{e:1573hammingweightfourlogicaloperatormeasurement}
\begin{array}{r c c c c c c c c c c c c c c c}
&.&.&.&.&.&.&.&.&.&.&.&1&1&1&1\\
&.&.&.&.&.&.&.&1&1&1&1&.&.&.&.\\
&.&.&.&1&1&1&1&.&.&.&.&.&.&.&.
\end{array}
\end{equation}
A single error on the input, or a single fault during the measurements, can flip at most one of the three outcomes, so the majority will still be correct.  

The following sequence of measurements works for a weight-three logical operator.  Here the first three measurements are of equivalent logical operators, and the last three are of stabilizers.  (It is also possible to use six logical operator measurements, and in fact that can give a lower total weight, $38$ instead of $41$.)  
\begin{equation*}
\place{\footnotesize \text{logical}}{20mu}{26pt}
\place{\footnotesize \text{operators}}{20mu}{16pt}
\place{\footnotesize \text{stabilizers}}{20mu}{-19pt}
\place{$\left\{\begin{array}{c} \\ \\ \end{array}\right.$}{75mu}{20pt}
\place{$\left\{\begin{array}{c} \\ \\ \end{array}\right.$}{75mu}{-21pt}
\qquad\qquad\begin{array}{r c c c c c c c c c c c c c c c}
&1&1&1&.&.&.&.&.&.&.&.&.&.&.&.\\
&1&.&.&1&1&.&.&1&1&.&.&.&.&1&1\\
&.&1&.&1&.&1&.&1&.&1&.&.&1&.&1\\
&1&.&1&.&1&.&1&.&1&.&1&.&1&.&1\\
&.&1&1&1&1&.&.&1&1&.&.&.&.&1&1\\
&.&.&.&.&.&.&.&1&1&1&1&1&1&1&1
\end{array}
\end{equation*}
Why are the last three measurements necessary?  If we only made the first three measurements, of equivalent logical operators, then without any errors logical $0$ would result in measurement outcomes $000$ and logical $1$ in outcomes $111$.  However, with a input error on the last qubit, logical~$0$ would result in measurement outcomes $011$, which cannot be distinguished from logical~$1$ with an erroneous first measurement.  With the last three stabilizer measurements, ideally the measurement outcomes will be either $000000$, for logical~$0$, or $111000$, for logical~$1$.  One can check that no one or two faults, either on the input or during the measurements can flip $000000$ to $111000$, and hence logical~$0$ and logical~$1$ will be distinguishable even if there is up to one fault.  

With the aid of a computer to verify fault tolerance, measurement sequences for logical operators of weights five, six or seven can be similarly found (\figref{f:1573measurementsforlogicalmeasurement}).  

Given that one has to make multiple measurements in order to measure a logical operator fault tolerantly, it makes sense to use the extracted information not just for determining the logical outcome, but also for correcting errors.  Can one combine measurement of a logical $Z$ operator with $X$ error correction, faster than running them sequentially?  Yes.  

As listed in \figref{f:1573measurementsforlogicalmeasurement}, in fact for any logical operator seven $Z$ measurements suffice for logical measurement and $X$ error correction together.  For a weight-five logical operator, just six $Z$ measurements suffice: 
\begin{equation*}
\begin{array}{r c c c c c c c c c c c c c c c}
&.&.&.&.&1&.&.&.&1&.&.&.&1&1&1\\
&.&.&.&.&1&.&.&1&.&1&1&1&.&.&.\\
&.&.&.&1&.&1&1&.&1&.&.&1&.&.&.\\
&.&1&1&1&.&.&.&1&.&.&.&.&1&.&.\\
&1&.&1&.&.&.&1&.&.&.&1&.&.&1&.\\
&1&1&.&.&.&1&.&.&.&1&.&.&.&.&1
\end{array}
\end{equation*}
This measurement sequence, of six equivalent weight-five operators, satisfies that no up to two input or internal faults can flip the ideal syndrome for logical $0$, $0^6$, to the ideal syndrome for logical~$1$, $1^6$.  Therefore with at most one fault, logical~$0$ can be distinguished from logical~$1$.  Then, the differences from the ideal syndromes can be used to diagnose and safely correct input errors.  

\begin{figure}[!th]
\vspace{0cm} 
\centering
\begin{tabular}{c}
\subfigure[$\llbracket 15,7,3 \rrbracket$ code \label{f:1573measurementsforlogicalmeasurement}]{
\raisebox{.3cm}{\includegraphics[scale=.9]{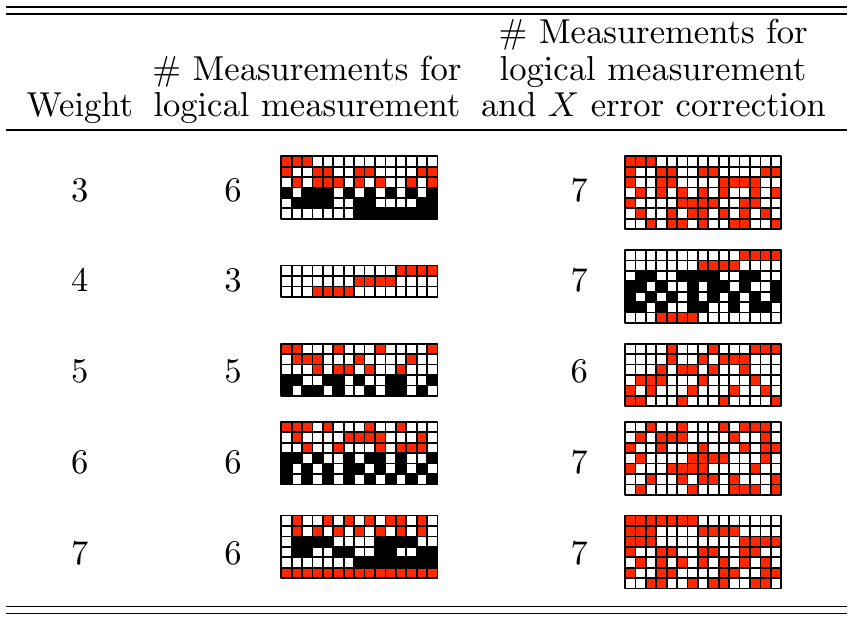}}
}
\\
\subfigure[$\llbracket 16,6,4 \rrbracket$ code, distance-three fault tolerance \label{f:1664measurementsforlogicalmeasurement}]{
\raisebox{.3cm}{\includegraphics[scale=.9]{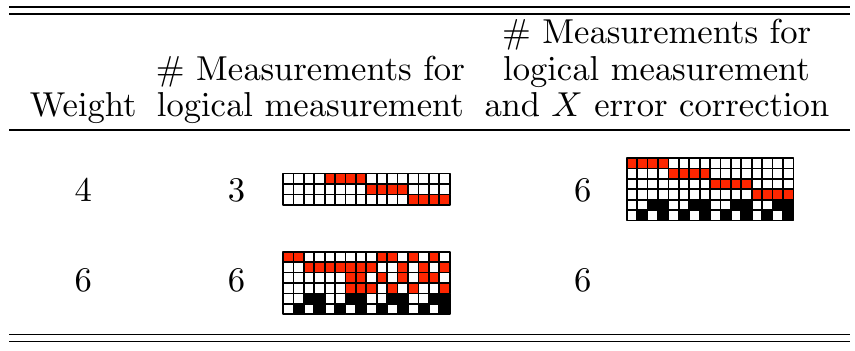}}
}
\\
\subfigure[$\llbracket 16,6,4 \rrbracket$ code, distance-four fault tolerance \label{f:1664measurementsforlogicalmeasurementdistance4}]{
\raisebox{.3cm}{\includegraphics[scale=.9]{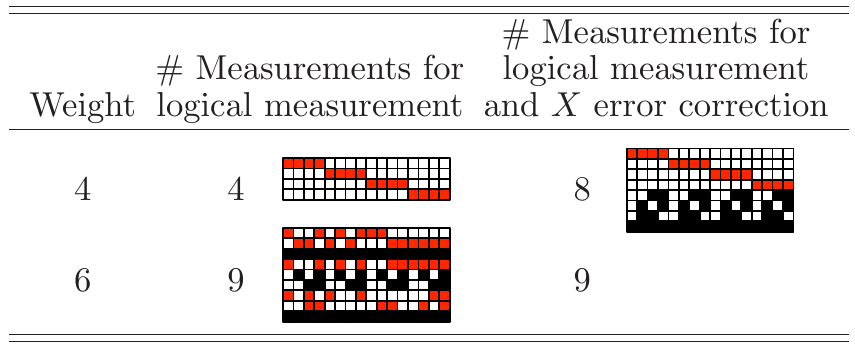}}
}
\end{tabular}
\caption{
For each weight, the tables give a number of measurements sufficient for fault-tolerantly measuring a logical operator of that weight.  For example, Eq.~\eqnref{e:1573hammingweightfourlogicaloperatormeasurement} gives three measurements that suffice to measure a weight-four logical operator.  The tables also give a length sufficient both for measuring a $Z$ logical operator and carrying out $X$ error correction.  As distance-three $X$ error correction on its own needs seven $Z$ measurements for the $\llbracket 15,7,3 \rrbracket$ code, and five measurements for the $\llbracket 16,6,4 \rrbracket$ code---or seven measurements for distance-four $X$ error correction---there are substantial savings from combining logical measurement with error correction.  
In the illustrated measurement sequences, all filled squares represent $Z$ operators.  Red rows correspond to logical operators and black rows to stabilizers.  
}
\end{figure}

Recall from \propref{t:hammingcodesnonadaptiveerrorcorrection} that $X$ error correction on its own uses seven nonadaptive $Z$ stabilizer measurements.  Thus by combining the $Z$ measurement steps, $6 + 7 = 13$ measurements suffice for a weight-five logical measurement and full error correction, versus $5 + 7 + 7 = 19$ steps for running error correction separately.  

This result would sound more impressive if we used a weaker baseline.  
The most naive procedure would simply fix a logical $Z$ operator~$P$, and repeat $d = 3$ times: fault-tolerant quantum error correction and measure~$P$.  With Shor's $16$-measurement sequence for $X$ error correction, this makes $3 (16 + 1) = 51$ measurements for logical $Z$ measurement and full error correction.  Even using the seven-measurement sequence from \propref{t:hammingcodesnonadaptiveerrorcorrection}, we get $3 (7 + 1) = 24$.  A little effort spent in optimizing logical measurement is well worth it.

\subsection{\texorpdfstring{$\llbracket 16,6,4 \rrbracket$}{[[16,6,4]]} code: 
Measurement and error~correction} \label{s:1664measurementerrorcorrection}

From \figref{f:15731664weightdistributions}, there are two weight equivalence classes of nontrivial logical operators, weight-four and weight-six operators.  
Although the $\llbracket 16,6,4 \rrbracket$ code has distance four, we will consider fault tolerance only to distance three, i.e., tolerating up to one input error or internal fault.  

Any weight-four operator can be measured fault tolerantly in three steps, just as in Eq.~\eqnref{e:1573hammingweightfourlogicaloperatormeasurement}.  (Adding an initial qubit makes the operators in \eqnref{e:1573hammingweightfourlogicaloperatormeasurement} valid logical operators for the $\llbracket 16,6,4 \rrbracket$ code.)  

For logical operators of weight four, six measurement steps suffice for combined logical $Z$ measurement and $X$ error correction: 
\begin{equation*}
\place{\footnotesize \text{logical}}{20mu}{20pt}
\place{\footnotesize \text{operators}}{20mu}{10pt}
\place{\footnotesize \text{stabilizers}}{20mu}{-27pt}
\place{$\left\{\begin{array}{c} \\ \\ \\ \end{array}\right.$}{75mu}{13pt}
\place{$\left\{\begin{array}{c} \\ \end{array}\right.$}{75mu}{-29pt}
\qquad\qquad\begin{array}{r c c c c c c c c c c c c c c c c}
&1&1&1&1&.&.&.&.&.&.&.&.&.&.&.&.\\
&.&.&.&.&1&1&1&1&.&.&.&.&.&.&.&.\\
&.&.&.&.&.&.&.&.&1&1&1&1&.&.&.&.\\
&.&.&.&.&.&.&.&.&.&.&.&.&1&1&1&1\\
&.&.&1&1&.&.&1&1&.&.&1&1&.&.&1&1\\
&.&1&.&1&.&1&.&1&.&1&.&1&.&1&.&1
\end{array}
\end{equation*}
(Measuring the four disjoint, equivalent logical operators suffices for fault-tolerant logical measurement.  For error correction, a logical operator measurement different from the others identifies which block of four qubits an input error occurred on, and the two stabilizer measurements then fully localize the error.)  

For logical operators of weight six, too, six steps suffice: 
\begin{equation*}
\place{\footnotesize \text{logical}}{20mu}{20pt}
\place{\footnotesize \text{operators}}{20mu}{10pt}
\place{\footnotesize \text{stabilizers}}{20mu}{-27pt}
\place{$\left\{\begin{array}{c} \\ \\ \\ \end{array}\right.$}{75mu}{13pt}
\place{$\left\{\begin{array}{c} \\ \end{array}\right.$}{75mu}{-29pt}
\qquad\qquad\begin{array}{r c c c c c c c c c c c c c c c c}
&1&1&.&.&.&.&.&.&.&1&1&.&1&.&1&.\\
&.&.&1&1&1&1&1&1&1&.&.&1&.&1&.&1\\
&.&.&.&.&.&.&1&1&.&1&.&1&.&1&1&.\\
&.&.&.&.&.&.&1&1&1&.&1&.&1&.&.&1\\
&.&.&1&1&.&.&1&1&.&.&1&1&.&.&1&1\\
&.&1&.&1&.&1&.&1&.&1&.&1&.&1&.&1
\end{array}
\end{equation*}

Recall from \secref{s:generalizing844and16114codes} that five $Z$ stabilizer measurements suffice for distance-three fault-tolerant $X$ error correction.  Combining a logical $Z$ measurement with $X$ error correction thus costs only one more measurement.  See \figref{f:1664measurementsforlogicalmeasurement}.  

\medskip

The above measurement sequences only give distance-three protection.  This allows for a fair comparison of the $\llbracket 16,6,4 \rrbracket$ code with the $\llbracket 15,7,3 \rrbracket$ code and with the other codes in \figref{f:codecomparison}.  To take full advantage of the $\llbracket 16,6,4 \rrbracket$ code's greater distance, though, more measurements are needed.  The sequences given in \figref{f:1664measurementsforlogicalmeasurementdistance4} allow for fault tolerance to distance four.  Recall 
that seven measurements suffice for distance-four fault-tolerant $X$ error correction alone.  
Thus, once again there are substantial savings from combining a logical measurement with error correction.  

\medskip

It would be interesting to study fault-tolerant logical measurement for other codes, beyond the $\llbracket 15,7,3 \rrbracket$ and $\llbracket 16,6,4 \rrbracket$ codes, as in Theorems~\ref{t:distancethreenonadaptive} and~\ref{t:distancethreenoncss}.  
As in \secref{s:hammingcodes}, one could also explore the impact of restricted measurements.  

\begin{open question}
For general distance-three stabilizer codes, CSS or not, find measurement sequences for fault-tolerant logical measurements.  
\end{open question}

\subsection{Measuring logical operators across multiple code blocks} \label{s:logicalmeasurementacrosscodeblocks}

From the above analyses, we can  
implement 
fault-tolerant $X$ error correction combined with measurement of any logical $Z$ operator, for the $\llbracket 7,1,3 \rrbracket$, $\llbracket 15,7,3 \rrbracket$ and $\llbracket 16,6,4 \rrbracket$ codes.  However, what if we want to measure a logical $Z$ operator across multiple code blocks, for example, $\widebar Z_1 \otimes \widebar Z_2$ on two code blocks, or perhaps $\widebar Z_1 \otimes \widebar Z_2 \otimes (\widebar Z_3 \widebar Z_4 \widebar Z_5 \widebar Z_6 \widebar Z_7)$ on three code blocks?  

Distance-three fault tolerance for a multi-block logical $Z$ operator requires the same condition needed for a single-block logical $Z$ operator: no two faults should be able to flip the all-zeros syndrome to the all-ones syndrome.  With this condition, logical~$0$ and~$1$ can be distinguished even with up to one input error or internal fault, across all the involved code blocks.  In general, one has to search to find working measurement sequences.  

Fortunately, in many cases we can use the measurement sequences found already.  For measuring $\widebar Z_a \otimes \widebar Z_b \otimes \cdots$, if the individual logical operators $\widebar Z_a, \widebar Z_b, \ldots$ are related by a permutation automorphism, and therefore working measurement sequences for $\widebar Z_a, \widebar Z_b, \ldots$ differ only by a qubit permutation, then these sequences can be combined into a sequence for measuring $\widebar Z_a \otimes \widebar Z_b \otimes \cdots$.  For example, for the $\llbracket 7,1,3 \rrbracket$ code, place two copies of the measurement sequence from \figref{f:713logicalmeasurementsequences5logicalops} side-by-side, in order to obtain a sequence of five equivalent $\widebar Z \otimes \widebar Z$ operators: 
\begin{equation*}
\place{\footnotesize \text{block 1}}{80mu}{20pt}
\place{\footnotesize \text{block 2}}{226mu}{20pt}
\place{\rotatebox{-90}{$\left\{\begin{array}{c} \\[2.2cm] \end{array}\right.$}}{80mu}{16pt}
\place{\rotatebox{-90}{$\left\{\begin{array}{c} \\[2.2cm] \end{array}\right.$}}{226mu}{16pt}
\raisebox{-1cm}{$\begin{array}{r c c c c c c c @{$\quad$} c c c c c c c}
&.&.&1&1&.&.&1 &.&.&1&1&.&.&1\\ 
&1&1&1&.&.&.&. &1&1&1&.&.&.&.\\
&1&.&.&1&1&.&. &1&.&.&1&1&.&.\\
&1&.&.&.&.&1&1 &1&.&.&.&.&1&1\\
&.&1&.&1&.&1&. &.&1&.&1&.&1&.
\end{array}$}
\end{equation*}
This is a fault-tolerant $\widebar Z \otimes \widebar Z$ measurement sequence.  Indeed, the syndrome errors that can be caused by a single fault in block~$1$ (e.g., $01110$ from an $X_1$ input error) are the same as those that a single fault in block~$2$ can cause.  The fault-tolerance condition for one block implies that no two faults can flip syndrome $0^5$ to~$1^5$.  

We therefore obtain fault-tolerant sequences for measuring $\widebar Z_a \otimes \widebar Z_b \otimes \cdots$, provided that each operator lies in the same permutation equivalence class; for the $\llbracket 15,7,3 \rrbracket$ and $\llbracket 16,6,4 \rrbracket$ codes, this means that they have the same weight (\figref{f:15731664weightdistributions}).  

If $\widebar Z_a$ and $\widebar Z_b$ have different weights, then more work is required to find a fault-tolerant measurement sequence for $\widebar Z_a \otimes \widebar Z_b$, because single faults in a $\widebar Z_a$ measurement sequence have different syndrome effects than single faults in a $\widebar Z_b$ measurement sequence.  We leave this search as an exercise.  

\begin{open question} 
Find measurement sequences for fault-tolerantly measuring $\widebar Z_a \otimes \widebar Z_b$ over two code blocks of a CSS code.  
\end{open question}

\subsection{Further problems for logical measurement} \label{s:furtherlogicalmeasurementproblems}

There are further logical measurement problems, with practical utility depending on the application.  

For example, one problem is to measure multiple logical operators in parallel, possibly combined with error correction.  With the $\llbracket 16,6,4 \rrbracket$ code, e.g., say we want to measure $\widebar Z_1$ and $\widebar Z_3$, from the basis of \figref{f:hadamardcodelogicalbases1664}, fault tolerant to distance three.  As both operators have weight four, we can measure them both separately in $3 + 3 = 6$ steps, or we can measure them both separately, with one logical measurement combined with error correction, in $3 + 6 = 9$ steps.  However, we can fault-tolerantly measure them together, with error correction, in seven steps, as follows: 
\begin{equation*}
\place{\footnotesize \text{$\widebar Z_1$}}{35mu}{34pt}
\place{$\left\{\begin{array}{c} \\[.1cm] \end{array}\right.$}{75mu}{34pt}
\place{\footnotesize \text{$\widebar Z_3$}}{35mu}{4pt}
\place{$\left\{\begin{array}{c} \\[.1cm] \end{array}\right.$}{75mu}{4pt}
\place{\footnotesize \text{$\widebar Z_1 \widebar Z_3$}}{35mu}{-15pt}
\place{$\left\{\begin{array}{c} \\[-.05cm] \end{array}\right.$}{75mu}{-15pt}
\place{\footnotesize \text{stabilizers}}{20mu}{-36pt}
\place{$\left\{\begin{array}{c} \\ \end{array}\right.$}{75mu}{-36pt}
\qquad\qquad\begin{array}{r c c c c c c c c c c c c c c c c}
&1&1&1&1&.&.&.&.&.&.&.&.&.&.&.&.\\
&.&.&.&.&1&1&1&1&.&.&.&.&.&.&.&.\\
&.&.&.&.&.&.&.&.&1&1&.&.&1&1&.&.\\
&.&.&.&.&.&.&.&.&.&.&1&1&.&.&1&1\\
&.&.&1&1&1&1&.&.&.&.&.&.&.&.&.&.\\
&.&.&1&1&1&1&.&.&.&.&1&1&1&1&.&.\\
&.&1&.&1&.&1&.&1&.&1&.&1&.&1&.&1
\end{array}
\end{equation*}
Essentially, instead of using separate $[3,1,3]$ classical repetition codes, in the first five steps we are using the $[5,2,3]$ classical code that encodes syndrome $(z_1, z_3)$ as $(z_1, z_1, z_3, z_3, z_1 \oplus z_3)$.  

In \secref{s:logicalmeasurementacrosscodeblocks} above we gave sequences for measuring $\widebar Z_a \otimes \widebar Z_b$ across two code blocks, in certain cases. 
What about combining the logical measurement with error correction, on two code blocks?  To consider this problem, one has to choose a suitable definition for fault tolerance.  Should a two-block error-correction procedure tolerate up to one input error or one internal fault total, across both blocks?  Or should it tolerate up to one input error or one internal fault on each block, so up to two faults total?  Or should it tolerate up to one input error on each block, and one internal fault total?  All these choices are possible, but tolerating more faults will generally require longer measurement sequences.  

We have given fault-tolerant implementations for all $X$-type and $Z$-type logical measurements.  
The codes considered allow transversal Hadamard.  
However, this is not enough to implement the full Clifford group on the encoded qubits.  
One could complete the Clifford group by injecting single-qubit $H$ gates and $S$ gates, or by designing fault-tolerant sequences for arbitrary logical Pauli measurements.  

\medskip
One can extend the notion of single-shot measurement sequences 
to logical measurements.
A measurement sequence for an $\llbracket n, k \rrbracket$ stabilizer code that 
implements a logical Pauli measurement combined with fault-tolerant quantum error correction 
is said to be {\em single-shot} if it contains at most $n - k + 1$ measurements.

\begin{open question} [Single-shot code]
For $d \geq 3$, find a stabilizer code equipped 
with a single-shot measurement sequence 
for distance-$d$ fault-tolerant quantum error correction, 
and
with single-shot measurement sequences
for distance-$d$ fault-tolerant logical Pauli measurement 
of each of the $4^k$ logical Pauli operators.
\end{open question}

In the case of the $\llbracket 16, 6, 4 \rrbracket$ code, we have constructed single-shot, distance-three fault-tolerant sequences for error correction and for logical measurement for all $X$ and $Z$ logical operators. We do not know if this can be extended to the full logical Pauli group.

\medskip
One could obtain even faster logical operations by optimizing measurement 
of sets of commuting logical operators.
A measurement sequence for an $\llbracket n, k \rrbracket$ stabilizer code that fault-tolerantly implements the simultaneous measurement of $m$ independent logical Pauli operators
combined with error correction is said to be {\em single-shot} if it contains at most 
$n - k + m$ measurements.
A code equipped with single-shot fault-tolerant measurement sequences
for any set of up to $m$ commuting logical Pauli operators 
is called an {\em $m$-fold single-shot code}.

\begin{open question} [$k$-fold single-shot code]
For $d \geq 3$ and $k \geq 2$, find an $\llbracket n, k \rrbracket$ stabilizer code that is 
$k$-fold single-shot code to distance $d$.
\end{open question}

\section{Conclusion}

We have obtained substantial speed-ups of fault-tolerant quantum error correction, using carefully designed syndrome-measurement sequences and codes optimized for single-shot error correction.  

We have constructed short Shor-style measurement sequences for fault-tolerant quantum error correction with small-distance codes adapted to different hardware capabilities, including available measurement types and adaptive or nonadaptive measurements.  
Our results contribute to make small block codes more competitive with topological codes.  
Faster error correction means fewer potential fault locations, which generally results in better performance.  
It also reduces the quantum computer's logical cycle time.  
We have focused on small-distance codes, which may be the most practical, but there is much potential for innovation in fault-tolerant error correction for higher-distance codes.  

We have designed families of single-shot quantum error correction schemes, and given very fast implementations of fault-tolerant logical $X$ and $Z$ measurements.  
These logical measurements can sometimes be performed integrated with error correction, using fewer measurements than an error correction cycle alone.  
It is likely that this approach can generalize to arbitrary logical Pauli measurements, which by frame tracking would allow implementing the entire logical Clifford group.  
Universality could then be achieved using magic state distillation and injection~\cite{BravyiKitaev04magic, BravyiHaah2012magic}. In \appref{s:1573measure6evenweightopsandec}, we design optimized measurement sequences that could be relevant for universality.

\medskip

{\bf Acknowledgements.}
The authors would like to thank Prithviraj Prabhu, Michael Beverland, Jeongwan Haah, Adam Paetznick, Vadym Kliuchnikov, Marcus Silva and Krysta Svore for insightful discussions.

\clearpage

\appendix

\section{Adaptive error correction measurement sequence for any self-dual, distance-three CSS code} 
\label{s:selfdualcssadaptive}

In \secref{s:distance3csscodesadaptive} we gave an adaptive error-correction measurement sequence that, for a distance-three CSS code with $R$ stabilizer generators, uses between $\numgenerators$ and $2 \numgenerators - 2$ measurements.  
(See also \figref{f:nonadaptiveadaptive}.)  
Here we present an adaptive error-correction sequence for \emph{self-dual}, distance-three CSS codes, that uses fewer measurements in the worst case.  

\begin{claim} \label{t:selfdualcssadaptive}
Consider an $\llbracket n, n - \numgenerators, 3 \rrbracket$ self-dual CSS code.  Then fault-tolerant error correction can be realized with an adaptive stabilizer-measurement sequence that makes between $\numgenerators$ and $\tfrac32 \numgenerators$ all-$X$ or all-$Z$ measurements.  
\end{claim}

\begin{proof}
Measure each of the the $\numgenerators/2$ $Z$ stabilizer generators and then each of the $\numgenerators/2$ $X$ stabilizer generators, stopping at the first nontrivial measurement outcome.  
\begin{description}[leftmargin=*]
\item[Case 1] If all syndrome bits are trivial, then stop, having made $\numgenerators$ measurements total.  
\item[Case 2] If the first nontrivial syndrome bit is among the $Z$ measurements, then measure all $\numgenerators$ generators and apply the corresponding one-qubit correction, if any.  The total number of measurements is between $\numgenerators + 1$ and~$\tfrac32 \numgenerators$.  This is fault tolerant because with at most one input or internal fault, the last $\numgenerators$ syndrome bits will be correct.  
\item[Case 3] If the first nontrivial syndrome bit is among the $X$ stabilizer measurements, then repeat that measurement.  If the result is trivial, then there must have been an internal fault, and no correction is required.  If the result is nontrivial again, then either there is an input $Z$ error or an internal $Z$ or $Y$ fault within the support of the stabilizer.  It is enough to measure the other $\numgenerators/2 - 1$ $X$ generators and apply the corresponding one-qubit $Z$ correction, if any.  This will correct for either an input $Z$ error or an internal $Z$ fault, and will convert an internal $Y$ fault to an $X$ error on the same qubit.  The total number of measurements is between $\numgenerators + 1$ and~$\tfrac32 \numgenerators$.  \qedhere
\end{description}
\end{proof}

\section{Automorphism group of \\ Hamming codes}
\label{s:automorphism_group}

Denote $Q_{16}$ the $\llbracket 16, 6, 4 \rrbracket$ code and let $Q_{15}$ be the $\llbracket 15, 7, 3 \rrbracket$ Hamming code.
In this section, we establish a relation between the automorphism groups of these codes, and we prove
that $\abs{ \Aut(Q_{16}) } = 16 \cdot \abs{ \Aut(Q_{15}) }$.  

\medskip
For $\ell = 15, 16$, denote by $C_{\ell}$ the classical self-orthogonal code whose codewords
correspond to the stabilizers of the code $Q_{\ell}$.  
The automorphism group of $Q_\ell$ is the set of permutations of the $\ell$ qubits that preserve
the code $C_{\ell}$.  

\medskip
The automorphism group of $Q_{16}$ acts on the set $\{0, \dots, 15\}$ of qubits and this action is transitive.  
As a consequence, Lagrange's theorem implies 
$$
\abs{ \Aut(Q_{16})_0 }  = \abs{ \Aut(Q_{16}) } / 16
$$
where ${\Aut(Q_{16})}_0$ denotes the set of automorphisms $\sigma$ of $Q_{16}$ that satisfy $\sigma(0) = 0$.

\medskip
Let us now establish an isomorphism between ${\Aut(Q_{16})}_0$ and ${\Aut(Q_{15})}$.  
Consider the transformation $\pi$ that maps $\sigma \in {\Aut(Q_{16})}_0$ onto its restriction $\bar \sigma$ to the set $\{1, \dots, 15\}$.  
Clearly, $\bar \sigma$ is a well defined bijection of the set $\{1, \dots, 15\}$.  
Moreover, since the code $C_{15}$ can be obtained from $C_{16}$ by puncturing the coordinate 0, $\bar \sigma$ preserves the code $C_{15}$.  
This proves that $\bar \sigma$ belongs to ${\Aut(Q_{15})}$.  
In other words, we defined a map 
\begin{align*}
\pi: {\Aut(Q_{16})}_0 & \longrightarrow \Aut(Q_{15}) \\
\sigma & \longmapsto \bar \sigma
\end{align*}
It is clearly an injective map.

\medskip
To prove that $\pi$ is surjective, consider $\bar \sigma \in \Aut(Q_{15})$.  
Define its extension $\sigma$ of $\bar \sigma$ which acts on $\{0, 1, \dots, 15\}$ with fixed-point~$0$.  
Let us prove that $\sigma$ is an automorphism of $Q_{16}$.  
The code space $C_{16}$ can be partionned into three sets 
$$
C_{16} = (1, \dots, 1) \ \sqcup \ 0 \oplus C_{15}^{\even} \ \sqcup \ 1 \oplus C_{15}^{\odd}
$$
and we the transformation $\sigma$ preserves each of these three sets.  
Indeed, $\sigma$ maps the vector $(1, \dots, 1)$ onto itself.  
Moreover, by definition the component $\bar \sigma$ preserves $C_{15}$ and like any permutation it also preserves the weight.  
Therefore, it leaves $C_{15}^{\even}$ and its complement $C_{15}^{\odd}$ invariant.  
This proves that $\sigma$ is an automorphism of $C_{16}$ whose image under $\pi$ is $\bar \sigma$.  
The map $\pi: {\Aut(Q_{16})}_0 \rightarrow \Aut(Q_{15})$ is a group isomorphism. 

\medskip
Overall, we have proved the isomorphism 
$$
\Aut(Q_{15}) \simeq \Aut(Q_{16})_0 
$$ 
and as a result the cardinality of this group is $\abs{ \Aut(Q_{15}) } = \abs{ \Aut(Q_{16})} / 16$.  

\medskip
The same argument applies for any pair of codes obtained by puncturing the extended Hamming code.  
Using standard properties of these codes, the proof takes only a few lines.  
Consider an extended Hamming code with parameters $[2^m, 2^m-m, 4]$.  
Puncturing a coordinates leads to the Hamming code $[2^m-1, 2^m-m-1, 3]$.  
The automorphism group of theses two codes are known to be respectively the affine group $\AGL_m(\Z_2)$ and the general linear group $\GL_m(\Z_2)$ acting on $\Z_2^m$.  
The automorphism group of Hamming code is therefore a normal subgroup of the automorphism group of the extended Hamming and its index is $2^m$.
The previous case corresponds to $m = 4$.

\section{\texorpdfstring{$\llbracket 15,7,3 \rrbracket$}{[[15,7,3]]} code: Measuring all even-weight logical $Z$ operators, for universality} \label{s:1573measure6evenweightopsandec}

In order to achieve fault-tolerant universal computation with the $\llbracket 15,7,3 \rrbracket$ code, it can be useful to measure all the even-weight logical $Z$ operators, which make up six of the seven encoded qubits.  
They have weight-four generators.  
With these six qubits initialized to encoded $\ket 0$, the gate $\big(\begin{smallmatrix}1&0\\0&e^{i \pi / 4}\end{smallmatrix}\big)$ can be applied tranversally to act on the last encoded qubit~\cite{PaetznickReichardt13universal, CampbellTerhalVuillot17universalityreview}.  

Steane's method to measure six encoded qubits is to apply transversal CNOT gates into, then measure, an encoded $\ket{0^6 {+}}$ state.  Instead, measuring one operator at a time, {\`a} la Shor, from \figref{f:1573measurementsforlogicalmeasurement} $6 \cdot 3 = 18$ measurements suffice,
or $22$ measurements with $X$ error correction.  

By measuring the logical operators together instead of sequentially, one can do better.  We have verified that the following sequence of $14$ measurements suffices to measure all even-weight logical $Z$ operators, with error correction: 
\begin{equation*}
\includegraphics[scale=.6]{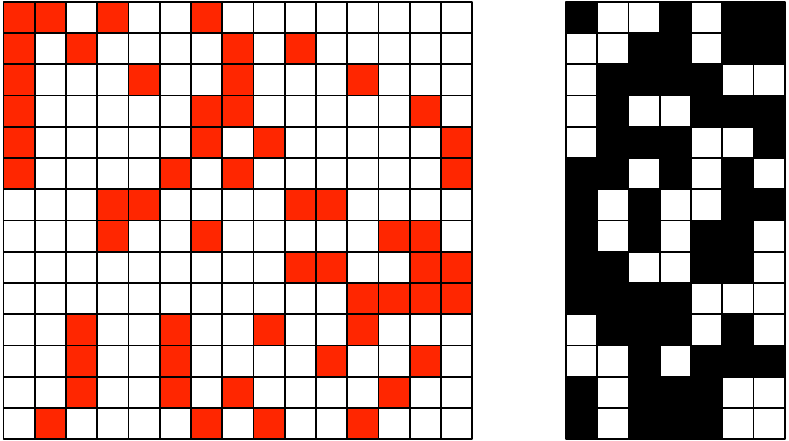}
\end{equation*}
Each $Z$ operator on the left is a combination, specified on the right, of the seven logical operators of \figref{f:hadamardcodelogicalbases1573}.  
This is similar to, if slightly more complicated than, the sequence for measuring two logical $Z$ operators in \secref{s:furtherlogicalmeasurementproblems}.

\section{Other measurement models} \label{s:othermeasurementmodels}

In this paper, we have studied in detail fault-tolerant error correction by sequentially, and either nonadaptively or adaptively, measuring stabilizers fault tolerantly using cat states, Shor-style.  Of course, this is not the only technique for fault-tolerant error correction.  For example, Knill-style error correction works essentially by teleporting an encoded state through an encoded Bell state~\cite{Knill03erasure}.  
Steane-style error correction, for CSS codes, uses transversal CNOT gates to/from encoded $\ket{+}$/$\ket{0}$ states~\cite{Steane97}.  The advantage of these methods is that they extract multiple syndrome bits in parallel; but the disadvantage is that the required encoded ancilla states are more difficult to prepare fault tolerantly than cat states, and need more qubits.  

However, there is room for variation even staying closer to the Shor-style error-correction framework, using cat states to measure single syndrome bits.  We consider four variants: syndrome bit extraction with flags to catch internal faults, partial parallel and parallel syndrome extraction, and nonadaptive flagged fault-tolerant syndrome bit extraction.  We demonstrate each technique on the $\llbracket 7,1,3 \rrbracket$ code.

\subsection{Syndrome bit extraction with flagged qubits}

Recall from Eq.~\eqnref{e:713codefivemeasurementerrorcorrectionsequence} that for the $\llbracket 7,1,3 \rrbracket$ Steane code, fault-tolerant $X$ error correction can be accomplished by measuring a fixed sequence of five $Z$ stabilizers.  
Consider instead measuring the following sequence of four stabilizers: 
\begin{equation} \label{e:713codefourmeasurementerrorcorrectionsequencewithbadinternalfault}
\definecolor{orange}{rgb}{1,0.5,0}
\place{\color{orange} \footnotesize X}{24mu}{28pt}\
\place{\color{red} \footnotesize X}{166mu}{0pt}\
\begin{array}{c c c c c c c}
I&I&I&Z&Z&Z&Z\\
I&Z&Z&I&I&Z&Z\\
Z&I&Z&I&Z&I&Z\\
I&Z&Z&Z&Z&I&I
\end{array}
\end{equation}
This is not enough for fault-tolerant error correction.  As indicated in red, an internal error on qubit~$7$ after the second stabilizer measurement generates the syndrome $0010$, which is confused with an input error on qubit~$1$ (indicated in orange).  This is the only bad internal error, however.  

One way to fix this problem is to place a ``flag" on qubit~$7$, as shown in \figref{f:713syndromeextractionwithflags}.  By temporarily coupling qubit~$7$ to another ancilla qubit, we ensure that if between the second and third stabilizer measurements an $X$ fault occurs on qubit~$7$, it will be detected.  Therefore this internal fault can be distinguished from an input error on qubit~$1$.  

\begin{figure}
\centering
\includegraphics[scale=.769]{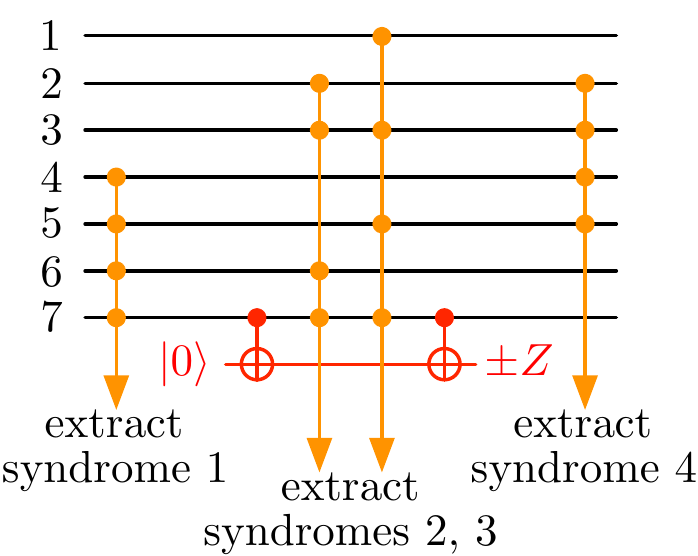}
\caption{By placing a temporary ``flag" on qubit~$7$ to catch the bad internal $X$ fault, fault-tolerant $X$ error correction for the $\llbracket 7,1,3 \rrbracket$ code can be implemented with four nonadaptive stabilizer measurements, compared to five measurements in Eq.~\eqnref{e:713codefivemeasurementerrorcorrectionsequence}.  Since this method requires one more ancilla qubit, it trades space for time.} \label{f:713syndromeextractionwithflags}
\end{figure}

This technique of adding flags to catch internal faults requires more qubits available for error correction; it trades space for time.  It easily extends to other codes.  First find all of the bad internal faults, then put flags around them.  (It is simplest to use separate flags for all code qubits that need them.  Using the same flag on multiple code qubits is not directly fault tolerant, because then a $Z$ fault on the flag could spread back to more than one code qubit.)

\subsection{Partial parallel syndrome extraction}

The bad internal fault in Eq.~\eqnref{e:713codefourmeasurementerrorcorrectionsequencewithbadinternalfault} can also be fixed by switching qubit $7$'s interactions with the cat states measuring the second and third stabilizers, as shown in \figref{f:713partialreordering}.  Then the possible syndromes from an internal fault on qubit~$7$ are $0100$ and $0110$, which are both okay.  
Again this technique trades space for time.  

\begin{figure}
\centering
\includegraphics[scale=.27]{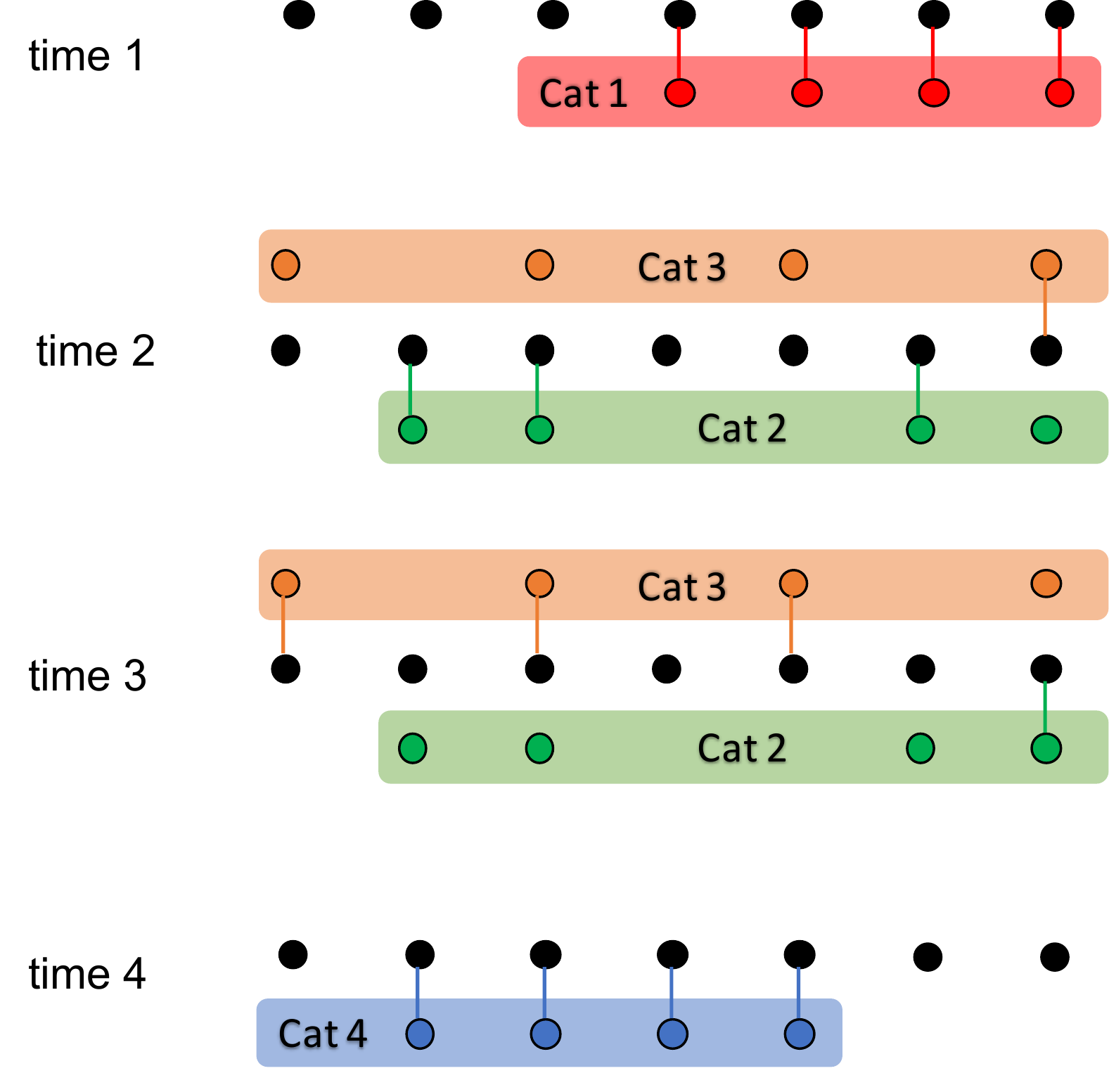}
\caption{
Another way of implementing $\llbracket 7,1,3 \rrbracket$ code fault-tolerant $X$ error correction with four nonadaptive stabilizer measurements is to extract some syndrome bits in parallel, partially reordered.  
} \label{f:713partialreordering}
\end{figure}

\subsection{Parallel syndrome extraction}

Alternatively, the bad internal fault in Eq.~\eqnref{e:713codefourmeasurementerrorcorrectionsequencewithbadinternalfault} can be avoided entirely by measuring the second and third stabilizers simultaneously, using a fault-tolerantly prepared six-qubit ancilla state, stabilized by 
\begin{equation*}
\begin{array}{c c c c c c}
I&I&Z&Z&Z&Z\\
Z&Z&Z&Z&I&I\\[.1cm]
X&X&I&I&I&I\\
I&I&X&X&I&I\\
I&I&I&I&X&X\\
X&I&X&I&X&I
\end{array}
\end{equation*}
Since both syndrome bits are extracted simultaneously, an $X$ fault on the data can flip both or neither, but unlike in~\eqnref{e:713codefourmeasurementerrorcorrectionsequencewithbadinternalfault} cannot go between them.  

With the $\llbracket 7,1,3 \rrbracket$ code, all three $Z$ stabilizers can be simultaneously measured, Steane style, using a seven-qubit encoded $\ket +$ state.  Measuring two at a time might be more useful for larger codes.

\subsection{Nonadaptive flagged fault-tolerant \\ syndrome bit extraction}

Naively, measuring a weight-$w$ stabilizer fault tolerantly requires a $w$-qubit cat state that has been prepared fault tolerantly.  However, this is not necessarily the case.  Methods of using cat states more efficiently have been developed by DiVincenzo and Aliferis~\cite{DiVincenzoAliferis06slow}, and by Stephens~\cite{Stephens14colorcodeft} and Yoder and Kim~\cite{YoderKim16trianglecodes}---techniques generalized in~\cite{ChaoReichardt18flags}.  

Flag fault tolerance~\cite{ChaoReichardt18flags, Reichardt18steane, ChamberlandBeverland17flags, TansuwannontChamberlandLeung18flag} is a technique that for certain codes uses just two ancilla qubits to measure a weight-$w$ stabilizer.  
In the simplest form of flag fault tolerance, a syndrome bit is extracted all onto a single qubit, while an extra ``flag" qubit is used to detect faults that could spread backwards into correlated data errors.  For example, \figref{f:fourqubitsyndromeflagged} shows a flagged circuit for measuring the syndrome bit of a weight-four $Z$ stabilizer.  A single $Z$ fault can spread to a weight-two data error, but then will also be detected by the $X$ basis measurement of the flag qubit, initialized as~$\ket +$.  
For a distance-three CSS code, when the flag is triggered the possible $Z$ errors spread back to the data are $IIII$, $ZIII$, $IIZZ$ and $IIIZ$.  
The error-correction schemes given in~\cite{ChaoReichardt18flags} are adaptive; given that the flag was triggered, additional $X$ stabilizer measurements are made to distinguish these four possibilities.  

\begin{figure}
\centering
\includegraphics[scale=.769]{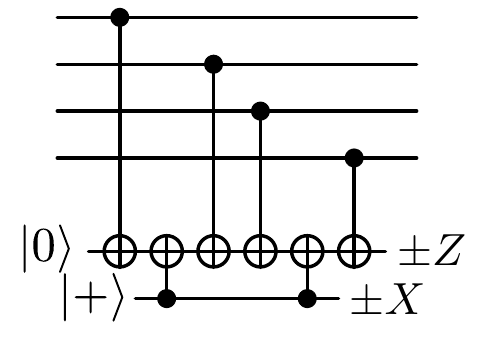}
\caption{A flagged circuit for measuring a weight-four $Z$ stabilizer.  A $Z$ fault on the syndrome qubit, initialized to~$\ket 0$, can spread to a weight-two data error, but then will also be detected by the $X$ basis measurement of the flag qubit.} \label{f:fourqubitsyndromeflagged}
\end{figure}

However, flag fault-tolerant error correction can also be nonadaptive.  For example, for the $\llbracket 7,1,3 \rrbracket$ code, consider the following sequence of ten stabilizer measurements: 
\begin{equation} \label{e:713tenstabilizersflagged}
\place{\footnotesize \text{measure}}{270mu}{22pt}
\place{\footnotesize \text{with flags}}{270mu}{12pt}
\place{\footnotesize \text{measure with}}{270mu}{-50pt}
\place{\footnotesize \text{cat states}}{270mu}{-60pt}
\place{$\left.\begin{array}{c} \\[3.25cm] \end{array}\right\}$}{210mu}{15pt}
\place{$\left.\begin{array}{c} \\[.3cm] \end{array}\right\}$}{210mu}{-57pt}
\place{\footnotesize 1}{116mu}{24pt}\place{\footnotesize 3}{144mu}{24pt}\place{\footnotesize 2}{171mu}{24pt}\place{\footnotesize 4}{198mu}{24pt}
\place{\footnotesize 1}{61mu}{10pt}\place{\footnotesize 3}{89mu}{10pt}\place{\footnotesize 2}{171mu}{10pt}\place{\footnotesize 4}{198mu}{10pt}
\place{\footnotesize 1}{34mu}{-4pt}\place{\footnotesize 2}{89mu}{-4pt}\place{\footnotesize 3}{144mu}{-4pt}\place{\footnotesize 4}{198mu}{-4pt}
\place{\footnotesize 1}{61mu}{10pt}\place{\footnotesize 3}{89mu}{10pt}\place{\footnotesize 2}{171mu}{10pt}\place{\footnotesize 4}{198mu}{10pt}
\place{\footnotesize 1}{116mu}{-17pt}\place{\footnotesize 3}{144mu}{-17pt}\place{\footnotesize 2}{171mu}{-17pt}\place{\footnotesize 4}{198mu}{-17pt}
\place{\footnotesize 1}{61mu}{-31pt}\place{\footnotesize 3}{89mu}{-31pt}\place{\footnotesize 2}{171mu}{-31pt}\place{\footnotesize 4}{198mu}{-31pt}
\begin{array}{r c c c c c c c}
&I&I&I&Z&Z&Z&Z\\
&I&Z&Z&I&I&Z&Z\\
&Z&I&Z&I&Z&I&Z\\[.1cm]
&I&I&I&X&X&X&X\\
&I&X&X&I&I&X&X\\
&X&I&X&I&X&I&X\\
&I&I&I&X&X&X&X\\
&I&X&X&I&I&X&X\\[.1cm]
&I&I&I&Z&Z&Z&Z\\
&I&Z&Z&I&I&Z&Z
\end{array}
\qquad\qquad
\end{equation}
The first three $Z$ stabilizer measurements can all be made using flags, because they are followed by a full round of $X$ error correction.  In fact, though, the five $X$ stabilizer measurements can also be made using flags, provided that the interactions are made in the specified order, because the final two $Z$ measurements are enough to diagnose the data error when a flag is triggered.  (For example, should either of the $IIIXXXX$ measurements be flagged, the possible errors $X_4, X_7$ and $X_5 X_7$ are correctable using the final two $Z$ measurements.  Should the $XIXIXIX$ measurement be flagged, the possible $X_1$ error is not detected, but this is okay for fault tolerance.)  However, the last two $Z$ measurements cannot be made using flags, because if a flag were triggered there would be no subsequent $X$ measurements to diagnose the error.    

It is important to develop error-correction schemes, nonadaptive or adaptive, that are both fast---requiring few rounds of interaction with the data---and efficient in the sense of using simple cat states or other efficiently prepared ancilla states.  Combining flag fault tolerance with standard Shor-style syndrome extraction, as in Eq.~\eqnref{e:713tenstabilizersflagged}, is a step in this direction, although its effectiveness will depend on implementation details such as geometric locality constraints.


\newcommand{\etalchar}[1]{$^{#1}$}

\end{document}
\fi